%% file: main.tex
\documentclass[10pt,a4paper,notitlepage]{article}
\usepackage[english]{babel}
\usepackage{graphicx}
\usepackage{subcaption}
\usepackage{psfrag}

\usepackage{enumitem}
\usepackage{framed}
\usepackage[normalem]{ulem}
\usepackage{amsmath}
\usepackage{amsthm}
\usepackage{amssymb}
\usepackage{amsfonts}
\usepackage{todonotes}
\usepackage{enumerate}
\usepackage[utf8]{inputenc}
\usepackage[dvipsnames]{xcolor}
\definecolor{RWTHBlau}{RGB}{0,84,159}
\definecolor{RWTHBordeaux}{RGB}{161,16,53}
\definecolor{RWTHGruen}{RGB}{87,171,39}
\definecolor{RWTHOrange}{RGB}{246,168,0}
\definecolor{RWTHTuerkis}{RGB}{0,152,161}
\usepackage{hyperref}
\hypersetup{
  colorlinks,
  citecolor=RWTHBlau,
  linkcolor=RWTHBlau,
  urlcolor=RWTHBlau}
\usepackage{geometry}
\usepackage[square, comma, sort&compress]{natbib}
\usepackage[parfill]{parskip}
\usepackage{algorithm}
\usepackage{algorithmic}
\usepackage{float}
\usepackage{overpic}
\usepackage{times}
\usepackage{tabularray}
\usepackage{tabularx}
\usepackage{multirow}
\usepackage{booktabs}
\usepackage{nicematrix}

\theoremstyle{definition}

\newtheorem{proposition}{Proposition}

\let\vec\mathbf
\newcommand{\tr}{\mathrm{tr}}
\newcommand{\dev}{\mathrm{dev}}
\newcommand{\sym}{\mathrm{Sym}}

\colorlet{author}{black}
\colorlet{author2}{black}

\usepackage{standalone}
\usepackage{tikz}
\usetikzlibrary{shapes.geometric, arrows.meta, positioning}
\usepackage{etoolbox}
\makeatletter
\patchcmd{\maketitle}{\vskip 4em}{\vskip 2em}{}{}
\makeatother

\tikzstyle{block} = [rectangle, draw, text centered, minimum height=2em]
\tikzstyle{arrow} = [thick, ->, >=Stealth]

\date{}

\begin{document}

\pgfkeys{/pgf/number format/.cd,1000 sep={}}
\title{\LARGE Inelastic Constitutive Kolmogorov-Arnold Networks: \\A generalized framework for automated discovery of interpretable inelastic material models}

\author{\large{$\text{Chenyi Ji}^{\mathrm{a}}$, $\text{Kian P. Abdolazizi}^{\mathrm{b}}$, $\text{Hagen Holthusen}^{\mathrm{c}}$, $\text{Christian J. Cyron}^{\mathrm{b,d}}$, Kevin Linka$^{\mathrm{a},*}$ }\\[0.2cm]
  \hspace*{-0.1cm}
  \small{\em ${{}^\mathrm{a}}$ Computational Mechanics in Medicine, Applied Medical Engineering, RWTH Aachen University}\\
  \small{\em Pauwelsstraße 20, 52074 Aachen, Germany}\\
  [0.1cm]
  \small{\em ${{}^\mathrm{b}}$ Institute for Continuum and Material Mechanics, Hamburg University of Technology}\\
  \small{\em Eißendorfer Straße 42, 21073 Hamburg, Germany}\\
  [0.1cm]
  \small{\em ${{}^\mathrm{c}}$ Institute of Applied Mechanics, University of Erlangen-Nuremberg}\\
  \small{\em Egerlandstrasse 5, 91058 Erlangen, Germany}\\
  [0.1cm]
  \small{\em ${{}^\mathrm{d}}$ Institute of Material Systems Modeling, Helmholtz-Zentrum Hereon}\\
  \small{\em Max-Planck-Straße 1, 21502 Geesthacht, Germany}\\
}

\vspace{-4cm}
\maketitle
\vspace*{1cm}


\vspace{-.5cm}
\small
\begin{center}
\begin{minipage}{0.9\textwidth}
{\bf Abstract.}
A key problem of solid mechanics is the identification of the constitutive law of a material, that is, the relation between \textcolor{author}{strain history} and stress. Machine learning has led to considerable advances in this field lately. 
Here we introduce inelastic Constitutive Kolmogorov–Arnold Networks (iCKANs). This novel artificial neural network architecture can discover in an automated manner symbolic constitutive laws describing both the elastic and inelastic behavior of materials. That is, it can translate data from material testing into corresponding free energy and dissipation potential functions in closed mathematical form.
We demonstrate the advantages of iCKANs using both synthetic data and experimental data of the viscoelastic polymer materials VHB 4910 and VHB 4905. The results demonstrate that iCKANs accurately capture complex viscoelastic behavior while preserving physical interpretability. It is a particular strength of iCKANs that they can process not only mechanical data but also arbitrary additional information available about a material (e.g., about temperature-dependent behavior). This makes iCKANs a powerful tool to discover in the future also how specific processing or service conditions affect the properties of materials. 

\vspace*{0.2cm}

\end{minipage}
\end{center}

%

\vspace{1cm}

\normalsize

\input{AA_section_Introduction.tex}
\input{AA_section_Constitutive.tex}
\input{AA_section_ICKAN.tex}

\input{AA_section_Symbolic.tex}
\input{AA_section_Results.tex}
\input{AA_section_Conclusion.tex}



\addtocontents{toc}{\vspace{2em}} 

\appendix 
\numberwithin{equation}{section}
\counterwithin{figure}{section}
\counterwithin{table}{section}

\input{Appendix_A}
\input{Appendix_B}

\addtocontents{toc}{\vspace{2em}} 


\bibliography{bibliography}  
\end{document}

%% file: AA_section_Introduction.tex
\section{Introduction}

Accurately predicting the behavior of complex material systems is a core challenge in engineering and the physical sciences. Conventional computational mechanics relies on constitutive models to describe how materials respond to external stimuli. However, traditional constitutive models, grounded in continuum mechanics, depend on simplifying assumptions that often fail to capture the full complexity of real-world material behavior—particularly for inelastic materials. Moreover, developing these models is typically incremental, labor-intensive, and demands substantial domain expertise, with the process repeated for each novel material system.

To overcome these limitations, data-driven material modeling has emerged as a promising alternative, bypassing explicit constitutive model formulation and leveraging experimental data directly \citep{fuhg2024review}. The model-free or direct data-driven paradigm, pioneered by \citet{kirchdoerfer2016data}, eliminates predefined constitutive equations in favor of using raw material measurements. Another avenue is symbolic regression, which constructs closed-form mathematical models from data by systematically combining analytical expressions to achieve a balance between accuracy and interpretability \citep{versino2017data, bomarito2021development, kabliman2021application, abdusalamov2023automatic}. Expanding on this, EUCLID (Efficient Unsupervised Constitutive Law Identification and Discovery) applies sparse regression to a broad model library, enabling learning from full-field data \citep{flaschel2021unsupervised, flaschel2023automated, flaschel2023automated2}. Recent developments have further enhanced these approaches to better handle noisy experimental data \citep{narouie2026unsupervised}.

A distinct class of data-driven approaches involves constitutive neural networks, which have gained prominence due to their flexibility and capacity for universal function approximation \citep{hornik1989multilayer}. Unlike methods that impose physical constraints via the loss function \citep{raissi2019physics, masi2021thermodynamics}, these networks encode physical consistency directly through architectural design. Notable variants include Constitutive Artificial Neural Networks (CANNs) \citep{linka2021constitutive, linka2023new}, which favor interpretable, sparse architectures, and Physics-Augmented Neural Networks (PANNs) \citep{rosenkranz2024viscoelasticty, linden2023neural}, which leverage denser, more expressive structures. Additional neural network classes have further broadened the modeling landscape: neural ordinary differential equations (ODEs) have been used to identify polyconvex strain energy functions \citep{tac2022data}; graph neural networks \citep{maurizi2022predicting} and neural operators \citep{you2022physics} have been explored for learning complex constitutive and surrogate models. These methods have also found applications in the design and optimization of metamaterials \citep{fernandez2021anisotropic,fernandez2022material}. Current research pushes the boundaries by targeting increasingly complex phenomena, such as anisotropy and inelasticity \citep{holthusen2026complement}.

To capture inelastic effects, the Generalized Standard Materials framework offers a broad and physically grounded foundation \citep{halphen1975materiaux}. This approach introduces an inelastic dissipation potential alongside the elastic free energy, governing the evolution of internal variables and inelastic deformation in a thermodynamically consistent manner. 
\color{author}
An early step toward integrating this thermodynamic constitutive framework with neural networks was taken by \citet{huang2022variational}, who proposed a variational neural-network formulation for material modeling. Integrating this framework with constitutive neural networks has enabled the discovery of material laws for finite-strain viscoelasticity \citep{kalina2026physics,as2023mechanics,holthusen2025generalized,tacc2023data,wiesheier2026data}, plasticity \citep{flaschel2025convex,boes2024accounting,jadoon2025automated}, fracture \citep{dammass2025neural}, and biological processes like growth and remodeling \citep{holthusen2025automated}. In a related vein, viscoelastic Constitutive Artificial Neural Networks (vCANNs) learn anisotropic nonlinear finite-strain viscoelasticity directly from time-series data \citep{abdolazizi2024viscoelastic} and have subsequently been deployed automatically as finite element material subroutines \citep{abdolazizi2026thermodynamically}.
\color{black}
Constitutive models based on internal variables have also been successfully employed in neural ODEs and neural operators to model history- and path-dependent behaviors \citep{jones2022neural,guo2025history}. Furthermore, approaches such as long short-term memory networks have been introduced to enhance computational efficiency in plasticity modeling \citep{li2025long}.

Constitutive Kolmogorov–Arnold Networks (CKANs) \citep{abdolazizi2025constitutive,thakolkaran2025can} offer a distinctive alternative by employing B-Splines as nonlinear, trainable activation functions whose mathematical forms can be extracted for interpretation \citep{liu2024kan}. The introduction of Kolmogorov–Arnold Networks (KANs) has prompted comparative studies with traditional neural networks across a range of applications. KANs have shown improved accuracy and convergence in solving partial differential equations relative to multilayer perceptrons \citep{wang2025kolmogorov, abueidda2025deepokan, kiyani2025optimizing}, with their greatest strengths evident in symbolic regression tasks \citep{ji2024comprehensive}. In the context of material modeling, KANs have been incorporated into CKANs and Input-Convex KANs to enhance interpretability and facilitate the discovery of convex strain energy functions, all while retaining the adaptability of neural network-based frameworks \citep{abdolazizi2025constitutive, thakolkaran2025can}. These developments underscore the potential of KAN-based architectures to increase the transparency of data-driven material models without sacrificing predictive accuracy, though challenges in modeling inelastic behavior persist.

In this work, we advance KAN-based constitutive modeling by extending it to inelastic materials and introducing inelastic Constitutive Kolmogorov–Arnold Networks (iCKANs). This novel approach integrates the generalized inelastic constitutive framework with the inherent interpretability of KANs, aiming to deliver a robust, data-driven, and transparent model for complex inelastic behavior. iCKANs leverage experimental stress–strain data, optionally enriched with non-mechanical features, to uncover the underlying elastic and dissipation potentials of materials. Crucially, the trainable activations of the KANs are subsequently symbolized, yielding interpretable closed-form expressions for these potentials. By enabling the automated derivation of symbolic dissipation potentials, previously accessible primarily to a limited group of experts, iCKANs address a longstanding challenge in finite strain inelasticity. The overall workflow of the iCKAN framework is illustrated in \autoref{fig:flowchart}.

\textbf{Outline.} Section \ref{sec:constitutive_framework} provides a concise review of the generalized constitutive framework for inelastic materials at finite strains. Section \ref{sec:iCKAN} details the proposed iCKANs' methodology, beginning with the foundations of Kolmogorov–Arnold Networks and their partially input-convex variants, followed by their application to constitutive modeling. Section \ref{sec:symbolic} addresses the symbolic extraction process to obtain interpretable, closed-form representations of the discovered models. Section \ref{sec:results} evaluates the performance of iCKANs on both synthetic and experimental viscoelastic datasets. The paper concludes in Section \ref{sec:conclusion} with a discussion of key findings and future research directions. 

\begin{figure}[H]
    \input{standalone_graphical_abstract}
    \caption{Inelastic Constitutive Kolmogorov-Arnold Network (iCKAN) pipeline for automated interpretable model discovery of inelastic materials. Three-dimensional stress-strain data with optional additional features (e.g., temperature) collected in a feature vector $\vec{f}$ are used to train the iCKAN model, which consists of two KANs representing the free energy $\psi$ and the dissipation potential $\omega$. 
    The trained model can then be analyzed using symbolic regression to extract interpretable mathematical expressions.}
    \label{fig:flowchart}
\end{figure}


%% file: standalone_graphical_abstract.tex
\centering
    \includegraphics[width=\linewidth]{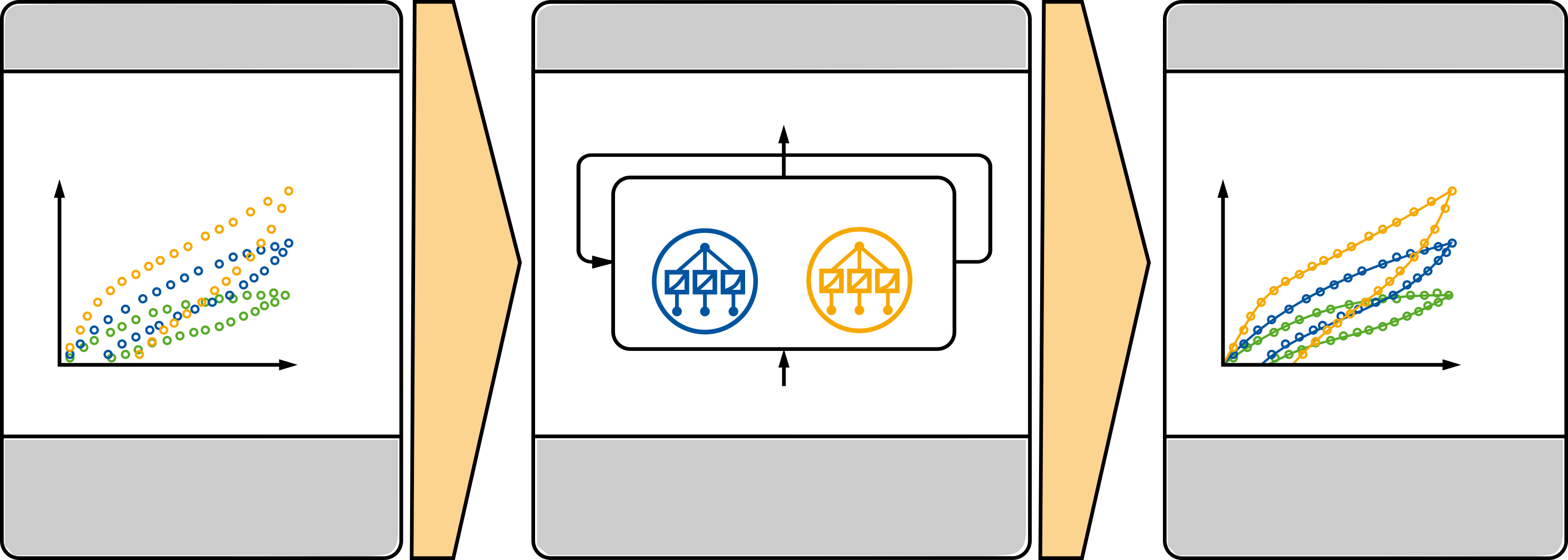}
    \put(-460,158){\textbf{Experimental data}}
    \put(-305,158){\textbf{Kolmogorov-Arnold Networks}}
    \put(-100,158){\textbf{Symbolic function}}
    \put(-472,15){\parbox{0.3\linewidth}{\footnotesize 
        $\bullet$ Time-dependent stress-\\
        \hspace*{.25cm}deformation data \\
        $\bullet$ Non-mechanical data $\vec{f}$
    }}
    \put(-387,112){\small$\vec{f} = \textcolor{RWTHOrange}{\vec{f}_1}$}
    \put(-387,95){\small$\vec{f} = \textcolor{RWTHBlau}{\vec{f}_2}$}
    \put(-387,78){\small$\vec{f} = \textcolor{RWTHGruen}{\vec{f}_3}$}
    \put(-442,47){\small Stretch $F_i$}
    \put(-478,67){\small \rotatebox{90}{Stress $P_i$}}
    \put(-350,30){\rotatebox{90}{Training with sparsification}}
    \put(-260,138){\small Stress $\vec{P}$}
    \put(-280,45){\small Deformation gradient $\vec{F}$}
    \put(-292,107){\small Free energy}
    \put(-247,107){\small Dissipation pot.}
    \put(-307,15){\parbox{0.3\linewidth}{\footnotesize 
        $\bullet$ Recurrent network architecture\\
        $\bullet$ Mathematical/physical constraints \\
        $\bullet$ (Partially) convex splines
    }}
    \put(-155,50){\rotatebox{90}{Symbolification}}
    \put(-100,135){\small Stress predictions}
    \put(-100,123){\small$P = a(F+b)^2$}
    \put(-30,112){\small$\vec{f} = \textcolor{RWTHOrange}{\vec{f}_1}$}
    \put(-30,95){\small$\vec{f} = \textcolor{RWTHBlau}{\vec{f}_2}$}
    \put(-30,78){\small$\vec{f} = \textcolor{RWTHGruen}{\vec{f}_3}$}
    \put(-85,47){\small Stretch $F_i$}
    \put(-120,67){\small \rotatebox{90}{Stress $P_i$}}
    \put(-114,11){\parbox{0.3\linewidth}{\footnotesize 
        $\bullet$ Interpretability of free energy\\ 
        \hspace*{.25cm}and dissipation potential\\
        $\bullet$ Efficient numerical simulations \\
    }}

%% file: AA_section_Constitutive.tex
\section{Constitutive modeling of inelastic materials}
\label{sec:constitutive_framework}

In this section, we briefly review the generalized constitutive framework for inelastic materials at finite strains by \citet{holthusen2023inelastic,holthusen2025generalized}, which can be seen as an equivalence to the framework of Generalized Standard Materials \citep{halphen1975materiaux}. The framework employs the multiplicative decomposition of the deformation gradient and is based on two fundamental scalar-valued quantities, the free energy $\psi$ and an dissipation potential $\omega$ \citep{holthusen2024theory}. 

\textbf{Multiplicative decomposition.} The general assumption of multiplicative decomposition is applied on the deformation gradient $\vec{F}$, resulting into an elastic part $\vec{F}_e$ and inelastic part $\vec{F}_i$, i.e., $\vec{F} = \vec{F}_e \vec{F}_i$. 
\color{author}
It is to note that neither $\vec{F}_e$ nor $\vec{F}_i$ can be derived from motion in general and should not be misunderstood as gradients. Conceptually, an intermediate configuration is introduced, which is reached by the inelastic part of the deformation gradient $\vec{F}_i$ from the reference configuration.
The elastic part of the deformation gradient $\vec{F}_e$ then maps the intermediate configuration to the current configuration. \color{black} 
However, the intermediate configuration is fictitious and not unique, since any rotation $\vec{Q}^\dag \in SO(3)$ applied to the intermediate configuration does not alter the physics of the resulting deformation gradient, i.e., $\vec{F} = \vec{F}_e {\vec{Q}^\dag}^\mathrm{T} {\vec{Q}^\dag} \vec{F}_i = \vec{F}_e^\dag\vec{F}_i^\dag$. Following \citet{holthusen2023inelastic,holthusen2024theory}, a unique co-rotated intermediate configuration is obtained by applying the rotation of the polar decomposition of $\vec{F}_i = \vec{R}_i \vec{U}_i$ to the intermediate configuration, i.e., $\bar{\vec{F}}_e = \vec{F}_e\vec{R}_i = \vec{F}\vec{U}_i^{-1}$, with $\vec{U}_i$ and $\vec{R}_i$ being the stretch part and rotation part of the inelastic part of the deformation gradient, respectively. 

\begin{figure}[H]
    \centering\small
    \includegraphics[width=0.564\textwidth]{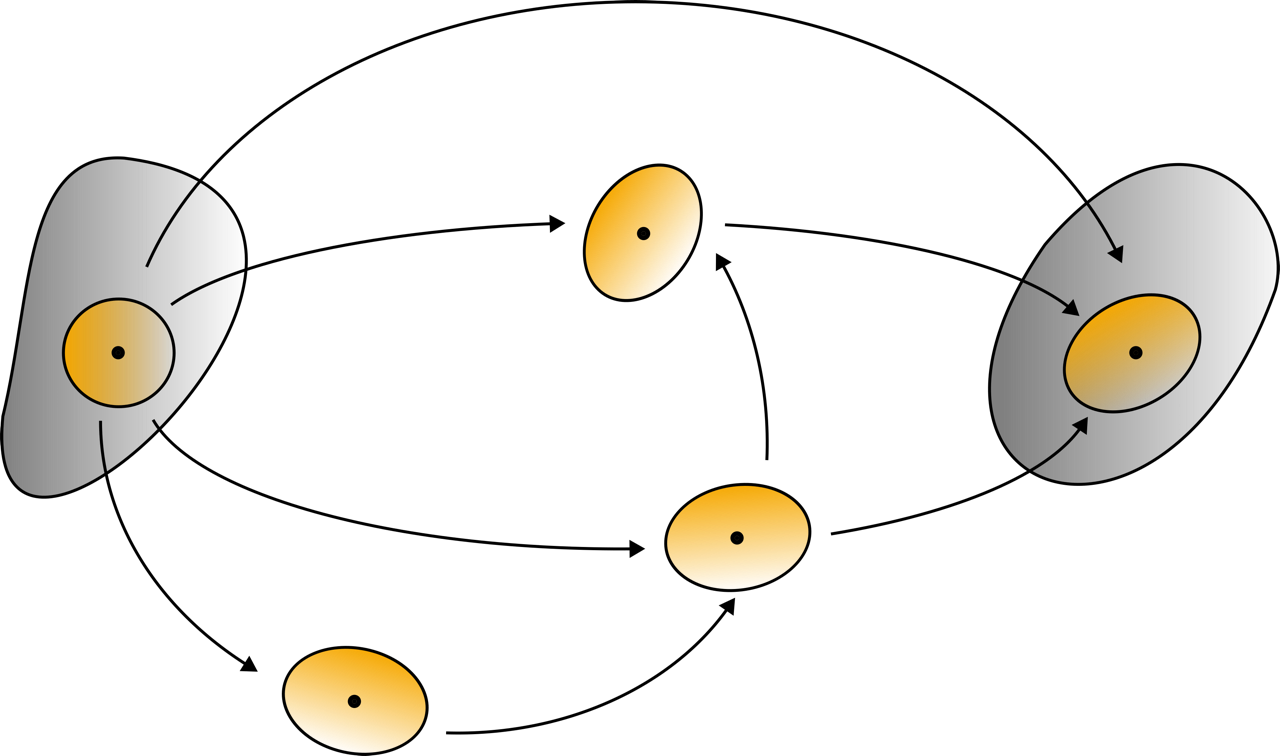}
    \put(-180,160){$\vec{F}$}
    \put(-195,115){$\vec{F}_i^\dag$}
    \put(-100,115){$\vec{F}_e^\dag$}
    \put(-195,55){$\vec{F}_i$}
    \put(-90,55){$\vec{F}_e$}
    \put(-130,80){$\vec{Q}^\dag$}
    \put(-245,20){${\vec{U}_i}$}
    \put(-135,5){${\vec{R}_i}$}
    \put(-330,90){reference}
    \put(-330,80){configuration}
    \put(-105,30){intermediate configuration}
    \put(-5,80){current configuration}
    \put(-180,143){arbitrarily rotated}
    \put(-180,133){intermediate configuration}
    \put(-250,-3){co-rotated}
    \put(-250,-13){intermediate configuration}
    \put(-258, 110){$\mathcal{B}_0$}
    \put(-25, 110){$\mathcal{B}_t$}
    \caption{Multiplicative split of the deformation gradient into elastic and inelastic part, including the non-uniqueness of the intermediate configuration and the co-rotated intermediate configuration \citep{holthusen2023inelastic}.}
    \label{fig:multiplicative_split}
\end{figure}

The right Cauchy-Green deformation tensors are defined as $\vec{C} = \vec{F}^\mathrm{T} \vec{F}$ and $\bar{\vec{C}}_e = \vec{U}_i^{-1} \vec{C} \vec{U}_i^{-1}$ in the current and co-rotated intermediate configuration, respectively. Due to the fact that $\bar{\vec{C}}_e = \vec{U}_i^{-1} \vec{C} \vec{U}_i^{-1}$ serves as an unique measurement for elastic stretches, the free energy can be defined depending on it, i.e., $\psi = \psi(\bar{\vec{C}}_e)$, while preserving the principles of objectivity and frame indifference. To ensure that the free energy is a scalar-valued isotropic function, we define it based on the principal invariants of $\bar{\vec{C}}_e$, i.e., $\psi = \psi(I_1^{\bar{\vec{C}}_e}, I_2^{\bar{\vec{C}}_e}, I_3^{\bar{\vec{C}}_e})$, with
\begin{equation}
    I_1^{\bar{\vec{C}}_e} = \tr(\bar{\vec{C}}_e), \quad I_2^{\bar{\vec{C}}_e} = \frac{1}{2} \left[ (\tr(\bar{\vec{C}}_e))^2 - \tr(\bar{\vec{C}}_e^2) \right], \quad I_3^{\bar{\vec{C}}_e} = \det(\bar{\vec{C}}_e)\,.
    \label{eq:invariants_Ce}
\end{equation}

\color{author}
\textbf{Dissipation inequality.} Following the second law of thermodynamics, the dissipation $\mathcal{D}$ has to be non-negative so that thermodynamic consistency is preserved. This can be expressed using the Clausius-Plank inequality $\mathcal{D} = \vec{S} : \frac{1}{2} \dot{\vec{C}} - \dot{\psi} \geq 0$, where $\vec{S}$ denotes the second Piola-Kirchhoff stress tensor. Here, we assume the functional dependicies of
$\psi(\vec{C}, \vec{U}_i)$, $\vec{S}(\vec{C}, \vec{U}_i)$ and \textcolor{author2}{$\dot{\vec{U}}_i = \vec{g}(\vec{C}, \vec{U}_i)$}, where $\vec{g}$ is a tensor-valued function to be specified later. By inserting these dependies and applying the chain rule, we obtain
\begin{equation}
  \begin{split}
    \mathcal{D} = \left(\vec{S} - 2 \,\vec{U}_i^{-1} \frac{\partial \psi}{\partial \bar{\vec{C}}_e} \vec{U}_i^{-1}\right): \frac{1}{2} \dot{\vec{C}} + 2\, \bar{\vec{C}}_e \frac{\partial \psi}{\partial \bar{\vec{C}}_e} : \dot{\vec{U}}_i \vec{U}_i^{-1} \geq 0 \\
    \textcolor{author2}{\forall \, ( \vec{C}, \dot{\vec{C}}, \vec{U}_i)\in \sym^+(3) \times \sym(3) \times \sym^+(3) \,,}
  \end{split}
\end{equation}
where $\sym(3)$ and $\sym^+(3)$ denote the spaces of symmetric and symmetric positive-definite tensors, respectively. This equation has to be fulfilled for any arbitrary value of $\dot{\vec{C}}$. Hence, the second Piola-Kirchhoff stress tensor $\vec{S}$ can be defined,
\begin{equation}
    \vec{S} = 2\, \vec{U}_i^{-1} \frac{\partial \psi}{\partial \bar{\vec{C}}_e} \vec{U}_i^{-1}\,.
\end{equation}
Additionally, we define the symmetric elastic Mandel-like stress tensor 
\begin{equation}
    \bar{\vec{\Sigma}} = 2\, \bar{\vec{C}}_e \frac{\partial \psi}{\partial \bar{\vec{C}}_e}\,.
\end{equation}
The first Piola-Kirchhoff stress follows from the standard relation in classical continuum mechanics, i.e., $\vec{P} = \vec{F} \vec{S}$. Moreover, by noting that $\bar{\vec{\Sigma}}$ is symmetric, the dissipation inequality can be reduced to
\begin{equation}
    \mathcal{D} = \bar{\vec{\Sigma}} : \dot{\vec{U}}_i \vec{U}_i^{-1} = \bar{\vec{\Sigma}} : \bar{\vec{D}}_i \geq 0 \quad \text{with} \quad \bar{\vec{D}}_i = \mathrm{sym}(\dot{\vec{U}}_i \vec{U}_i^{-1})\quad
   \textcolor{author2}{\forall \, ( \vec{C}, \vec{U}_i)\in \sym^+(3) \times \sym^+(3) \,.}
\end{equation}
\color{black}

\textbf{Evolution equation.}
At this stage, a suitable evolution equation for the inelastic rate tensor
$\bar{\vec{D}}_i$ remains to be specified. To this end, a dual dissipation
potential $\omega$ is introduced, which is assumed to be a scalar-valued,
isotropic function of the thermodynamically consistent driving force
$\bar{\vec{\Sigma}}$, i.e.,
$\omega = \omega(\bar{\vec{\Sigma}})$.
It has been shown that if $\omega$ is constructed to be convex, zero-valued
at $\bar{\vec{\Sigma}}=\vec{0}$, and non-negative, the reduced dissipation
inequality is satisfied automatically \citep{germain1983continuum}.
While these conditions are \emph{sufficient}, they are not \emph{necessary}
to guarantee non-negative dissipation.
Following \citep{holthusen2025generalized}, a more general formulation is
obtained by introducing an invariant-based dissipation potential of the form
\begin{equation}
  \omega(\bar{\vec{\Sigma}})
  =
  \omega\!\left(
    I_1^{\bar{\vec{\Sigma}}},
    \sqrt{J_2^{\bar{\vec{\Sigma}}}},
    \sqrt[3]{J_3^{\bar{\vec{\Sigma}}}}
  \right),
\end{equation}
which is required to be convex, zero-valued at its origin, and non-negative with respect
to the three invariants\footnote{It is worth noting that the third invariant is generally not convex with respect to $\bar{\vec{\Sigma}}$.}
\begin{equation}
\begin{alignedat}{3}
  I_1^{\bar{\vec{\Sigma}}} &= \tr(\bar{\vec{\Sigma}}),
  &\quad
  J_2^{\bar{\vec{\Sigma}}} &=
  \tfrac{1}{2}\tr\!\bigl(\dev(\bar{\vec{\Sigma}})^2\bigr),
  &\quad
  J_3^{\bar{\vec{\Sigma}}} &=
  \tfrac{1}{3}\tr\!\bigl(\dev(\bar{\vec{\Sigma}})^3\bigr),
  \\[0.5ex]
  \frac{\partial I_1^{\bar{\vec{\Sigma}}}}{\partial \bar{\vec{\Sigma}}}
  &= \vec{I},
  &
  \frac{\partial J_2^{\bar{\vec{\Sigma}}}}{\partial \bar{\vec{\Sigma}}}
  &= \dev(\bar{\vec{\Sigma}}),
  &
  \frac{\partial J_3^{\bar{\vec{\Sigma}}}}{\partial \bar{\vec{\Sigma}}}
  &= \dev\!\bigl(\dev(\bar{\vec{\Sigma}})^2\bigr).
\end{alignedat}
\label{eq:invariants_Sig}
\end{equation}
Here,
$\dev(\bar{\vec{\Sigma}}) = \bar{\vec{\Sigma}} - \tfrac{1}{3} I_1^{\bar{\vec{\Sigma}}} \vec{I}$
denotes the deviatoric part of a second-order tensor.
It is worth emphasizing that the square and cubic roots of the second and
third invariants, respectively, are essential to ensure non-negative
dissipation within this framework \citep{holthusen2025generalized}.
The evolution equation for the inelastic rate tensor then follows directly
from the dissipation potential as
\begin{equation}
  \bar{\vec{D}}_i
  =
  \frac{\partial \omega(\bar{\vec{\Sigma}})}{\partial \bar{\vec{\Sigma}}}
  =
  \frac{\partial \omega}{\partial I_1^{\bar{\vec{\Sigma}}}}\, \vec{I}
  +
  \frac{\partial \omega}{\partial J_2^{\bar{\vec{\Sigma}}}}\, \dev(\bar{\vec{\Sigma}})
  +
  \frac{\partial \omega}{\partial J_3^{\bar{\vec{\Sigma}}}}\,
  \dev\!\bigl(\dev(\bar{\vec{\Sigma}})^2\bigr).
  \label{eq:Dinelastic}
\end{equation}

\color{author}
Thus, the dissipation inequality reads
\begin{equation}
    \mathcal{D} = \bar{\vec{\Sigma}} :  \left(\frac{\partial \omega}{\partial I_1^{\bar{\vec{\Sigma}}}}\, \vec{I}
  +
  \frac{\partial \omega}{\partial J_2^{\bar{\vec{\Sigma}}}}\, \dev(\bar{\vec{\Sigma}})
  +
  \frac{\partial \omega}{\partial J_3^{\bar{\vec{\Sigma}}}}\,
  \dev\!\bigl(\dev(\bar{\vec{\Sigma}})^2\bigr)\right) \geq 0 \quad \forall \, \bar{\vec{\Sigma}} \in \sym(3) \,.
  \label{eq:dissipation}
\end{equation}
In this way, the dissipation inequality is satisfied for any convex, zero-valued at the origin, and non-negative dissipation potential $\omega$. A proof of this statement is reported in Appendix~\ref{app:limit_analysis}.

However, it is important to note that these conditions are sufficient but not necessary to ensure non-negative dissipation. In fact, there exist dissipation potentials that do not satisfy these conditions but still lead to non-negative dissipation. For instance, a non-convex dissipation potential can still yield non-negative dissipation if the resulting inelastic rate tensor $\bar{\vec{D}}_i$ is such that $\bar{\vec{\Sigma}} : \bar{\vec{D}}_i \geq 0$ for all admissible states of the system \citep{holthusen2026complement}. Therefore, while the conditions on $\omega$ provide a convenient way to guarantee thermodynamic consistency, they do not encompass all possible constitutive models that could be physically valid.
\color{black}


\textbf{Incompressible materials.} For the special case of incompressible materials, i.e., $\det(\vec{F}) = 1$, a Lagrange multiplier term is added to the free energy $\psi$ to obtain the augmented free energy $\hat{\psi}$, 
\begin{equation}
    \hat{\psi} = \psi(\bar{\vec{C}}_e) - p (I_3^{\vec{C}})\,.
\end{equation}
\textcolor{author}{The term $p$ is a Lagrange multiplier, which is determined using the boundary conditions to enforce incompressibility.} It is important to note that the elastic part of the deformation is, in general, not incompressible, meaning that $I_3^{\bar{\vec{C}}_e}$ not necessarily remain equal to one \citep{holthusen2024theory}.

%% file: AA_section_ICKAN.tex
\section{Inelastic Constitutive Kolmogorov–Arnold Networks (iCKANs)}
\label{sec:iCKAN}

Building on the constitutive framework presented in the previous section, this section extends the idea of CKANs \citep{abdolazizi2025constitutive} by introducing inelastic Constitutive Kolmogorov--Arnold Networks (iCKANs) for automatic model discovery of inelastic materials. We start with a brief review of Kolmogorov-Arnold Networks (KANs) \citep{liu2024kan} and their modified variant of monotonic input-convex KANs \citep{thakolkaran2025can}, and then we continue with presenting partially input-convex KANs, in which the output is generally convex with respect to specified inputs. Afterwards, we apply these  partially input-convex KANs to express the elastic and dissipation potential functions of the constitutive material model introduced in the last section, leading to the formulation of iCKANs. Thus, iCKANs naturally preserve mandatory constraints and flexibility due to the expression of activation functions via B-Splines.

\subsection{Input-convex Kolmogorov-Arnold Networks}

Serving as a promising alternative to multi-layer perceptrons, KANs replace the combination of a linear layer following with a nonlinear activation function with a naturally nonlinear, trainable activation. KANs are inspired by the Kolmogorov-Arnold representation theorem, which states that any continuous multivariate function on a bounded domain $f : [0,1]^n \to \mathbb{R}$ can be represented as a finite composition of univariate functions
\begin{equation}
f(\vec{x}_0) = f(x_{0,1},x_{0,2},\dots, x_{0,n}) = \sum_{j=1}^{2n+1} \phi_{1,1,j} \left(\sum_{i=1}^{n}\phi_{0,j,i}(x_{0,i}) \right)\,,
\end{equation}
where $x_{0,i}$ are the elements of the input vector $\vec{x}_0$ and $\phi_{0,j,i} : [0,1] \to \mathbb{R}$ and $\phi_{1,1,j} : \mathbb{R} \to \mathbb{R}$ are univariate continuous functions \citep{kolmogorov1961representation, abdolazizi2025constitutive}. This equation can be seen as a two-layer network with a topology of $[n,2n+1,1]$. Extending this idea, a deeper architecture with $L$ layers can be constructed, resulting in a KAN with the topolgy of $[n_0,n_1,\dots,n_L]$, where $n_i$ is the number of neurons in the $i$-th layer. Thus, KANs can be expressed as the following nested summation
\begin{equation}
f(\vec{x}) =  \sum_{i_{L-1}=1}^{n_{L-1}} \phi_{L-1,i_L,i_{L-1}} \left(\sum_{i_{L-2}=1}^{n_{L-2}} \dots \left( \sum_{i_2=1}^{n_2} \phi_{2,i_3,i_2} \left(\sum_{i_1=1}^{n_1} \phi_{1,i_2,i_1} \left(\sum_{i_0=1}^{n_0}\phi_{0,i_1,i_0}(x_{i_0}) \right)\right) \right) \dots\right)\,.
\label{eq:KAN}
\end{equation}
The activation function connecting the $i$-th neuron in the $l$-th layer $(l,i)$ and the $i$-th neuron in the $(l+1)$-th layer $(l+1,j)$ is denoted as $\phi_{l,j,i}$ for $l = 0, \dots, L-1$, $i = 1, \dots, n_l$ and $j = 1, \dots, n_{l+1}$. An activation function is defined as follows 
\begin{equation}
\phi_{l,j,i}(x) = w_b \cdot b(x) + w_s \cdot s(x)
\label{eq:activation_function}
\end{equation} 
with a basis function $b(x)$, usually $b(x) = x/ (1 + e^{-x})$, and a spline function $s(x) = \sum_i c_i B_i (x)$, where $c_i$ are trainable control points and $B_i$ are B-spline basis functions. The factors $w_b$ and $w_s$ are trainable coefficients. The output of neuron $(l+1,j)$, $x_{l+1,j}$, is defined as the sum of the inputs $x_{l,i}$ after transformation, i.e., $x_{l+1,j} = \sum_{i=1}^{n_l}\phi_{l,j,i}(x_{l,i})$, where each input $x_{l,i}$ is processed by its associated univariate activation function $\phi_{l,j,i}$. For the simplicity of notation, the layer-wise mapping is defined using the nonlinear function matrix $\vec{\Phi}_l$ as follows
\begin{equation}
    \vec{x}_{l+1} = \vec{\Phi}_l(\vec{x}_l) = \begin{bmatrix}
        \phi_{l,1,1(\bullet)} & \phi_{l,1,2(\bullet)} & \dots & \phi_{l,1,n_l(\bullet)} \\
        \phi_{l,2,1(\bullet)} & \phi_{l,2,2(\bullet)} & \dots & \phi_{l,2,n_l(\bullet)} \\
        \vdots & \vdots & \ddots & \vdots \\
        \phi_{l,n_{l+1},1(\bullet)} & \phi_{l,n_{l+1},2(\bullet)} & \dots & \phi_{l,n_{l+1},n_l(\bullet)}   
    \end{bmatrix}\vec{x}_l\,.
\end{equation}
Thus, \autoref{eq:KAN} can be written compactly as 
\begin{equation}
    \mathrm{KAN}(\vec{x}) = (\Phi_{L-1} \circ \Phi_{L-2} \circ \dots\circ \Phi_1 \circ \Phi_0 )(\vec{x})\,.
\end{equation}

\subsubsection{Fully input-convex architecture}
\label{sec:MIKAN}

A function that is twice differentiable is convex if and only if its second derivative is nonnegative over its entire domain. 
For instance, consider a two-layer KAN with the input $\vec{x}_0 = \{x_{0,1},\dots,x_{0,n_0}\}$, we can express it as the following function
\begin{equation}
f(\vec{x}_0) = \sum_{j=1}^{n_i} \phi_{1,1,j} \left(\sum_{i=1}^{n_0}\phi_{0,j,i}(x_{0,i}) \right)\,.
\end{equation}
The first and second derivatives of $f$ with respect to $x_{0,k}$ are given by
\begin{gather}
\frac{\partial f}{\partial x_{0,k}} = \sum_{j=1}^{n_1}\phi_{1,1,j}' \left( \sum_{i=1}^{n_0} \phi_{0,j,i}(x_{0,k}) \right) \cdot \phi_{0,j,k}'(x_{0,k})\,,\\
\frac{\partial^2 f}{\partial x_{0,k}^2} = \sum_{j=1}^{n_1}\phi_{1,1,j}'' \left( \sum_{i=1}^{n_0} \phi_{0,j,i}(x_{0,i}) \right) \cdot \phi_{0,j,k}'(x_{0,k})^2 + \sum_{j=1}^{n_1}\phi_{1,1,j}' \left( \sum_{i=1}^{n_0} \phi_{0,j,i}(x_{0,i}) \right) \cdot \phi_{0,j,k}''(x_{0,k})\,,
\end{gather}
respectively. Consequently, a sufficient condition for $f$ to be convex with respect to the inputs $\vec{x}_0$ is that the following criteria hold
\begin{equation}
    \begin{split}
    \phi_{0,j,i}'' &\geq 0 \quad \forall i = 1, \dots, n_0\,;\,\forall j = 1, \dots, n_1\,; \\
    \phi_{1,1,j}'' &\geq 0 \quad \forall j = 1, \dots, n_1\,;\\
    \phi_{1,1,j}' &\geq 0 \quad \forall j = 1, \dots, n_1\,.
    \end{split}
\end{equation}
That is, all first-layer activation functions $\phi_{0,j,i}$ must be convex, while all second-layer activation functions $\phi_{1,1,j}$ must be both convex and non-decreasing. 

\textbf{Convex and non-decreasing activations.} Constraining all activations to be convex and monotonic non-decreasing results into a monotonic input-convex KAN. To satisfy the aforementioned criteria and achieve monotonically increasing and convex activations, \citet{thakolkaran2025can} proposes to constrain the control points of the B-splines activation functions. First, zero base functions are chosen, i.e. $b(x)=0$. Then, $w_s$ is constrained to be positive, monotonically increasing, and convex by $w_s^* = \mathrm{softplus}(w_s)$. Finally, the control points $c_i$ are modified to become convex coefficients $c_i^*$. Thus, $s^*(x)=\sum_i c_i^* B_i (x)$ is monotonically increasing and convex, i.e.,
\begin{equation}
    c_{i+2}^* - c_{i+1}^*  \geq c_{i+1}^* - c_i^* \geq 0, \,\forall i\,.
\end{equation}
Finally, the original activation in \autoref{eq:activation_function} is modified to 
\begin{equation}
    \phi(x) = w_s^* \cdot \sum_i c_i^* B_i (x)\,.
\end{equation}
To ensure that the activations outside the initial grid range also remain convex and monotonically increasing, linear extrapolation is applied at both ends of the grid range \citep{polo2024monokan, thakolkaran2025can}.

\textbf{Generally convex activations.} The previously introduced approach provides a convex and non-decreasing KAN-output with respect to its inputs. However, in some cases, specifically for the dissipation potential, the output is required to be convex with a desired stationary point. Therefore, we apply an additional activation function to the KAN-output to ensure the convexity while allowing for non-monotonic behavior. It is important to note that the network architecture remains unchanged from the monotonic input-convex KANs and only the output of the KAN is passed through this additional activation function. 

Given a convex and non-decreasing multivariate function $f(\vec{x})$, one can construct a new function $\tilde{f}(\vec{x})$ that is convex, zero-valued at its origin and has a zero derivative at $\vec{x}=\vec{\alpha}$ by passing it through the following defined operation $\mathcal{H}(\bullet)$,
\begin{equation}
\tilde{f}(\vec{x}) = \mathcal{H}(f(\vec{x})) = f(\vec{x}) - f(\vec{\alpha}) - \left.\ \vec{\nabla} f(\vec{x})\right|_{\vec{x} = \vec{\alpha}}  \cdot (\vec{x}-\vec{\alpha})\,,
\label{eq:omega_conv}
\end{equation}
For a more general case, we can apply a further manipulation step,
\begin{equation}
\hat{f}(\vec{x}) = \mathrm{max}(\tilde{f}(\vec{x}) - c, 0) \qquad \text{with}\qquad c > 0\,. \label{eq:omega_c}
\end{equation}
\textcolor{author}{Notably, the resulting function $\hat{f}(\mathbf{x})$ is not continuously differentiable. Consequently, its derivative should be interpreted in the sense of a subgradient \citep{germain1983continuum,holthusen2026complement}.\footnote{When applied to the dissipation potential, the thermodynamic constraints are preserved. In particular, all admissible subgradients at the non-differentiable points satisfy the thermodynamic consistency requirements, thereby ensuring non-negative dissipation.}} Applying the resulting functions $\tilde{f}(\vec{x})$ and $\hat{f}(\vec{x})$ to the KAN-output preserve convexity. A demonstration of this for an univariate simplification is shown in \autoref{fig:postprocessing}. The proof of this theorem can be found in Appendix \ref{app:convexification}. 

\subsubsection{Partially input-convex architecture}
\label{sec:PIKAN}

In general, the outputs are required to be convex with respect to selected input arguments, but not necessarily with respect to all inputs. In particular, the free energy and the dissipation potential must be convex in the invariant-based arguments, while convexity with respect to additional non-mechanical features is not required. For this reason, we developed a partially input-convex KAN, as shown in \autoref{fig:PICKAN} \citep{amos2017input, deschatre2025input}. The output $z= f(\vec{x},\vec{y})$ is convex with respect to the input variables $\vec{x} = [x_1, x_2, x_3]$, but not necessarily convex with respect to the input variable $\vec{y} = [y_1]$. This is achieved by modifying the first layer of a generally input-convex KAN. While the activations corresponding to the convex input variables $\vec{x}$ remain convex and non-decreasing, the activations with respect to the input variable $\vec{y}$ are unconstrained. The outputs of the first layer are then additively combined and passed through the second layer. It is to note, that this implementation is a simplified variant of partially input-convex networks \citep{amos2017input}, as a more generalized implementation is beyond the scope of this work.

\begin{figure}[ht]
\centering
\begin{subfigure}{.49\textwidth}
    \includegraphics{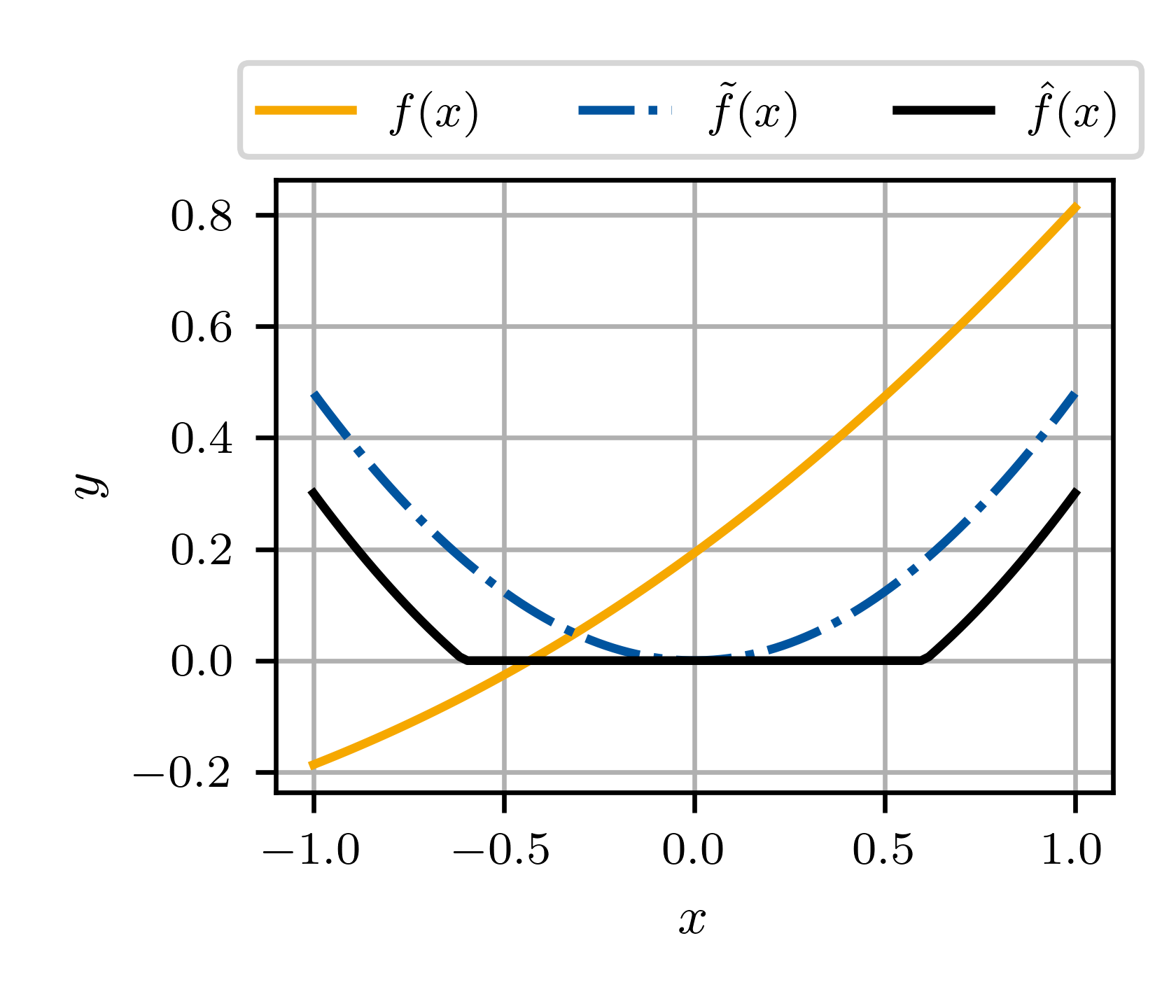}
    \vspace{-1em}
\caption{Postprocessing of a convex, non-decreasing function}
\label{fig:postprocessing}
\end{subfigure}
\begin{subfigure}{.49\textwidth}
    \centering
    \vspace{1em}
    \includegraphics[width=0.72\textwidth]{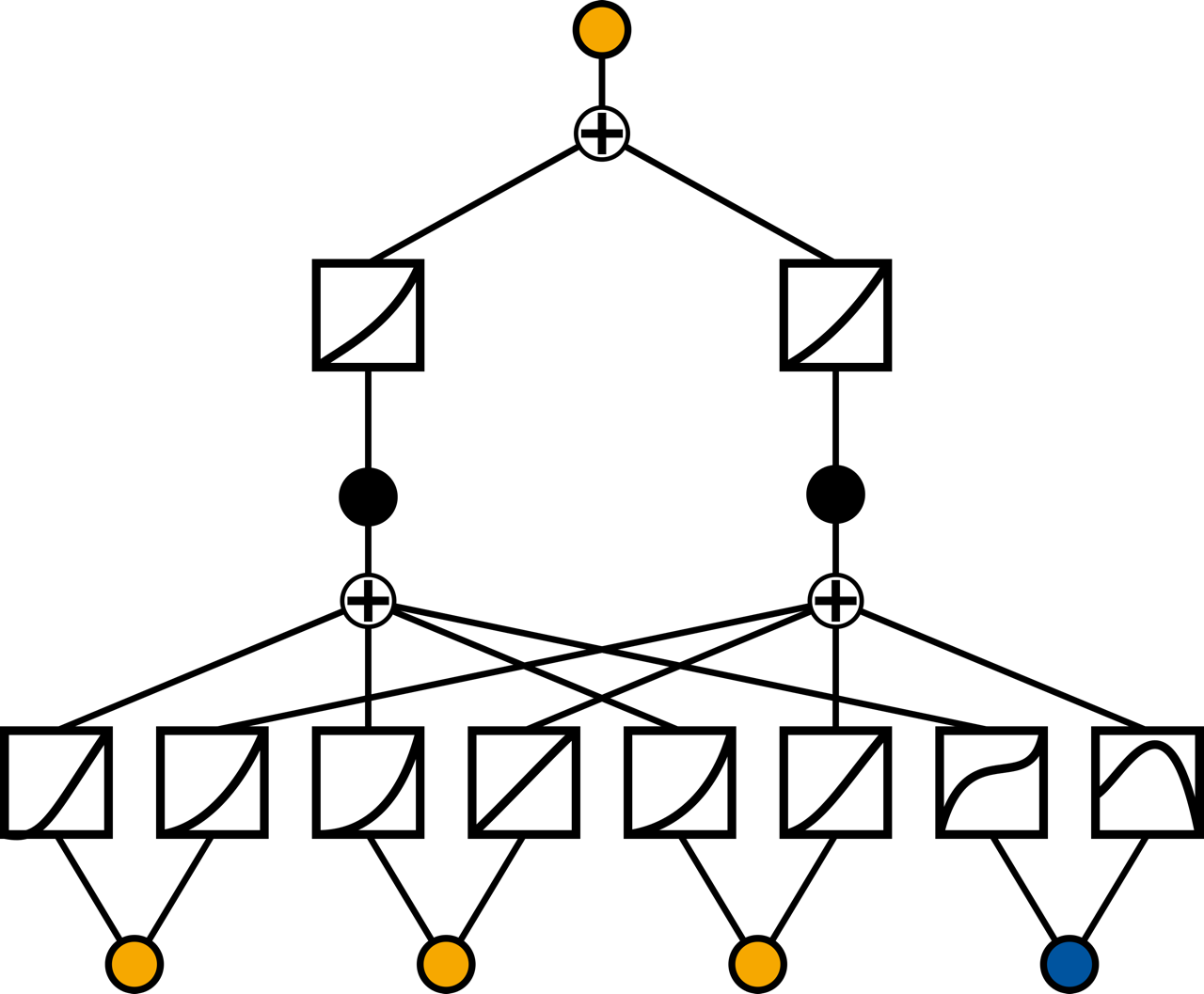}
    \put(-155,-10){$x_1$}
    \put(-110,-10){$x_2$}
    \put(-65,-10){$x_3$}
    \put(-22,-10){$y_1$}
    \put(-87,145){$z$}
    \vspace{0.5em}
    \caption{A partially input-convex Kolmogorov-Arnold Network}
    \label{fig:PICKAN}
\end{subfigure}
\caption{(a) Demonstration of the postprocessing of a convex and non-decreasing function $f(x)$ to achieve a general convex function with the designed stationary point at $x=0$. $\tilde{f}(x)$ is achieved by $\mathcal{H}$-operation in \autoref{eq:omega_conv} and $\hat{f}(x)$ is achieved by \autoref{eq:omega_c}. (b) A partially input-convex Kolmogorov-Arnold Network, where the output $z$ is convex with respect to the yellow-marked input variables $(x_1, x_2, x_3)$ and unconstrained to the blue-marked input variables $y_1$.}
\end{figure}

\subsection{Constitutive modeling with KANs}

Recalling the constitutive framework in Section \ref{sec:constitutive_framework}, the free energy $\psi$ and dissipation potential $\omega$ must be defined to formulate an inelastic material. In this work, we propose using partially input-convex KANs to model these two scalar-valued functions. Thus, we present the inelastic Constitutive Kolmogorov-Arnold Networks (iCKANs) framework, which can be considered an extension to the Constitutive Kolmogorov-Arnold Networks (CKANs) for hyperelastic materials introduced in \citet{abdolazizi2025constitutive}.

\subsubsection{Free energy}

CKANs have been formulated using different choices of functional bases for the free energy function, e.g., pricipal invariants and principal stretches \citep{abdolazizi2025constitutive}. In this work, we adopt the variant, in which the principal invariants are selected as the functional basis for constructing the KAN representing the free energy. In the original CKAN implementation, an input-monotonic KAN is employed \citep{polo2024monokan}. Here, we use the previously introduced input-convex KAN to enforce convexity of the learned free energy representation.

\color{author}
We follow the common design principle of constructing the free energy as a polyconvex function with respect to the deformation gradient $\vec{F}$ \citep{ball1976convexity}.\footnote{To be more precise, polyconvexity is employed here as a constitutive design principle, i.e., in the sense of convexity with respect to the deformation gradient, its cofactors, and its determinant. We do not invoke polyconvexity to establish the existence of minimizers of the underlying inelastic problem, which would additionally require coercivity assumptions and a more detailed mathematical analysis.} It is evident that the free energy should be dependent on the principal invariants of the elastic right Cauchy-Green deformation tensor in the co-rotated intermediate configuration, $\bar{\vec{C}}_e = \vec{U}_i^{-1} \vec{C} \vec{U}_i^{-1}$. For nearly incompressible materials, it is common practice to decouple the volumetric and isochoric contributions of the elastic part of the deformation gradient \citep{holthusen2024theory,flory1961thermodynamic}. Thus, modified invariants $(\hat{I}_1^{\bar{\vec{C}}_e}, \hat{I}_2^{\bar{\vec{C}}_e}, \hat{I}_3^{\bar{\vec{C}}_e})$ are defined based on the principal invariants (\autoref{eq:invariants_Ce}) as \citep{hartmann2003polyconvexity,holthusen2024PAMM}
\begin{gather}
    \hat{I}_1^{\bar{\vec{C}}_e} = \tilde{I}_1^{\bar{\vec{C}}_e} -3 , \quad \hat{I}_2^{\bar{\vec{C}}_e} = (\tilde{I}_2^{\bar{\vec{C}}_e})^{3/2}-3^{3/2}, \quad \hat{I}_3^{\bar{\vec{C}}_e} = (J^{\bar{\vec{C}}_e}-1)^2\,,\\
    \text{with} \quad  J^{\bar{\vec{C}}_e} = \sqrt{\det(\bar{\vec{C}}_e)}\,, \quad \tilde{I}_1^{\bar{\vec{C}}_e} = \left(J^{\bar{\vec{C}}_e}\right)^{-2/3} I_1^{\bar{\vec{C}}_e}, \quad \tilde{I}_2^{\bar{\vec{C}}_e} = \left(J^{\bar{\vec{C}}_e}\right)^{-4/3} I_2^{\bar{\vec{C}}_e}\,.
\end{gather} 
For fixed $\vec{U}_i$, these modified invariants are polyconvex functions of the deformation gradient and can be expressed as convex functions of $\vec{F}$, its cofactors $\mathrm{cof}(\vec{F})$, and its determinant $\det(\vec{F})$,  \citep{holthusen2026complement,holthusen2024PAMM}. A discussion of the polyconvexity properties for the cases where $\vec{U}_i$ changes in $\vec{F}$ is beyond the scope of this work. 

Accordingly, the free energy is constructed using these modified invariants of $\bar{\vec{C}}_e$ as input argument for the input-convex KAN
\begin{equation}
    \psi = \psi(\bar{\vec{C}}_e) = \psi(\hat{I}_1^{\bar{\vec{C}}_e}, \hat{I}_2^{\bar{\vec{C}}_e}, \hat{I}_3^{\bar{\vec{C}}_e}) \,.
\end{equation}
It is should be noted that the polyconvexity with respect to the arguments is not a necessary condition for the free energy to fulfill the physics. Furthermore, a proof that the output free energy $\psi$ and its derivative are zero at undeformed state can be found in Appendix~\ref{app:psi_undeformed}.
\color{black}

\subsubsection{Dissipation potential}

What remains to be defined is the dissipation potential $\omega$, which governs the evolution equation of the inelastic strain. \textcolor{author}{To satisfy thermodynamic consistency, it is sufficient to construct the dissipation potential $\omega$ as zero-valued at its origin, nonnegative and convex with respect to its arguments, which are the modified stress invariants of the elastic Mandel stress $\bar{\vec{\Sigma}}$ (\autoref{eq:invariants_Sig}),
\begin{equation}
    \hat{I}_1^{\bar{\vec{\Sigma}}} = I_1^{\bar{\vec{\Sigma}}}, \quad \hat{J}_2^{\bar{\vec{\Sigma}}} = \sqrt{J_{2}^{\bar{\vec{\Sigma}}}}, \quad \hat{J}_3^{\bar{\vec{\Sigma}}} = \sqrt[3]{J_{3}^{\bar{\vec{\Sigma}}}}\,.
\end{equation}
Thus, we adopt this as design principle and define the dissipation potential dependent on the modified stress invariants as
\begin{equation} 
\omega = \omega(\bar{\vec{\Sigma}}) = \omega\left(\hat{I}_{1}^{\bar{\vec{\Sigma}}}, \hat{J}_{2}^{\bar{\vec{\Sigma}}}, \hat{J}_{3}^{\bar{\vec{\Sigma}}}\right)\,.
\label{eq:omega_3_arguments}
\end{equation}}
Unlike the previously mentioned condition for the free energy, the dissipation potential must be convex with respect to the modified stress invariants for the dissipation inequality to be fulfilled. Noteworthy, we could enhance the existing arguments as long as the convexity condition is satisfied. For instance, to increase network flexibility and expressibility, the negative stress invariants, or additional invariants, e.g., the principal invariants \citep{holthusen2025generalized}, can be appended to the argument of the dissipation potential network. These alternative constructions are presented in Appendix \ref{app:Potential_arguments}. 

Following the approach of the previously presented general input-convex KAN, a convex, zero-valued and nonnegative dissipation potential $\omega$ is constructed by applying $\mathcal{H}$-operator (\autoref{eq:omega_conv}) as a postprocessing activation on the output of the monotonic input-convex KAN, $\omega^{\text{KAN}}$, i.e.,
\begin{equation}
\omega = \mathcal{H}(\omega^\text{KAN}) = \omega^{\text{KAN}} -  \left.\omega^{\text{KAN}}\right|_{\bar{\vec{\Sigma}}=\vec{0}} - \left.\begin{bmatrix}\partial \omega/\partial \hat{I}_1^{\bar{\vec{\Sigma}}}\\\partial \omega/\partial \hat{J}_2^{\bar{\vec{\Sigma}}}\\\partial \omega/\partial \hat{J}_3^{\bar{\vec{\Sigma}}}\end{bmatrix}\right|_{\bar{\vec{\Sigma}} = 0} \cdot \begin{bmatrix}\hat{I}_1^{\bar{\vec{\Sigma}}} \\ \hat{J}_2^{\bar{\vec{\Sigma}}} \\ \hat{J}_3^{\bar{\vec{\Sigma}}}\end{bmatrix}\,.
\label{eq:convex_omega}
\end{equation}

As previously stated, a further postprocessing step shown in \autoref{eq:omega_c} can be applied to $\omega$. Due to the special construction of the dissipation potential, which allows for a zero interval of the derivative of dissipation potential, meaning that no inelastic response is activated until a certain stress level is reached, and one single iCKAN model is sufficient to represent a viscoelastic material under relaxation. This can be represented as the rheological model in \autoref{fig:RM_inelastic}. The free energy $\psi$ is associated with the spring, and the dissipation potential $\omega$ with the dashpot. However, at this end, the training of this parameter $c$ \textcolor{author}{of \autoref{eq:omega_c}} is unstable, leading us to choose to leave out this particular step and combine two iCKANs in parallel to express a viscoelastic material. 

\begin{figure}[ht]
    \centering\small
    \vspace{0.3em}
    \includegraphics[width=.3\textwidth]{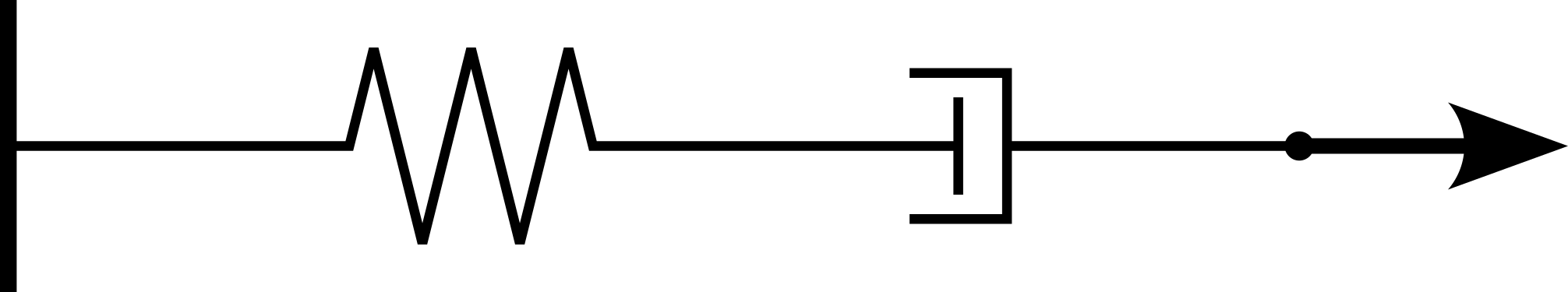}
    \put(-110,27){$\psi(\bar{\vec{C}}_e)$}
    \put(-65,27){$\omega(\bar{\vec{\Sigma}})$}
    \vspace{0.3em}
    \caption{Representation of the iCKAN formulation in the form of a Maxwell model. The free energy $\psi$ is associated with the spring, and the dissipation potential $\omega$ with the dashpot.}
    \label{fig:RM_inelastic}
\end{figure}

\subsubsection{Feature dependency}

The appropriate functional form of the free energy and dissipation potential generally depend on material descriptors such as composition, microstructure, or processing conditions. We collect this information in a feature vector $\vec{f}$ and augment the arguments of both the elastic and dissipation potentials as
\begin{equation}
    \psi = \psi(\hat{I}_1^{\bar{\vec{C}}_e}, \hat{I}_2^{\bar{\vec{C}}_e}, \hat{I}_3^{\bar{\vec{C}}_e},\vec{f}) \quad \text{and} \quad \omega = \omega\left(\hat{I}_1^{\bar{\vec{\Sigma}}}, \hat{J}_2^{\bar{\vec{\Sigma}}}, \hat{J}_3^{\bar{\vec{\Sigma}}},\vec{f}\right)\,.
\end{equation}
This feature is solely non-mechanical and neither the free energy nor the dissipation potential is convex with respect to it. 

\subsubsection{Numerical implementation}

\color{author}
\textbf{Square root and cubic root.} To enable the differentiability at the origin, the square root and cubic root of the stress invariants are modified as \citep{holthusen2025generalized}
\begin{equation}
\sqrt{x} = \frac{x}{(x+\epsilon_1)^{1/2}}, \quad \sqrt[3]{x} = \frac{x}{(x+\epsilon_2)^{2/3}}, \qquad \text{with } \epsilon_1, \epsilon_2 > 0\,.
\end{equation}
This is specifically applied for the calculation of $\hat{J}_2^{\bar{\vec{\Sigma}}} = \sqrt{J_2^{\bar{\vec{\Sigma}}}}$ and $\hat{J}_3^{\bar{\vec{\Sigma}}} = \sqrt[3]{J_3^{\bar{\vec{\Sigma}}}}$, respectively. It is to note, while under this modification, the dissipation inequality in \autoref{eq:dissipation} remains satisfied. As a matter of fact, this modification is not only necessary to ensure numerical stability and convergence during training, but also consistent with the physical interpretation of the dissipation potential. In this modified case, the derivative of the dissipation potential with respect to the invariants is strictly defined at the origin. A proof of the theorem is shown in Appendix~\ref{app:modified_roots}.
\color{black}

The evolution equation defined by the inelastic rate tensor, $\vec{D}_i = \partial \omega / \partial \vec{\Sigma}$ (see also \eqref{eq:Dinelastic}), is solved numerically. The exponential integrator map has been shown to be particularly effective for inelastic problems and in the finite strain regime \citep{holthusen2025generalized,holthusen2025automated}. Numerically, this evolution equation can be integrated using either explicit or implicit time integration schemes. In the explicit formulation, the current inelastic stretch tensor $\vec{U}_{i,t}$ depends solely on the previous time step and the current increment. Conversely, the implicit formulation introduces a nonlinear dependence on both current and previous states, necessitating the solution of a nonlinear system at each time step. Although the implicit scheme is generally less sensitive to time step size, it is computationally more demanding than the explicit approach. \textcolor{author}{To accelerate the solution of the nonlinear system in the implicit scheme, a Liquid Time Constant Network \citep{hasani2021liquid} can be employed as an alternative to the traditional iterative Newton solver \citep{holthusen2026complement,as2023mechanics, rosenkranz2024viscoelasticty}.} Section~\ref{sec:results} presents a side-by-side comparison of both formulations using synthetic data.
Given its superior robustness and performance across all numerical examples discussed in the following section, this work primarily focuses on explicit time integration. Detailed information on the implicit time integration scheme is provided in Appendix~\ref{app:implicit_iCKAN}.

\textbf{Explicit time integration.} Following the explicit integration scheme, the inelastic right Cauchy-Green tensor at the current timestep, $\vec{C}_{i,t}$, is approximated using the inelastic strain rate from the previous timestep, $\vec{D}_{i,t-1}$, as
\begin{equation}
    \vec{C}_{i,t} = \vec{U}_{i,t-1} \exp(2 \,\Delta t\, \vec{D}_{i,t-1}) \, \vec{U}_{i,t-1}, \quad \vec{U}_{i,t} = \vec{C}_{i,t}^{1/2}\,,
    \label{eq:explicit}
\end{equation}
where $\Delta t$ is the time increment between timestep $t-1$ and $t$ \citep{holthusen2025generalized}. The overall architecture of the explicit iCKAN is illustrated in \autoref{fig:iCKAN_explicit}. \textcolor{author}{Noteworthy, while a Tyler expansion might be sufficient to generate the matrix square-root under moderate deformations, a closed-form representation by \citet{hudobivnik2016closed} are utilized in this work for nonlinear finite strains, specifically for computing the matrix square roots, $\vec{U}_{i,t} = \vec{C}_{i,t}^{1/2}$ \citep{holthusen2025generalized}.} The explicit scheme proves to be computationally efficient, as no iterative solver is required. Stability is controlled through the selection of the time increment $\Delta t$, offering a straightforward and transparent mechanism for maintaining numerical robustness. \textcolor{author}{In the presented work, the time increment $\Delta t$ is chosen sufficiently small to avoid stability issues associated with the explicit integration scheme. Since $\Delta t$ is included as an input to the iCKAN, the network can account for the employed temporal resolution.}

\begin{figure}[ht]
    \centering
    \includestandalone{standalone_iCKAN_explicit}
    \caption{Explicit iCKAN architecture at timestep $t$. The inputs at each step are the current deformation gradient and time increment $(\vec{F}_t,\Delta t)$ together with the state variables $(\vec{C}_{t-1},\vec{U}_{i,t-1})$ from the previous step. The KAN models for the free energy $\psi$ and the dissipation potential $\omega$ are evaluated to update the state variables according to \autoref{eq:explicit}. The updated state is then propagated to the next time step. In addition, the output stress $\vec{P}$ is computed by evaluating the free energy $\psi$ with the updated state variables.}
    \label{fig:iCKAN_explicit}
\end{figure}

\textbf{Initial grid range.} Since the KAN activation functions are represented by B-splines defined on fixed domains, suitable input grid range must be specified for each input dimension. In the present recurrent formulation, these effective inputs are not known a priori, as they are generated during the forward pass and evolve throughout training. We therefore adopt a data-driven initialization strategy to estimate appropriate grid ranges from physically meaningful tensor invariants computed on the training data. The detailed construction and its justification are provided in Appendix~\ref{app:init_grid_range}.

%% file: standalone_iCKAN_explicit.tex
\begin{tikzpicture}[
    font=\rmfamily \small,
    >=LaTeX,
    cell/.style={
        rectangle, 
        rounded corners=5mm, 
        draw,
        very thick,
        },
    operator/.style={
        circle,
        draw,
        inner sep=-0.5pt,
        minimum height =.2cm,
        },
    function/.style={
        ellipse,
        draw,
        inner sep=1pt
        },
    ct/.style={
        circle,
        draw,
        line width = .75pt,
        minimum width=1cm,
        inner sep=1pt,
        },
    ct2/.style={
        circle,
        draw,
        line width = 1.5pt,
        minimum width=1.2cm,
        inner sep=1pt,
        },
    gt/.style={
        rectangle,
        draw,
        minimum width=4mm,
        minimum height=3mm,
        inner sep=1pt
        },
    mylabel/.style={
        font=\small\rmfamily
        },
    ArrowC1/.style={
        rounded corners=.2cm,
        thick,
        },
    ArrowC2/.style={
        rounded corners=.5cm,
        thick,
        },
    ArrowC3/.style={
        rounded corners=.3cm,
        thick,
        },
    ]

    \node [cell, minimum height = 4cm, minimum width=9cm, fill=RWTHBlau!5] at (0,0){} ;

    \node [ct, fill=RWTHBlau!25] (ibox1) at (-3,-1.20) {$\bar{\vec{C}}_{e,t-1}$};
    \node [inner sep=0pt, outer sep=0pt, draw=none] (ibox2) at (-1.5,-1.20) {\includegraphics[width=1.2cm]{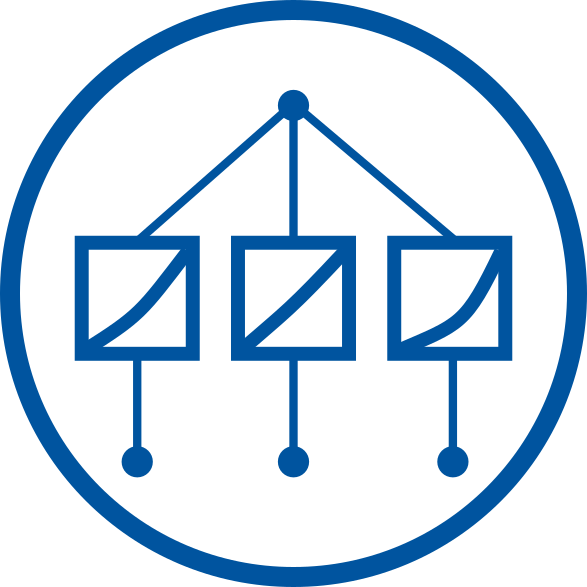}};
    \node [ct, fill=RWTHBlau!25] (ibox3) at (0.25,-1.20) {$\bar{\vec{\Sigma}}$};
    \node [inner sep=0pt, outer sep=0pt, draw=none] (ibox4) at (1.75,-1.20) {\includegraphics[width=1.2cm]{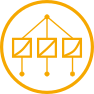}};
    \node [ct, fill=RWTHBlau!25] (ibox5) at (3.5,-1.20) {$\bar{\vec{D}}_i$};

    \node [ct, fill=RWTHBlau!25] (LTC) at (-0.5,0.75) {$\bar{\vec{U}}_i$};
    \node [ct, fill=RWTHBlau!25] (add1) at (1.25,0.75) {$\bar{\vec{C}}_e$};
    \node [inner sep=0pt, outer sep=0pt, draw=none] (psi2) at (2.75,0.75) {\includegraphics[width=1.2cm]{figures_KAN.png}};

    \node[ct, label={[mylabel]State}, fill=RWTHOrange!25, align=center] (h) at (-5.5,-1.20) {$\vec{C}_{t-1}$,\\$\vec{U}_{i,t-1}$};
    \node[ct, label={[mylabel]left:Input}, fill=RWTHOrange!25] (x) at (-4.25,-3) {$\vec{F}_{t},\Delta t$};

    \node[ct, label={[mylabel]State}, fill=RWTHBordeaux!25,align=center] (h1) at (5.5,1.75) {$\vec{C}_{t}$,\\$\vec{U}_{i,t}$};
    \node[ct, label={[mylabel]left:Output}, fill=RWTHBordeaux!25] (x2) at (2.75,3) {$\vec{P}_{t}$};

    \put(-36,-15){$\psi_{t-1}$}
    \put(88,38){$\psi_t$}
    \put(56,-15){$\omega_{t-1}$}

    \draw [->, ArrowC1] (h) -- (ibox1);
    \draw [-, ArrowC1] (ibox1) -- (ibox2);
    \draw [->, ArrowC1] (ibox2) -- (ibox3);
    \draw [-, ArrowC1] (ibox3) -- (ibox4);
    \draw [->, ArrowC2] (x) |- (ibox1);
    \draw [->, ArrowC2] (x) |- (LTC);
    \draw [->, ArrowC1] (LTC) -- (add1);
    \draw [-, ArrowC1] (add1) -- (psi2);
    \draw [->, ArrowC1] (ibox4) -- (ibox5);
    \draw [->, ArrowC1] (LTC) |- (h1);
    \draw [->, ArrowC1] (psi2) -- (x2);
    \draw [-, ArrowC1] (ibox5) |- (1,-0.125);
    \draw [->, ArrowC1] (1,-0.125) -| (LTC);

\end{tikzpicture}

%% file: AA_section_Symbolic.tex
\section{Symbolic constitutive modeling}
\label{sec:symbolic}

\color{author}
While the mapping form input to output of a KAN can be expressed using closed-form expression, the interpretability of the activation functions in form of B-splines are limited. To obtain an interpretable closed-form expression of the constitutive models, symbolic constituvie modeling is applied to the trained KANs for the free energy and dissipation potential after model training. In this section, we briefly review the sparsification and symbolification process, for more details, please refer to \citep{abdolazizi2025constitutive,liu2024kan}.
\color{black}

\textbf{Sparsification and pruning.} To achieve a sparse network structure, L1 regularization is applied to the trainable parameters, which in this case correspond to the control points of the B-Spline activations. Furthermore, pruning methods can be applied to the trained KANs to remove less important connections or nodes, further simplifying the model. This can help improve interpretability by focusing on the most relevant features and reducing the overall complexity of the network.

\textbf{Symbolification of convex, monotonic activations.} The trained activation functions in KANs can be replaced by symbolic formulas using symbolic regression techniques. The symbolification process aims to find a mathematical expression that closely approximates the behavior of the trained KANs while being more interpretable. This is achieved by searching through a space of mathematical expressions $f_{l,j,i}$ and selecting the one that best fits the data generated by the KANs. For each activation $\phi_{l,j,i}$, we aim to find the optimal symbolically approximated activation,
\begin{equation}
    y_{l,j,i} = c_{l,j,i}\cdot(f_{l,j,i}(a_{l,j,i}\cdot x_{l,i}+b_{l,j,i}))+d_{l,j,i}\,,
\end{equation}
so that the squared error between $\phi_{l,j,i}$ and $y_{l,j,i}$ is minimized. For this, a library of designed symbolic functions is provided and additional affine parameters $(a,b,c,d)_{l,j,i}$ are introduced. Furthermore, the parameters $(a,c)_{l,j,i}$ are constrained to be nonnegative to preserve the monotonicity and convexity of the activation functions.

Due to the recurrent network architecture, only the activation values from the final forward pass (i.e., the last time step) are retained. Consequently, a forward pass over the entire dataset is required to recover activation values across the full input domain. Otherwise, the extracted symbolic function does not faithfully reproduce the original B-spline activation.

%% file: AA_section_Results.tex
\section{Numerical examples}
\label{sec:results}

In this section, we validate the proposed iCKAN framework on synthetic and experimental datasets, including VHB 4910 \citep{hossain2012experimental} and temperature-dependent VHB 4905 polymer \citep{liao2020thermo}. These examples demonstrate that iCKAN capture inelastic material behaviors and additional feature effects. Symbolic expressions of the trained iCKANs are presented to highlight the interpretability of the approach.

For all training processes, the initial weights and biases of the KANs are initialized using a uniform random distribution. The loss function is defined as the normalized mean squared error (NMSE) between the predicted and target first Piola-Kirchhoff stresses, $\vec{P}$ and $\hat{\vec{P}}$, respectively. The error is normalized by the largest absolute entry of the target stress for each batch sample, computed over all time steps and stress components. Formally, the stress loss is given by
\begin{equation}
\mathcal{L}_{\text{stress}} =
\frac{1}{BTS}
\sum_{b=1}^{B}
\sum_{t=1}^{T}
\sum_{s=1}^{S}
\left|
\frac{
P_{b,t,s} - \hat{P}_{b,t,s}
}{
\max\limits_{t,s} \big|\hat{P}_{b,t,s}\big|
}
\right|^2 \, ,
\end{equation}
where $B$ is the batch size, $T$ is the number of time steps, and $S$ is the number of stress components. 

\color{author}
For explicit time integration scheme, the total loss function is the summation of the stress loss $\mathcal{L}_{\text{stress}}$ and the L1 regularization loss $\mathcal{L}_{\text{L1}}$,
\begin{equation}
    \mathcal{L}_\text{L1} =
\sum_{l=0}^{L-1}
\sum_{j=1}^{n_{\ell+1}}
\sum_{i=1}^{n_{\ell}}
\sum_{m=1}^{M}
\left| c_m^{(\ell,j,i)} \right|,
\end{equation}
to achieve a sparse network, i.e.,
\begin{equation}
\mathcal{L}_{\text{total}} = \mathcal{L}_{\text{stress}} + \lambda_\text{L1}\cdot \mathcal{L}_{\text{L1}}\,,
\end{equation}
with $\lambda_\text{L1}$ being the magnitude of the L1 regularization. 

In the explicit time integration scheme, the evolution equations are enforced directly through the numerical update of the state variables, i.e., they are satisfied by construction during the forward computation and are not included as an additional loss term. In this sense, the optimization is constrained by the evolution equations being enforced exactly in the integration scheme.
For the implicit time integration scheme, the loss function is extended by an additional residual term associated with the evolution equations as shown in \autoref{eq:loss_implicit}, leading to a penalized optimization formulation.

The AMSGRAD optimizer \citep{reddi2019convergence}, which combines the benefits of ADAM and RMSPROP, is utilized for training in combination with a cyclic learning rate scheduler \citep{smith2017cyclical}, where the learning rate cyclically varies between a base learning rate and a maximum learning rate over a given step size. Furthermore, gradient clipping is applied to prevent gradient explosion. 

\color{black}

\subsection{Verification on synthetic data} 
\label{sec:results_synthetic}

\textcolor{author}{Synthetic data are generated using an explicit time integration scheme within the iCKAN framework, where the corresponding KAN-based free energy and dissipation potential are replaced by symbolified expressions. The free energy is chosen as 
\mbox{$\psi = 0.5 \cdot \hat{I}_1^{\bar{\vec{C}}_e} + 0.5 \cdot \hat{I}_3^{\bar{\vec{C}}_e}$}, and the dissipation potential is defined as
\mbox{$\omega = 0.1 \cdot (\hat{I}_1^{\bar{\vec{\Sigma}}})^2 + (\hat{J}_2^{\bar{\vec{\Sigma}}})^2$}.} We only change the first entry of the deformation gradient from the undeformed state and apply a cyclic loading case of tension, relaxation, compression and relaxation. We set the maximum stretches to $F_{11}^{\max} = \{1.1, 1.2, 1.3\}$ and set the time of tensile loading to $t_{\text{load}} = \{0.2, 0.5, 1.0\}$ s, resulting into different strain rates. At the end, the dataset consists of deformation gradients $\vec{F}$ and time increments $\Delta t$ as inputs and the corresponding first Piola–Kirchhoff stresses $\vec{P}$ as outputs. 

\textbf{Model training.} For the training, we use only the tensile loading and relaxation part of the datasets with $t_{\text{load}} = \{0.2, 0.5\}$ s for all three stretch levels. We consider an iCKAN which can be illustrated as the rheological model in \autoref{fig:RM_inelastic}. Both the free energy and dissipation potential are modeled by input-convex KANs with topology [3,1]. The initial grid ranges are approximated as described in Appendix~\ref{app:init_grid_range}. To improve robustness and account for previously unseen input ranges, the grid bounds are extended by 20\%. For arguments that may attain negative values, the lower bound is chosen symmetrically with respect to the upper bound. We examine both the explicit and implicit variant of iCKAN. Training is performed with an initial learning rate of $5\cdot10^{-4}$ and a cyclic scheduler between $5\cdot10^{-5}$ and $1\cdot10^{-4}$ with a step size of 50 iterations. L1 regularization of $\lambda_\text{L1} = 10^{-5}$ and gradient clipping at 0.1 are applied. For the implicit iCKAN, the initial evolution penalty was chosen as $\lambda_{\text{evo}}=1000$ and increased by a factor of two every 250 epochs. 

\textbf{Symbolification.} We demonstrate the symbolification on the trained iCKAN. All activations in the KAN for free energy are approximated using a library of convex, non-decreasing functions $f(x)$ on $[0,\infty)$. This process is illustrated in \autoref{subfig-1:psi} and the full set of candidate functions is listed in Appendix~\ref{app:candidate_functions}. The symbolic formulation of the output dissipation potential $\omega^\text{KAN}$ can be obtained analogously while the enforcement of convexity, zero-valued at its origin and non-negativity is achieved by applying $\mathcal{H}$-operator in \autoref{eq:convex_omega} on the symbolified $\omega^\text{KAN}$. 

Alternatively, for the special case of the below presented single-layer KAN, the network output reduces to
\begin{equation}
    \omega^\text{KAN} = \sum_{i=1}^{n_0}\phi_{0,j,i}(x_{0,i})\,.
\end{equation}
In this case, it is equivalent to apply $\mathcal{H}$-operator to the full output $\omega^\text{KAN}$ or to each activation function $\phi_{0,0,i}$, i.e.,
\begin{equation}
    \mathcal{H}(\omega^\text{KAN}) = \sum_{i=1}^{n_0}\mathcal{H}(\phi_{0,0,i}(x_{0,i}))\,.
\end{equation}
Consequently, the desired dissipation potential can be obtained by directly symbolizing $\mathcal{H}(\phi_{0,0,i})$ using candidate functions that are convex, nonnegative, and zero at the origin (see Appendix~\ref{app:candidate_functions}). The resulting symbolic expression inherently satisfies all constraints without requiring further postprocessing, i.e., $\omega = \omega^\text{KAN}$. This procedure is illustrated in \autoref{subfig-2:omega}. 

\begin{figure}[H]
    \input{standalone_iCKAN_symb_synthetic}
    \caption{Symbolification of the activation functions of the trained iCKAN on synthetic data. Black curves denote the learned B-spline activations, and red curves their symbolified approximations. (a) KAN-based free energy: convex, non-decreasing activations are approximated by symbolic functions that preserve convexity and monotonicity over the admissible domain. (b) KAN-based dissipation potential: activations are transformed by $\mathcal{H}$-operator (\autoref{eq:omega_conv}) and then symbolified by convex, nonnegative functions that are zero-valued at the origin.}
    \label{fig:Symbolification}
\end{figure}

\color{author}
\textbf{Remark.} However, it should be noted that this simplification does not apply to deeper network architectures. For instance, for a two-layer KAN, the output can be expressed as
\begin{equation}
    \omega^\text{KAN} = \sum_{j=1}^{n_i} \phi_{1,1,j} \left(\sum_{i=1}^{n_0}\phi_{0,j,i}(x_{0,i}) \right).
\end{equation}
It is apparent that applying $\mathcal{H}$-operator to the full output or to each activation function yield different results, i.e.,
\begin{equation}
\mathcal{H}(\omega^\text{KAN})  \neq \sum_{j=1}^{n_i} \tilde{\phi}_{1,1,j} \left(\sum_{i=1}^{n_0}\tilde{\phi}_{0,j,i}(x_{0,i}) \right) \qquad \text{with}\qquad \tilde{\phi}_{l,j,i}(x) = \mathcal{H}(\phi_{l,j,i}(x)).
\end{equation}
\color{black}

\textbf{Results.} The final symbolic expressions for the free energy and dissipation potential function identified by the implicit and explicit iCKAN approaches are presented in \autoref{fig:results_synthetic}c. The corresponding predictions are shown in \autoref{fig:results_synthetic}a and \autoref{fig:results_synthetic}b. 
\color{author}
The full results regarding model validation with synthetic data are provided in Appendix~\ref{app:results_synthetic}, where \autoref{tab:results_synthetic} shows the error metrics of the model prediction for training dataset and testing dataset, which is evaluated before and after the symbolification process. The results indicate that symbolification introduces only a minor difference in prediction error, while the overall agreement between iCKAN predictions and both training and testing datasets remains high.

\begin{figure}[H]
    \centering
    \begin{subfigure}{\textwidth}
    \includegraphics{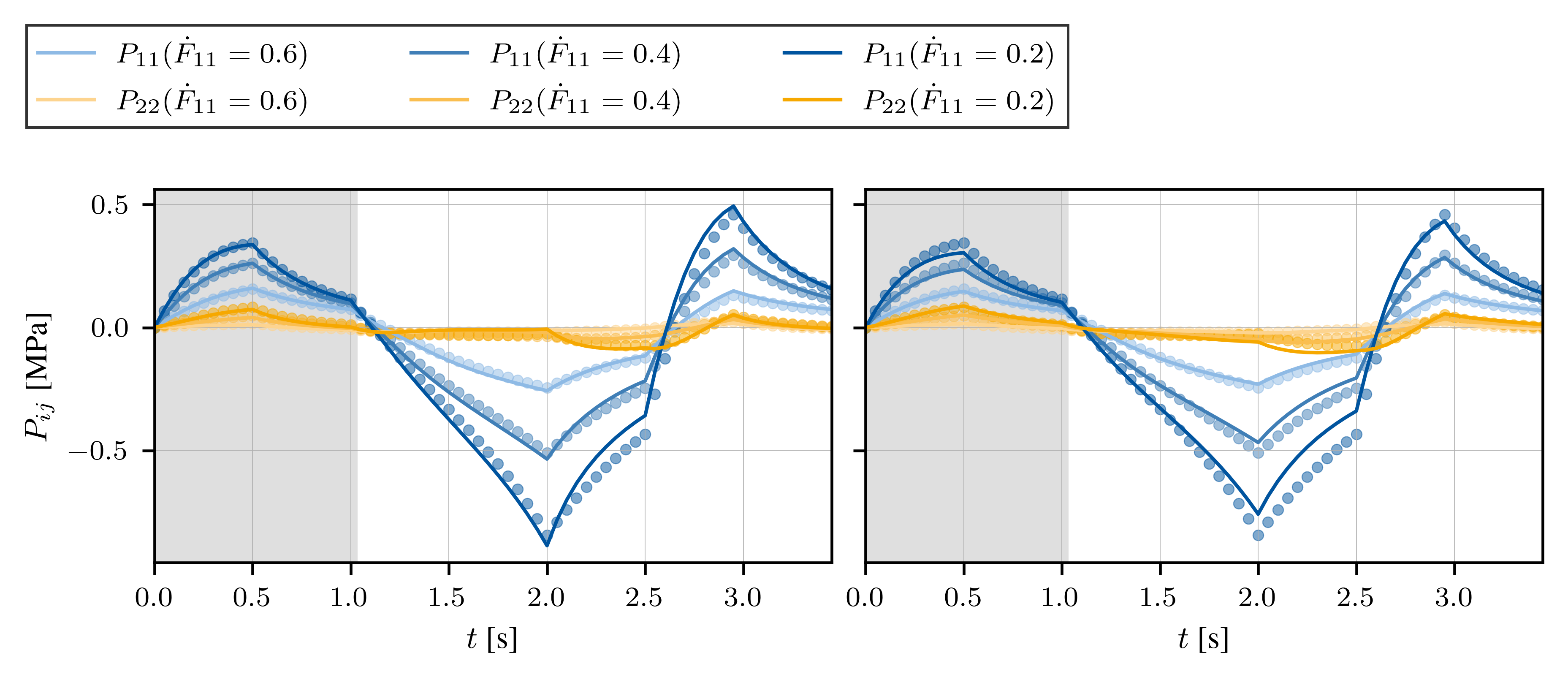}
        \put(-360,-5){(a) \textbf{Explicit time integration}}   
        \put(-150,-5){(b) \textbf{Implicit time integration}}        
        \put(-140,135){\includegraphics{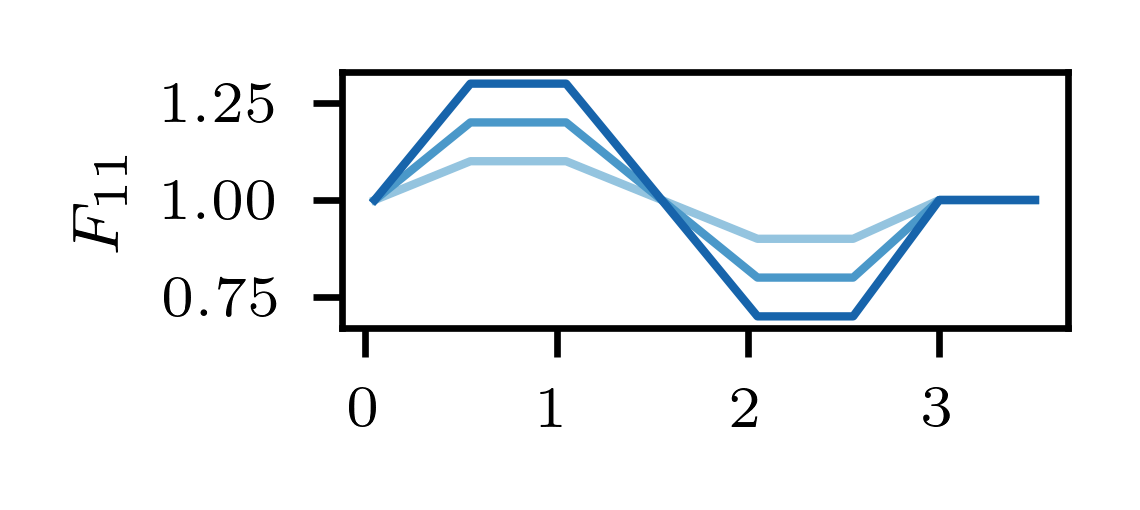}}
        \put(-20,145){\footnotesize $t$ [s]}
    \end{subfigure}
    \vspace{.5em}
    \label{fig:placeholder}
    \begin{subfigure}{\textwidth}
        \centering\small
        \input{table_symb_synthetic}
        \vspace{.2em}

        \normalsize
        (c) \textbf{Discovered free energy and dissipation potential function for the synthetic dataset}
    \end{subfigure}
    \caption{Results of the symbolified iCKAN on synthetic compressible material data using
        (a) the explicit time integration scheme and 
        (b) the implicit time integration scheme, using the expressions presented in Table (c). 
        The gray regions in (a) and (b) indicate the range of the training data, and the colored dots represent the corresponding reference data.}
    \label{fig:results_synthetic}
\end{figure}

Note that the considered different loading durations correspond to different strain rates. Therefore, the observed differences in prediction error should be interpreted as strain-rate-dependent effects rather than as a consequence of the loading duration itself.
\color{black}
Interestingly, both time integration approaches produced comparable results on this simple synthetic dataset, with only minor differences in the discovered symbolic forms of the free energy and dissipation potential. 
\color{author}
Furthermore, a brief sensitivity analysis of the influence of the predictive error scales with the time increment is shown in Appendix~\ref{app:sensitivity_analysis_time_increment}, indicating that the observed error is primarily governed by the numerical properties of the time integration scheme rather than by limitations of the iCKAN architecture.
\color{black}
This demonstrates that, in principle, either approach leads to similar outcomes and can be effectively employed in practice to discover material models using iCKANs.

\subsection{Model discovery for experimental data}

To access the practical applicability of iCKANs, we examine the framework on experimental data of viscoelastic polymers. Due to the fact that the provided experimental data are one-dimensional, the information of volumetric modes are lacking. Thus, we assume incompressible materials \citep{holthusen2025generalized,abdolazizi2024viscoelastic}. The detailed hyperparameters are provided in Appendix~\ref{app:info_experiment}, together with the loss convergence curves over the training epochs.

\color{author}
Trainings were performed on a local workstation equipped with an Apple M5 system-on-chip (10 CPU cores) and 16 GB unified memory. Approximately 10 minutes of CPU time were required for every 50 training epochs, corresponding to a total training time of roughly 5 hours. For each dataset, several preliminary training runs were conducted to identify suitable hyperparameter settings, including learning rates, regularization strengths, and network topologies. The results presented in the manuscript correspond to the final selected configurations. For each dataset, 10 independent training runs with different random initializations were conducted. The reported results correspond to the best-performing model among these runs. 
\color{black}

\subsubsection{Viscoelastic model discovery of VHB 4910 polymer}
\label{sec:results_VHB4910}

Very-High-Bond (VHB) 4910 is a polymer exhibiting highly nonlinear and viscoelastic behavior. Experimental uniaxial loading–unloading data for VHB 4910 are reported by \citet{hossain2012experimental} for four maximum stretch levels, \mbox{$F_{11}=\{1.5,,2.0,,2.5,,3.0\}$}, tested under three different stretch rates, \mbox{$\dot{F}_{11}=\{0.01,,0.03,,0.05\}\,\text{s}^{-1}$}. In our framework, rate dependence is incorporated directly through the network input pair $(\vec{F},\Delta t)$, allowing the model to learn time-dependent behavior without introducing explicit viscosity terms. 

To capture more complex material behaviors, we utilize two iCKANs in parallel, motivated by classical rheological representations of viscoelastic materials, as illustrated in \autoref{fig:KANs_VHB4910}b. In continuum mechanics, the total stress response of viscoelastic solids is often represented as the superposition of multiple parallel mechanisms, each corresponding to a distinct free energy and dissipation potential. The two branches of the parallel structure are indexed by superscripts $^1(\bullet)$ and $^2(\bullet)$, which share the same input deformation gradient $\vec{F}$. Each branch learns its own free energy $^i\psi(^i\bar{\vec{C}}_e)$ and dissipation potential $^i\omega(^i\bar{\vec{\Sigma}})$, while enforcing thermodynamic consistency. The stress contribution $^i\vec{P}$ of each branch is obtained consistently via differentiation of the learned potentials. The total stress response is then computed as the additive superposition $\vec{P}= {}^1\vec{P} + {}^2\vec{P}$, which is consistent with classical rheological network models and enables the representation of multiple concurrent elastic and inelastic mechanisms with different characteristic responses.

\textbf{Model training.} For training, we use the dataset corresponding to the highest stretch level, $F_{11}=3.0$, while the remaining stretch levels are reserved for validation to assess generalization for unseen data. We approximate the initial grid range as previously and extend them by 20\% on both sides. 

\begin{figure}[ht]
    \input{standalone_iCKAN_results_VHB4910}
    \vspace{-2em}
    \put(-390,-12){(a) \textbf{Discovered model for VHB 4910 polymer}}
    \put(-120,20){(b) \textbf{Rheological model}}
    \vspace{2em}
    \caption{KAN architectures of discovered model for VHB 4910 polymer. (b) Rheological model consisting two iCKAN branches in parallel combination. Both branches receive the same input deformation gradient $\vec{F}$ and each learns an free energy and an dissipation potential. The first branch identifies $^1\psi(^1\bar{\vec{C}}_e)$ and $^1\omega(^1\bar{\vec{\Sigma}})$ and produces the stress $^1\vec{P}$ while preserving thermodynamic consistency. The second branch identifies $^2\psi(^2\bar{\vec{C}}_e)$ and $^2\omega(^2\bar{\vec{\Sigma}})$ and produces the stress $^2\vec{P}$ under the same thermodynamic constraints. The total stress response is obtained by summing the stresses of both branches, $\vec{P} = {}^1\vec{P} + {}^2\vec{P}$.}
    \label{fig:KANs_VHB4910}
\end{figure}

\textbf{Symbolification and results.} At the end, we utilize the same library of convex, monotoonic functions as in the previous section to symbolify the activations of the trained model. The symbolified discovered KAN models of the corresponding free energy and dissipation potential are shown in \autoref{fig:KANs_VHB4910}a with their corresponding symbolic functions in \autoref{fig:results_VHB4910}b. The predicted results of the symbolified iCKAN model are shown in \autoref{fig:results_VHB4910}a. 
\color{author}
A quantitative comparison of the model performance before and after the symbolification process is provided in \autoref{tab:results_VHB}. Although the symbolification step introduces a modest increase in prediction error, the resulting symbolic expressions retain the overall predictive capability of the original iCKAN models while offering improved interpretability and insights into the underlying constitutive behavior.

Although the application of the $\mathcal{H}$-operator can reduce the immediate interpretability of the discovered dissipation potential, further symbolic post-processing may be performed when the variable dependencies are sufficiently separated. This enables the derivation of more compact and physically interpretable expressions while retaining the thermodynamic constraints imposed by the $\mathcal{H}$-operator.

\color{black}
\begin{figure}[H]
    \centering
    \begin{subfigure}{\textwidth}
        \centering
        \small
        \includegraphics{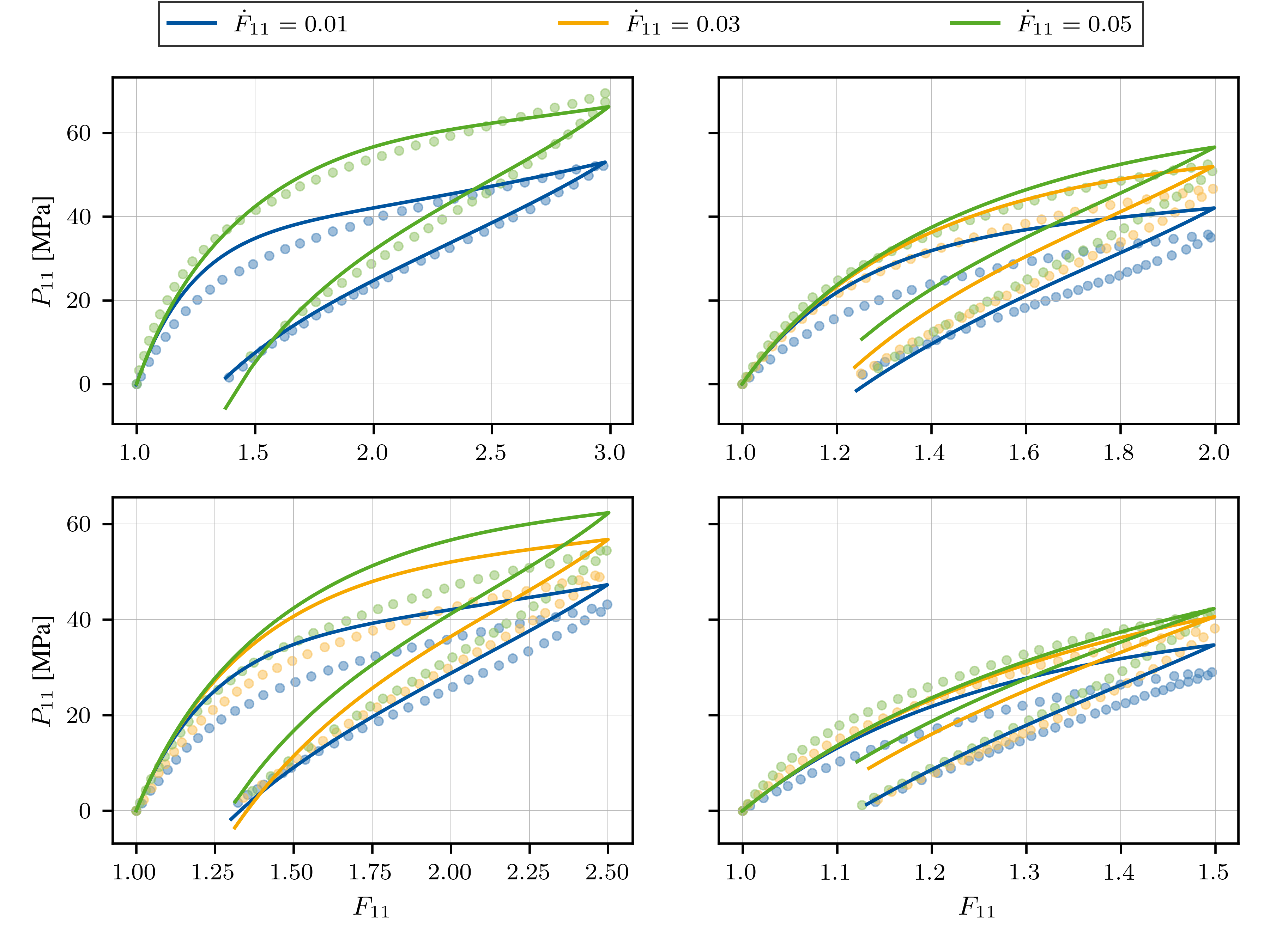}
        \put(-265,193){\fcolorbox{black}{white}{\textbf{Training}}}
        \put(-50,193){\fcolorbox{black}{white}{\textbf{Testing}}}
        \put(-260,50){\fcolorbox{black}{white}{\textbf{Testing}}}
        \put(-50,50){\fcolorbox{black}{white}{\textbf{Testing}}}
        \put(-300,0){(a) \textbf{Predictions of the symbolified iCKANs for VHB 4910}}
    \end{subfigure}

    \begin{subfigure}{\textwidth}
        \vspace{.5em}
        
        \label{tab:symbol_params_VHB4910}
        \centering\small
        \input{table_symb_VHB4910}
        \vspace{.2em}
        
        (b) \textbf{Discovered free energy and dissipation potential functions for VHB 4910 polymer}
    \end{subfigure}
    \caption{Discovered model for VHB 4910 polymer. (a) Predictions by the symbolified trained iCKAN under three different constant loading/unloading rates, denoted as $\dot{F}_{11}$. The prediction of the trained iCKAN model is represented by the solid lines and the experimentally measured data is represented by dots in corresponding colors. Only the experimental data corresponding to $F_{11}^{\text{max}} = 3.0$ [-] are utilized for training. (b) Corresponding symbolic functions.}
    \label{fig:results_VHB4910}
\end{figure}

\subsubsection{Thermo-viscoelastic model discovery of VHB 4905 polymer}
\label{sec:result_VHB4905}
Building on the previous section, we next consider the identification of material models augmented by an additional feature vector $\vec{f}$. Our goal is to find a single iCKAN expression capable of describing the material behavior under varying non-mechanical conditions. To this end, we study the VHB 4905 polymer, which exhibits highly deformable, viscoelastic, and temperature-sensitive behavior. Uniaxial experimental data for VHB 4905 polymer at different temperatures, \mbox{$\theta = \{0, 10, 20,40,60,80\}\, ^\circ \mathrm{C}$}, under multiple strain rates, $\dot{F}_{11} = \{0.03,0.05,0.1\}\, \mathrm{s}^{-1}$, and maximum stretch ratios, $F_{11} = \{2.0,3.5,4.0\}$, are reported in \citep{liao2020thermo}. 

\textbf{Temperature as feature.} The strain rate and deformation information are naturally incorporated into the model inputs. In addition, the temperature $\theta$ is explicitly included as a feature in the arguments of the free energy \citep{abdolazizi2025constitutive}, i.e., 
\begin{equation}
    ^i\psi(^i\bar{\vec{C}}_e) = {}^i\psi(\hat{I}_1^{^i\bar{\vec{C}}_e}, \hat{I}_2^{^i\bar{\vec{C}}_e}, \hat{I}_3^{^i\bar{\vec{C}}_e},\theta)\,.
\end{equation}
It is important to note that this approach differs from classical thermo-mechanical constitutive modeling, where temperature is included via either thermal strain decomposition or temperature-dependent material parameters. Instead, this feature-augmented construction provides a general, data-driven mechanism to capture the influence of unknown non-mechanical factors on the free energy, and consequently, on the material stress response. The initial grid range for this temperature feature is set directly by the experimental temperature range, while the initial grid ranges are selected according to Appendix~\ref{app:init_grid_range} and subsequently extended by 20\% on both ends. Although it is reasonable to assume the feature has influence on both free energy and dissipation potential, we choose to omit it for the dissipation potential due to network robustness. 

\color{author}
\textbf{Training.} For each temperature level, the loading paths corresponding to $\dot{F}_{11}=\{0.03,0.1\}\,\mathrm{s}^{-1}$ at $F_{11}^{\max}=4.0$ are used for training, while the remaining loading paths are reserved for testing. The resulting data split is summarized in \autoref{tab:split_VHB4905}. The adopted split contains interpolation scenarios with respect to strain rate across all temperature levels, providing a challenging validation setting in a multi-temperature learning framework. The loss during training is shown in \autoref{fig:duo_iCKAN_VHB4910_4905_explicit_loss}(b). 

\begin{table}[H]
\centering
\hspace{5em}
\begin{minipage}{0.35\textwidth}
\centering
\small
\textbf{Training}
\vspace{0.3em}

\begin{NiceTabular}{|c|c|c|c|c|c|c|}
\hline
& \multicolumn{6}{c}{Temperature $\theta$ [°C]}\\
\cline{2-7}
$\dot F_{11}$ [s$^{-1}$] &0 &10 & 20 & 40 & 60 & 80\\
\hline
0.03 & $\triangle$ &  -- & $\triangle$ &  $\triangle$ &  $\triangle$ &  $\triangle$\\
0.05 & -- & -- & -- & --& -- & -- \\
0.10 & $\triangle$ &  $\triangle$ & $\triangle$ &  $\triangle$ &  $\triangle$ & $\triangle$\\
\hline
\end{NiceTabular}
\end{minipage}
\hfill
\begin{minipage}{0.35\textwidth}
\centering
\small
\textbf{Testing}
\vspace{0.3em}

\begin{NiceTabular}{|c|c|c|c|c|c|c|}
\hline
& \multicolumn{6}{c}{Temperature $\theta$ [°C]}\\
\cline{2-7}
$\dot F_{11}$ [s$^{-1}$] &0 &10 & 20 & 40 & 60 & 80\\
\hline
0.03 & -- & -- & -- & --& -- & -- \\
0.05 & $\triangle$ & -- & $\triangle$ & $\triangle$ & $\triangle$ & $\triangle$ \\
0.10  & $\lozenge$$\square$ & $\square$ & $\square$ & $\lozenge$$\square$ & $\lozenge$$\square$ & $\lozenge$$\square$\\
\hline
\end{NiceTabular}
\end{minipage}
\hspace{5em}
\vspace{0.5em}

$\lozenge$: $\lambda_{\max}=2$, \quad
$\square$: $\lambda_{\max}=3$, \quad
$\triangle$: $\lambda_{\max}=4$

\caption{Assignment of training and testing datasets of VHB 4905 polymer at $\theta \in [0,80]\,^\circ \mathrm{C}$.}
\label{tab:split_VHB4905}
\end{table}
\color{black}

\textbf{Symbolification.} We symbolify the trained activations, except for those corresponding to the temperature feature in the first layer of each KAN for elastic energy, using the same library of convex, monotonic functions as in the previous section. The resulting KAN architectures are shown in \autoref{fig:KANs_VHB4905}. For each branch, there remains an activation function associated with temperature in the free energy $^i\psi(^i\bar{\vec{C}}_e)$, which we denote as $^1g(\theta)$ and $^2 g(\theta)$ for the first and second branch, respectively. 

\begin{figure}[H]
    \centering
    \vspace{1em}
    \input{standalone_iCKAN_results_VHB4905}
    \caption{Discovered model for VHB 4905 thermo polymer at $\theta \in [0,80]\,^\circ \mathrm{C}$ }
    \label{fig:KANs_VHB4905}
\end{figure}

As these activations are not constrained to be monotonic or convex, their shapes can be arbitrary. To capture this flexibility, we adopt higher-order polynomial functions. This approach allows us to automatically discover a symbolic expression for $^1g(\theta)$. For $^2g(\theta)$, however, a single symbolifc function did not provide sastisfactory fit. Therfore, we allow for piecewise symbolification, and manually approximate it using two piecewise quadratic functions defined over the temperature intervals $[0,40)$ and $[40,80]\, ^\circ \mathrm{C}$, respectively. The coefficients are chosen such that $C^1$-continuity is enforced at $\theta = 40\, ^\circ \mathrm{C}$. 

\textbf{Results.} The predicted results of the symbolified iCKAN model are shown in \autoref{fig:duo_iCKAN_feat_VHB4905_explicit_result_0_80}. 
The symbolified formula for the discoverd model are listed in \autoref{fig:symb_funcs_KANs_0_80}. 
\color{author}
Analogous to the previous example, a quantitative comparison of the model performance before and after the symbolification process is provided in \autoref{tab:results_VHB}. While the symbolification step introduces a modest increase in prediction error, the resulting symbolic expressions still capture the overall material behavior and preserve the main constitutive trends. 

The influence of network depth was also investigated during model development. Increasing the depth of the KAN architecture did not lead to a significant improvement in the fit to the experimental data, while considerably increasing the computational cost and reducing the interpretability of the resulting symbolic expressions. Therefore, the selected architecture was considered to provide a favorable balance between accuracy, computational efficiency, and interpretability. While deeper network variants were explored during preliminary experiments, no systematic hyperparameter study was conducted, as the primary objective of this work is not architectural optimization but the identification of interpretable constitutive relations.

This modeling choice also enables a clearer interpretation of the learned constitutive structure, particularly with respect to feature dependencies. In this context, the trained iCKAN consistently learns a decreasing dependence of the free energy on temperature. This observation is in agreement with established theoretical and experimental knowledge for polymeric materials, where increasing temperature generally leads to a reduction in stiffness and stored elastic energy, within the temperature range covered by the training data. The result demonstrates that the proposed feature-augmented iCKAN is capable of identifying physically meaningful feature dependencies directly from the experimental data. Symbolification of the learned activations provides interpretable functional forms, revealing how temperature influences the free energies in each branch.

\begin{figure}[H]
    \centering\small
    \input{table_symb_VHB4905}
    \vspace{.2em}
    \caption{Discovered free energy and dissipation potential functions for VHB 4905 thermo polymer at $\theta \in [0,80]\,^\circ \mathrm{C}$.}
    \label{fig:symb_funcs_KANs_0_80}
\end{figure}

\begin{figure}[H]
    \centering
    \begin{subfigure}{\textwidth}
        \centering
        \includegraphics{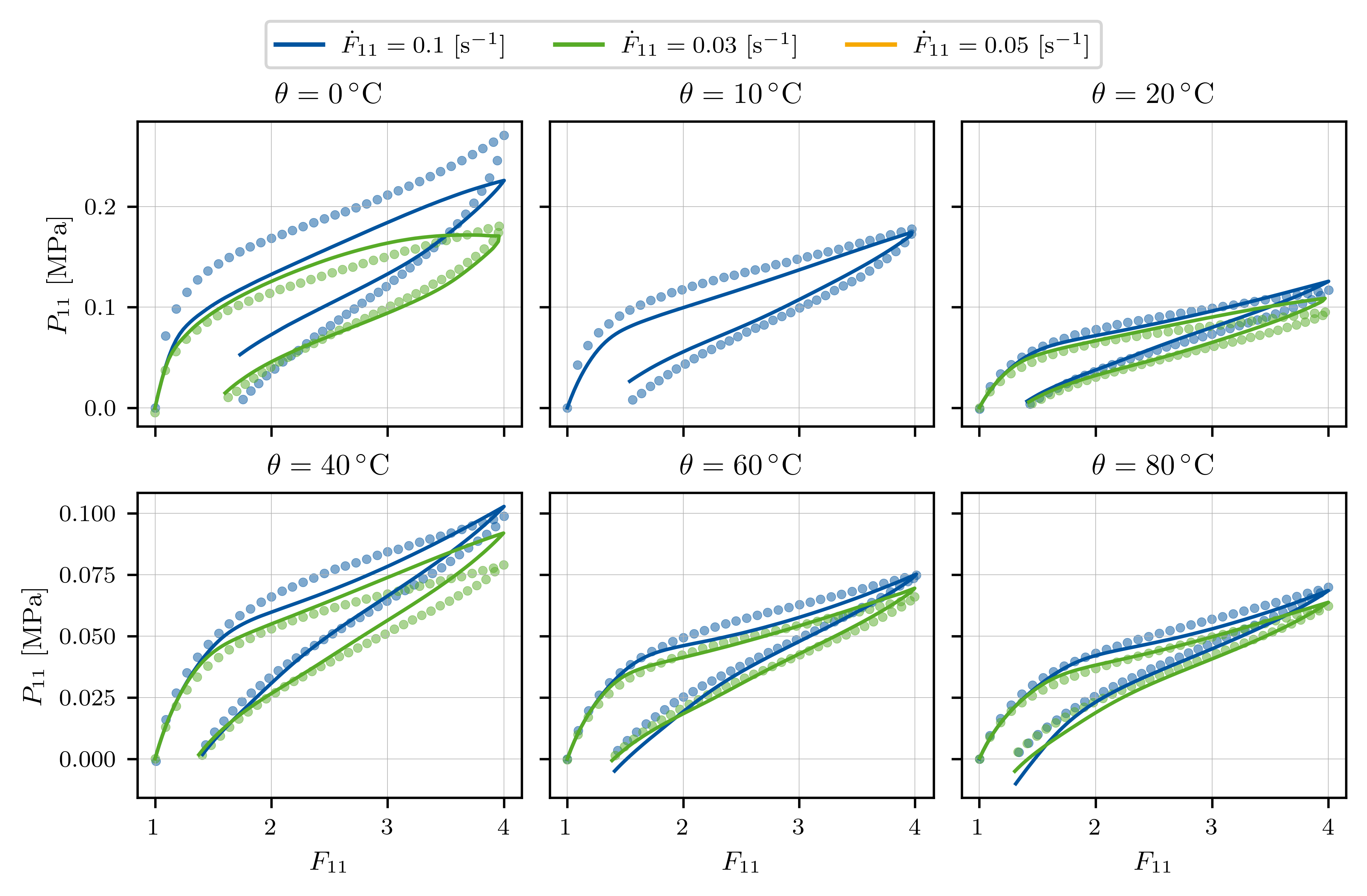}
        \put(-250,-5){(a) \textbf{Training data}}
    \end{subfigure}
    \begin{subfigure}{\textwidth}
        \centering
        \includegraphics{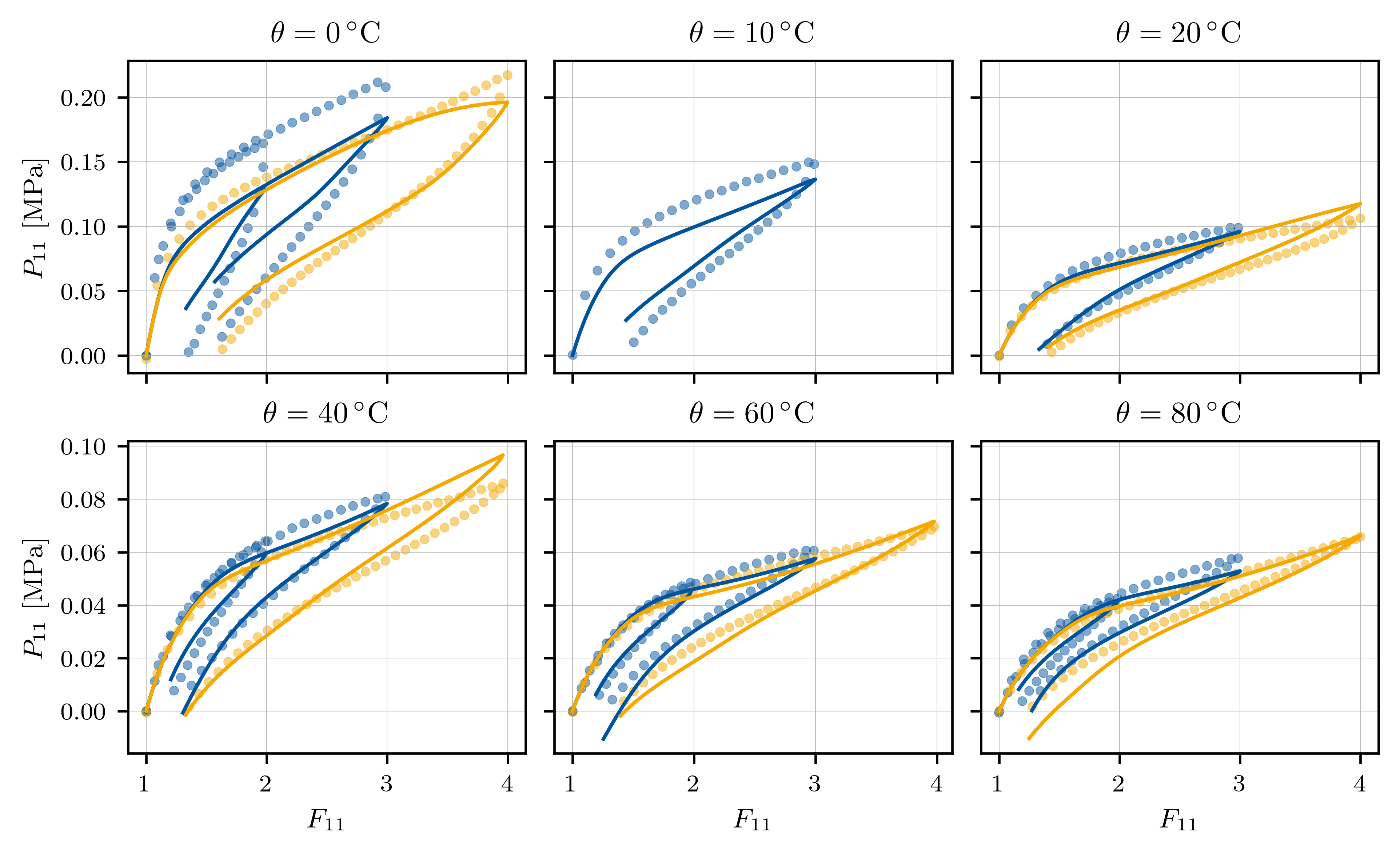}
        \put(-250,-5){(b) \textbf{Testing data}}
    \end{subfigure}
    \caption{(a) Training set and (b) validation set of discovered model for the experimental data of VHB 4905 polymer at $\theta \in [0,80]\,^\circ \mathrm{C}$ under three different constant loading/unloading rates, denoted as $\dot{F}_{11}$. The prediction of the trained iCKAN model is represented by the solid lines and the experimentally measured data is represented by the dots in the corresponding colors.}
\label{fig:duo_iCKAN_feat_VHB4905_explicit_result_0_80}
\end{figure}

\subsubsection{Single-temperature iCKAN analysis at $\theta = 20 \,^\circ \mathrm{C}$}
While it is evident that increased generality of the feature space can slightly reduce prediction performance. Therefore, we employ a dedicated iCKAN model exclusively on the temperature level of $\theta = 20 \,^\circ \mathrm{C}$. In addition, the rheological architecture is extended by incorporating an equilibrium branch in parallel with the Maxwell branches. This equilibrium contribution is implemented by omitting the dissipation potential $^3\omega(^3\bar{\vec{\Sigma}})$ from the this branch. Consequently, the third branch reduces to a purely elastic CKAN model \citep{abdolazizi2025constitutive}. The total first Piola-Kirchhoff stress is obtained as \mbox{$\vec{P}= {}^1\vec{P} + {}^2\vec{P}+ {}^3\vec{P}$}.

For all KANs employed in the network, a topology of $[3,3,1]$ is adopted. The dataset is resampled and partitioned into training and testing sets following the procedure of \citet{kalina2026physics}. The initial KAN grid ranges are selected analogously to the previous examples and subsequently extended by 50\% on both sides to accommodate the larger variation in the experimental data. We choose a L1 regularization coefficient of $\lambda_\text{L1} = 10^{-4}$ and gradient clipping at 0.1 are applied. Training is performed using an initial learning rate of $5\cdot10^{-5}$ together with a cyclic scheduler between $5\cdot10^{-5}$ and $5\cdot10^{-4}$ with a step size of 50 iterations. In particular, the low prediction errors obtained for the dedicated $\theta = 20 \,^\circ \mathrm{C}$ model demonstrate that iCKANs are capable of accurately representing the constitutive response when the feature space is restricted.

The result is shown in \autoref{fig:duo_iCKAN_feat_VHB4905_explicit_result_temp20} with a stress loss $\mathcal{L}_\text{stress}$ for training and testing sets of $9.98 \cdot 10^{-5}$ and $1.5 \cdot 10^{-4}$, respectively. These results support the argument that the observed discrepancies in Section~\ref{sec:result_VHB4905} are not necessarily indicative of an intrinsic limitation of KAN-based constitutive models, but rather reflect the trade-off between interpretability, robustness, and predictive accuracy inherent in the chosen architecture. Nevertheless, \citet{samadi2024smooth} shows that the equivalence of KANs to universal MLPs is only guaranteed under certain circumstances, highlighting potential limitations of KAN-based potentials. 

\begin{figure}[H]
    \centering
    \small
    \includegraphics{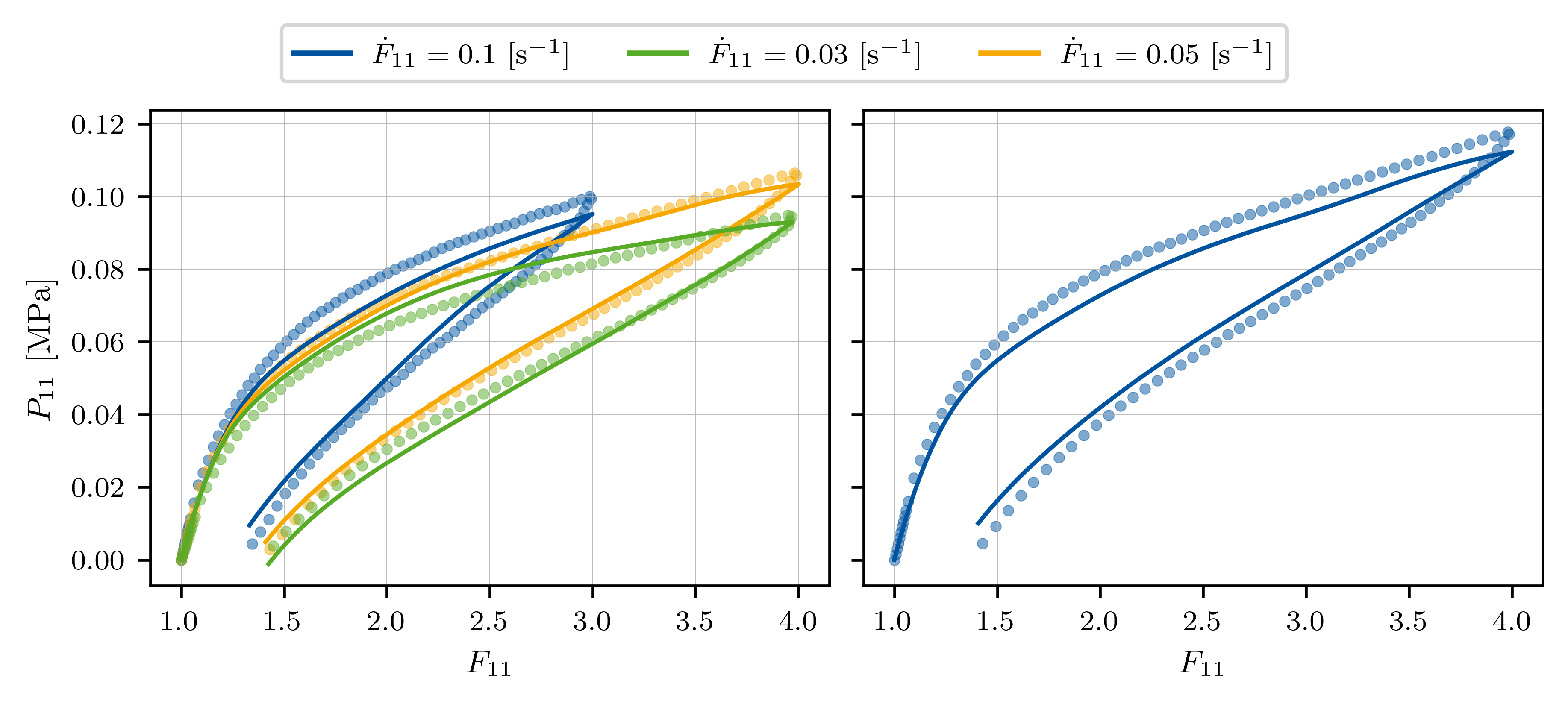}
    \put(-260,50){\fcolorbox{black}{white}{\textbf{Training}}}
    \put(-50,50){\fcolorbox{black}{white}{\textbf{Testing}}}
    \caption{Training set and testing set of discovered model for the experimental data of VHB 4905 polymer at \mbox{$\theta = 20\, ^\circ \mathrm{C}$} under three different constant loading/unloading rates, denoted as $\dot{F}_{11}$. The prediction of the trained iCKAN model is represented by the solid lines and the experimentally measured data is represented by the dots in the corresponding colors.}
\label{fig:duo_iCKAN_feat_VHB4905_explicit_result_temp20}
\end{figure}

\color{black}


%% file: standalone_iCKAN_symb_synthetic.tex
\small
\vspace{2em}
    \hfill
    \subfloat[\textbf{Symbolification for free energy}\label{subfig-1:psi}]{
    \vspace{1em}
    \begin{overpic}[width=.12\textwidth]{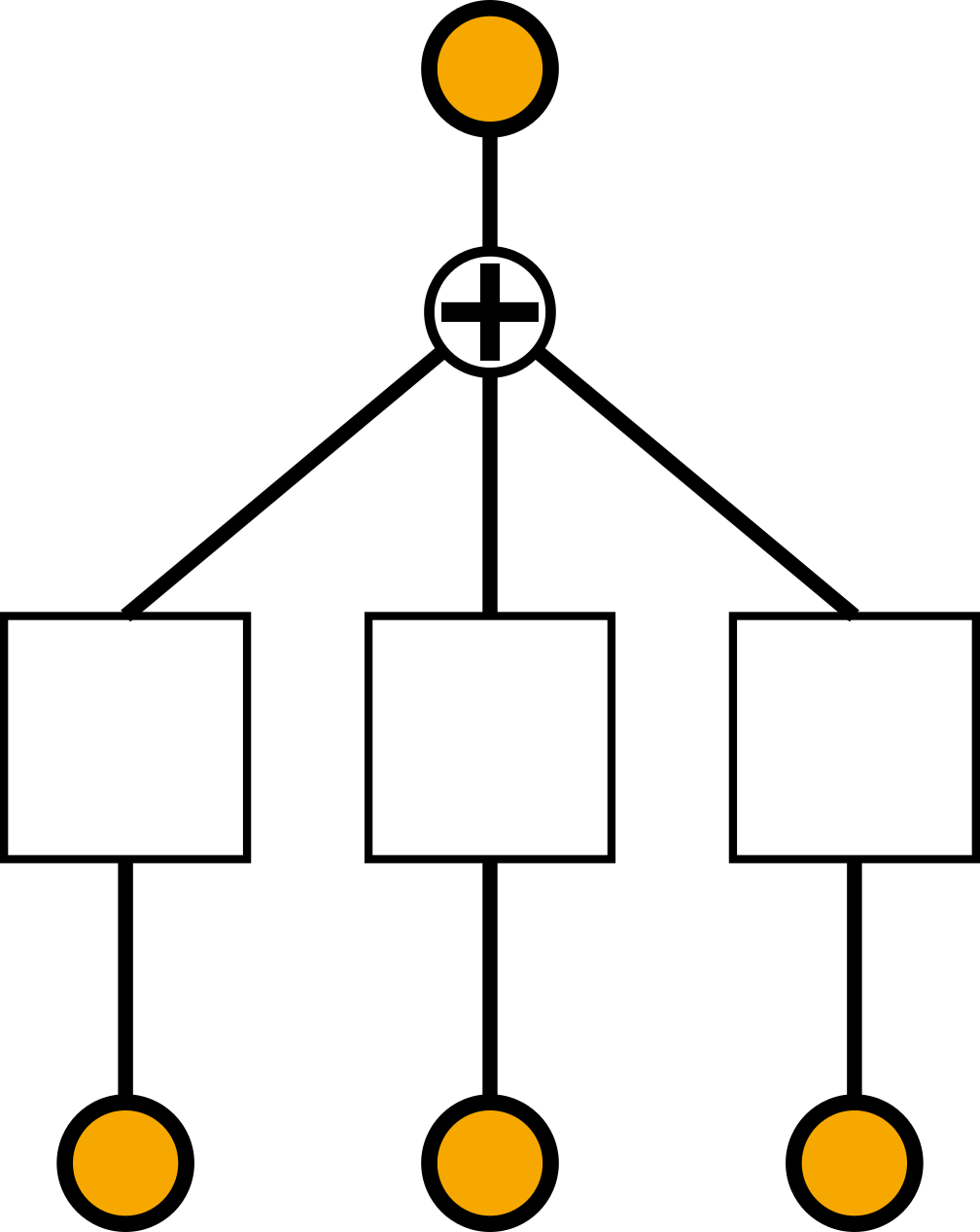}
        \put(-1,28){\includegraphics[width=0.037\textwidth]{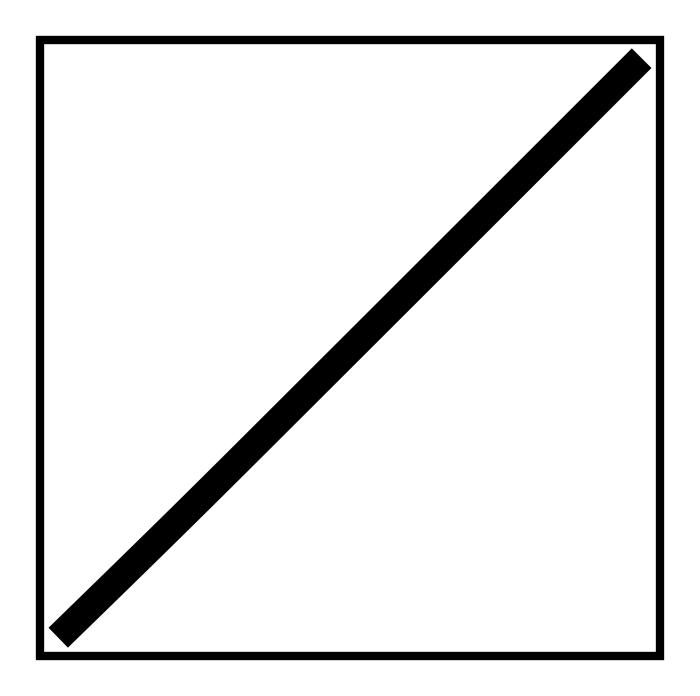}}
        \put(29,28){\includegraphics[width=0.037\textwidth]{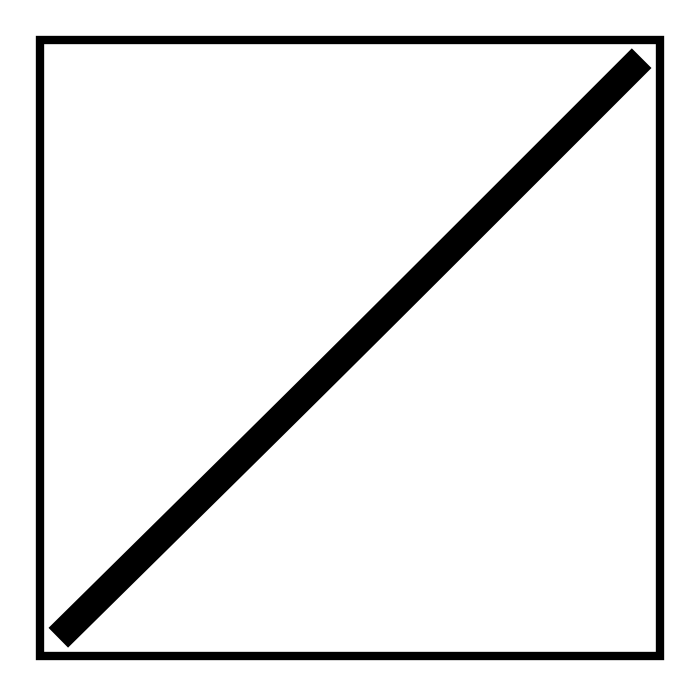}}
        \put(58,28){\includegraphics[width=0.037\textwidth]{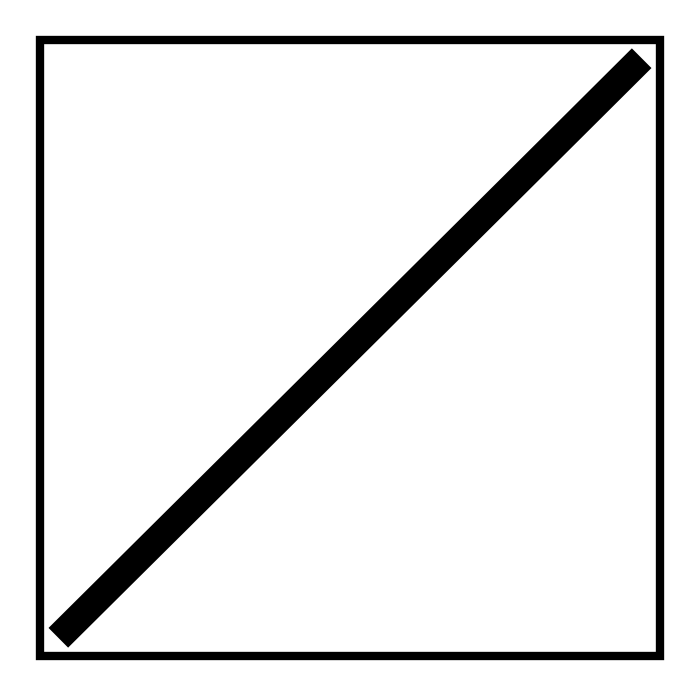}}
        \put(2,-18){$\hat{I}_1^{\bar{\vec{C}}_e}$}
        \put(32,-18){$\hat{I}_2^{\bar{\vec{C}}_e}$}
        \put(63,-18){$\hat{I}_3^{\bar{\vec{C}}_e}$}
        \put(32,107){$\psi^{\text{KAN}}$}
    \end{overpic}
    \begin{minipage}[c]{0.1\textwidth}
        \centering
        \vspace{-3em}
        \raisebox{2.5em}{\large\boldmath$\overset{\text{symb.}}{\longrightarrow}$}
    \end{minipage}
    \begin{overpic}[width=.12\textwidth]{figures_KAN_31.png}
        \put(0,28){\includegraphics[width=0.037\textwidth]{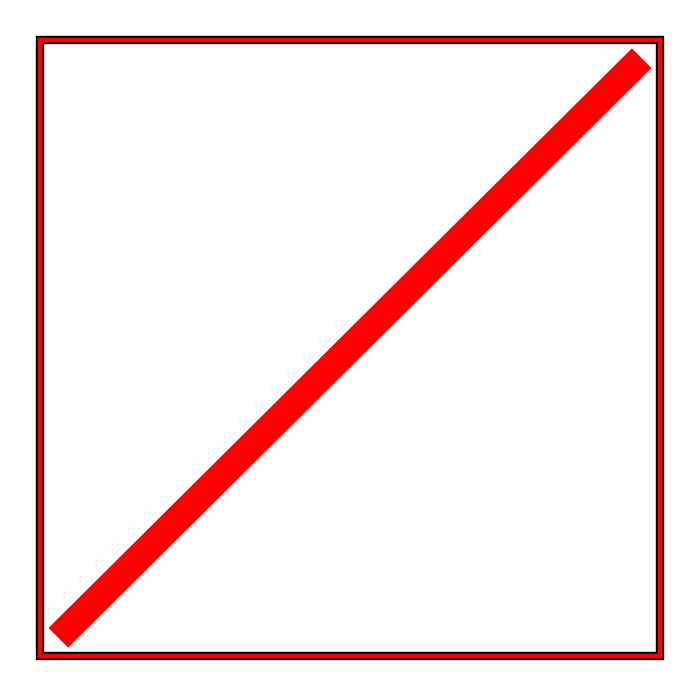}}
        \put(29,28){\includegraphics[width=0.037\textwidth]{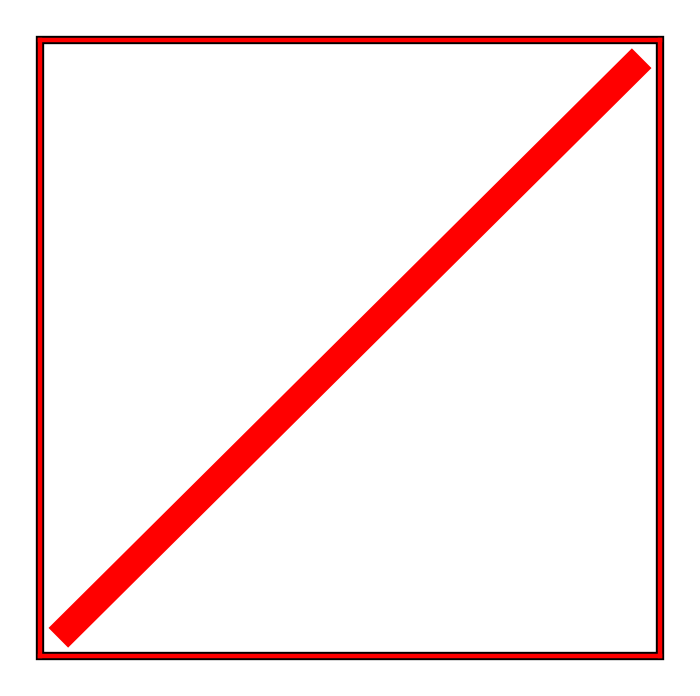}}
        \put(58,28){\includegraphics[width=0.037\textwidth]{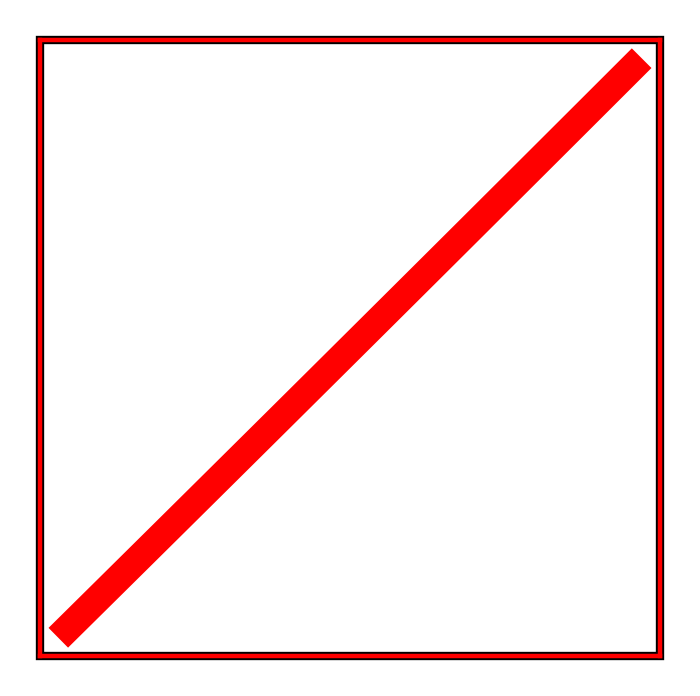}}
        \put(2,-18){$\hat{I}_1^{\bar{\vec{C}}_e}$}
        \put(32,-18){$\hat{I}_2^{\bar{\vec{C}}_e}$}
        \put(63,-18){$\hat{I}_3^{\bar{\vec{C}}_e}$}
        \put(32,107){$\psi^{\text{KAN}}$}
    \end{overpic}
    \vspace{1em}
    }
    \hspace{3em}
    \subfloat[\textbf{Symbolification for dissipation potential}\label{subfig-2:omega}]{
    \vspace{1em}
    \begin{overpic}[width=.12\textwidth]{figures_KAN_31.png}
        \put(0,28){\includegraphics[width=0.037\textwidth]{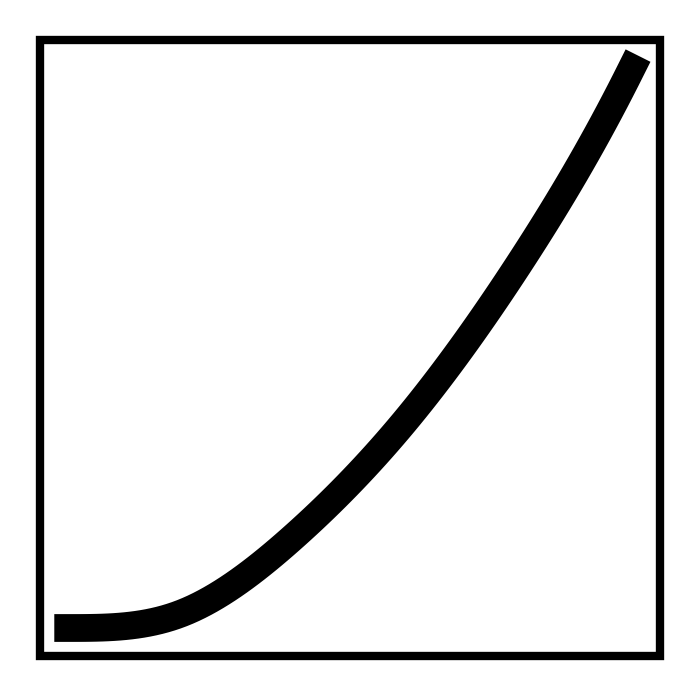}}
        \put(29,28){\includegraphics[width=0.037\textwidth]{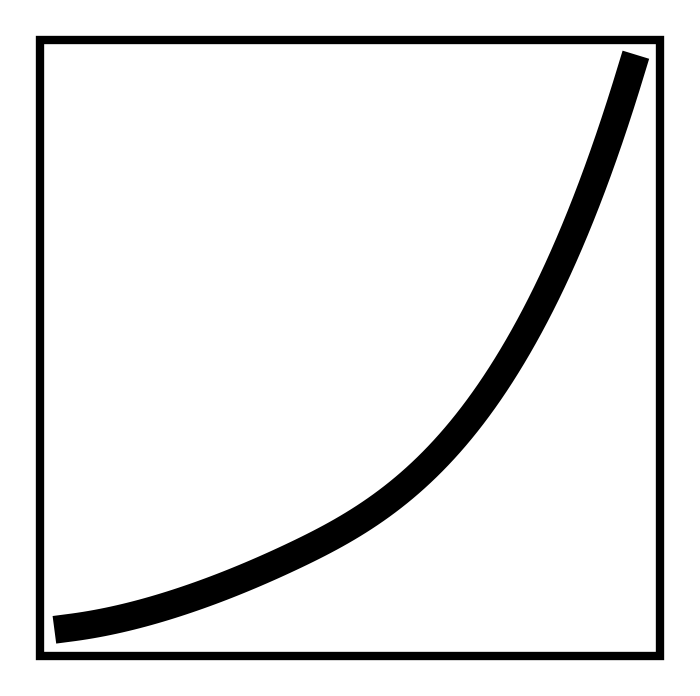}}
        \put(58,28){\includegraphics[width=0.037\textwidth]{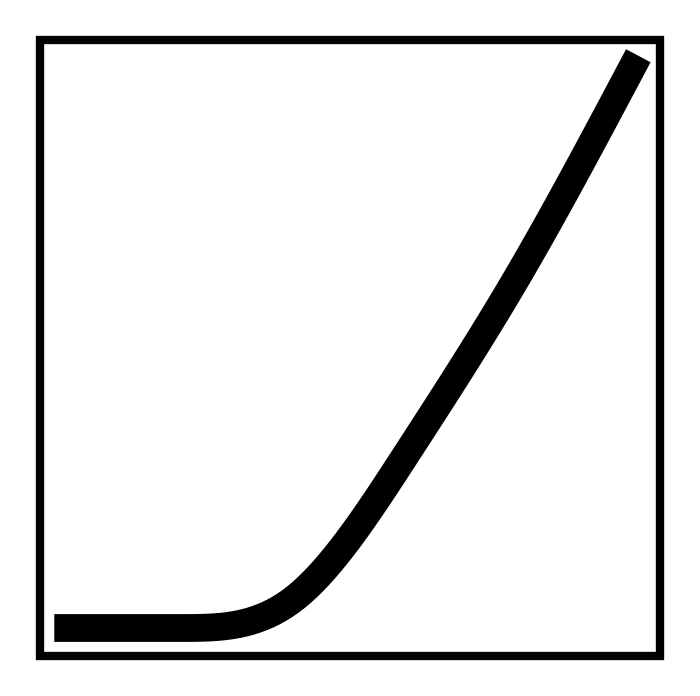}}
        \put(2,-18){$\hat{I}_1^{\bar{\vec{\Sigma}}}$}
        \put(28,-18){$\hat{J}_2^{\bar{\vec{\Sigma}}}$}
        \put(60,-18){$\hat{J}_3^{\bar{\vec{\Sigma}}}$}
        \put(32,107){$\omega^{\text{KAN}}$}
    \end{overpic}
    \begin{minipage}[c]{0.1\textwidth}
        \centering
        \vspace{-3em}
        \raisebox{2.5em}{\large\boldmath$\overset{\text{symb.}}{\longrightarrow}$}
    \end{minipage}
    \begin{overpic}[width=.12\textwidth]{figures_KAN_31.png}
        \put(0,28){\includegraphics[width=0.037\textwidth]{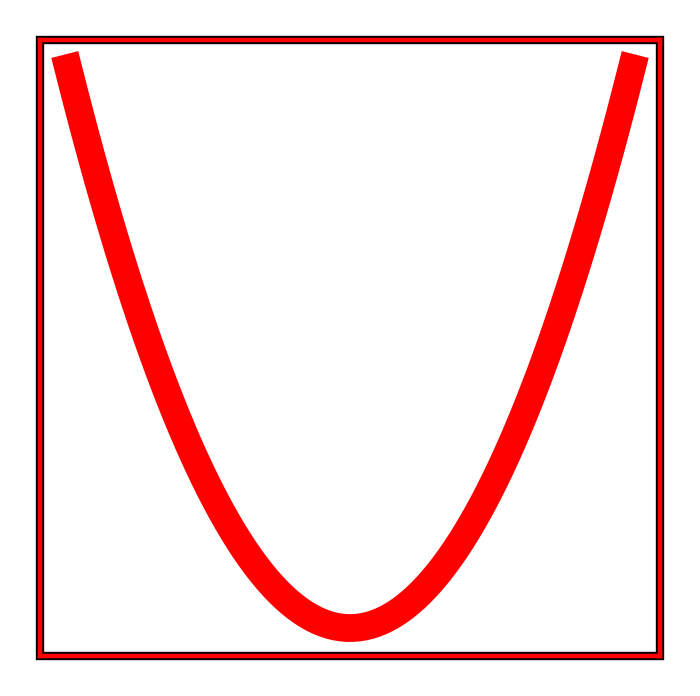}}
        \put(29,28){\includegraphics[width=0.037\textwidth]{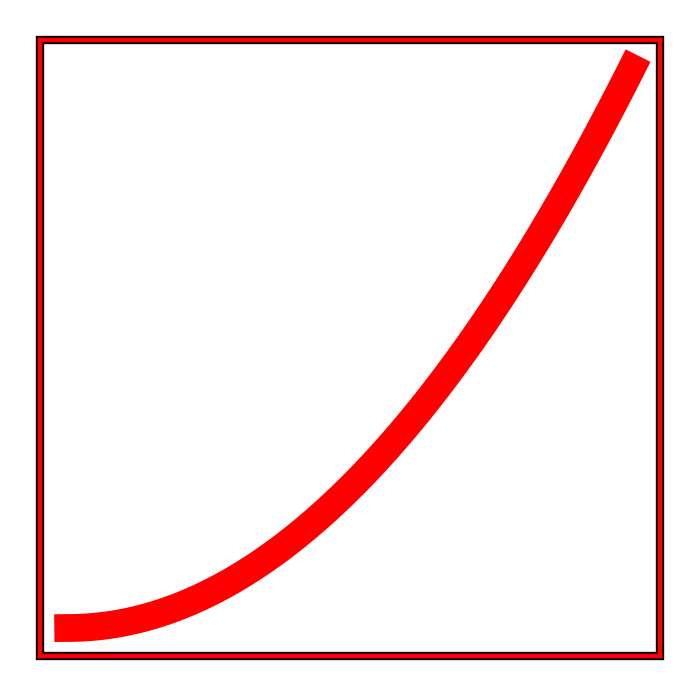}}
        \put(59,28){\includegraphics[width=0.037\textwidth]{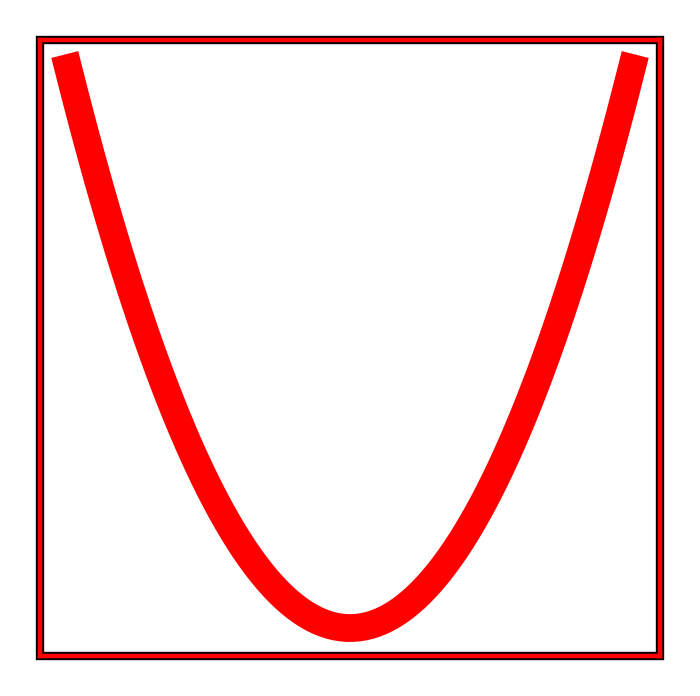}}
        \put(2,-18){$\hat{I}_1^{\bar{\vec{\Sigma}}}$}
        \put(28,-18){$\hat{J}_2^{\bar{\vec{\Sigma}}}$}
        \put(60,-18){$\hat{J}_3^{\bar{\vec{\Sigma}}}$}
        \put(32,107){$\omega^{\text{KAN}}$}
    \end{overpic}
    \vspace{1em}
    }
    \hfill

%% file: table_symb_synthetic.tex
\begin{tabularx}{\textwidth}{p{6.5cm}X}
\hline
Free energy &
$\psi(\bar{\vec{C}}_e) =
a \cdot \hat{I}_{1}^{\bar{\vec{C}}_e}
+ b \cdot \hat{I}_{2}^{\bar{\vec{C}}_e}
+ c \cdot \hat{I}_{3}^{\bar{\vec{C}}_e}$ \\
\hline
\end{tabularx}
\begin{tabularx}{\textwidth}{p{3.3cm}p{1.75cm}p{1.75cm}p{1.75cm}p{1.75cm}p{1.75cm}p{1.75cm}}
\multirow{3}{*}{Material parameters} &
\multicolumn{3}{c}{\textbf{Explicit}} & \multicolumn{3}{c}{\textbf{Implicit}} \\
\cmidrule(r{0.5em}){2-4} \cmidrule(l{.5em}){5-7}
& $a$ & $b$ & $c$ & $a$ & $b$ & $c$ \\
& 0.0 & 0.2194 & 0.5027 & 0.2722 & 0.1093 & 0.5255\\
\bottomrule
\end{tabularx}

\vspace{4mm}

\begin{tabularx}{\textwidth}{p{6.5cm}X}
\hline
Dissipation potential &
$\omega(\bar{\vec{\Sigma}}) =
a \cdot (\hat{I}_1^{\bar{\vec{\Sigma}}})^2
+ b \cdot (\hat{J}_2^{\bar{\vec{\Sigma}}})^2
+ c \cdot (\hat{J}_3^{\bar{\vec{\Sigma}}})^2$ \\
\hline
\end{tabularx}
\begin{tabularx}{\textwidth}{p{3.3cm}p{1.75cm}p{1.75cm}p{1.75cm}p{1.75cm}p{1.75cm}p{1.75cm}}
\multirow{3}{*}{Material parameters} &
\multicolumn{3}{c}{\textbf{Explicit}} & \multicolumn{3}{c}{\textbf{Implicit}} \\
\cmidrule(r{0.5em}){2-4} \cmidrule(l{.5em}){5-7}
& $a$ & $b$ & $c$ & $a$ & $b$ & $c$ \\
& 0.109 & 0.9394 & 0.4286 & 0.1045 & 1.2613 & 0.6397\\
\bottomrule
\end{tabularx}

%% file: standalone_iCKAN_results_VHB4910.tex
\centering
\small
\vspace{2em}
\begin{subfigure}{0.2\textwidth}
\centering
\begin{overpic}[height=\textwidth]{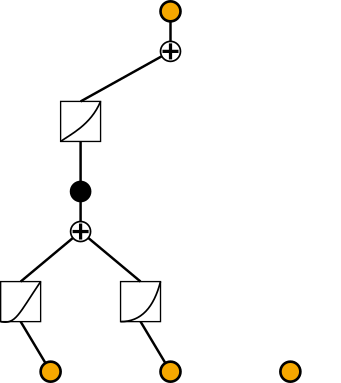}
    \put(-1,15){\includegraphics[height=0.15\textwidth]{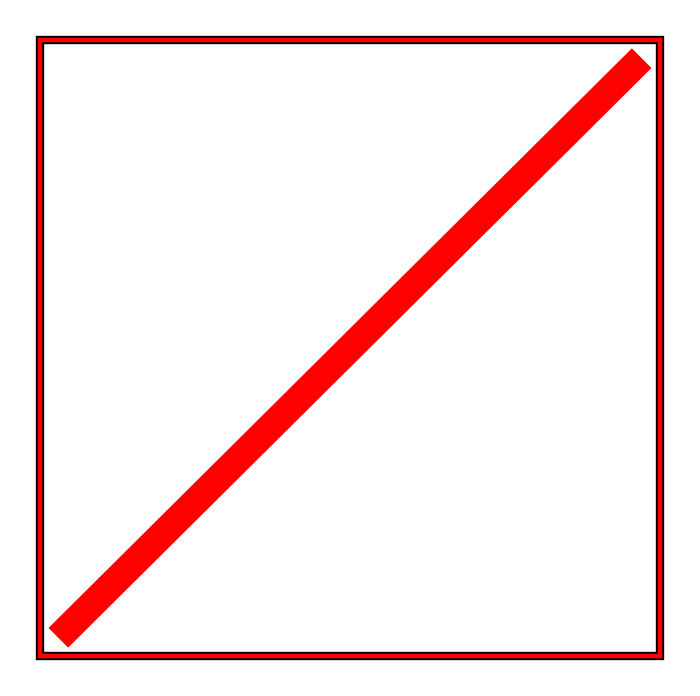}}
    \put(30,15){\includegraphics[height=0.15\textwidth]{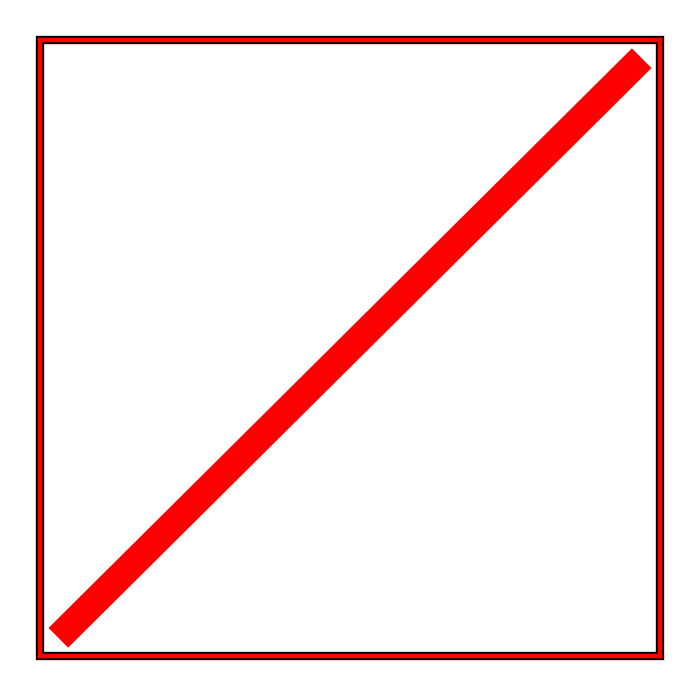}}
    \put(15,62){\includegraphics[height=0.15\textwidth]{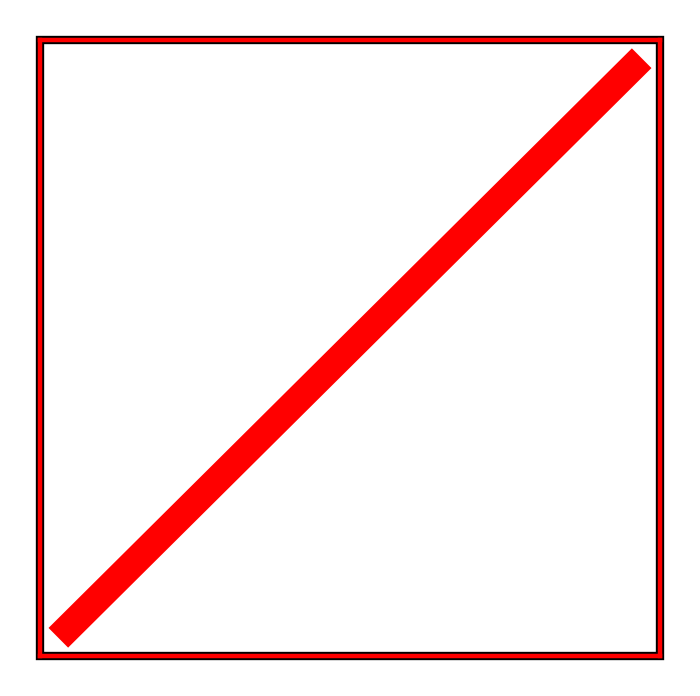}}
    \put(10,-14){$\hat{I}_1^{^1\bar{\vec{C}}_e}$}
    \put(42,-14){$\hat{I}_2^{^1\bar{\vec{C}}_e}$}
    \put(74,-14){$\hat{I}_3^{^1\bar{\vec{C}}_e}$}
    \put(38,107){$^1\psi^\text{KAN}$}
\end{overpic}
\vspace{1.5em}
\caption{$^1\psi(^1\bar{\vec{C}}_e)$}
\end{subfigure}
\hspace{1em}
\begin{subfigure}{0.2\textwidth}
\centering
\begin{overpic}[height=\textwidth]{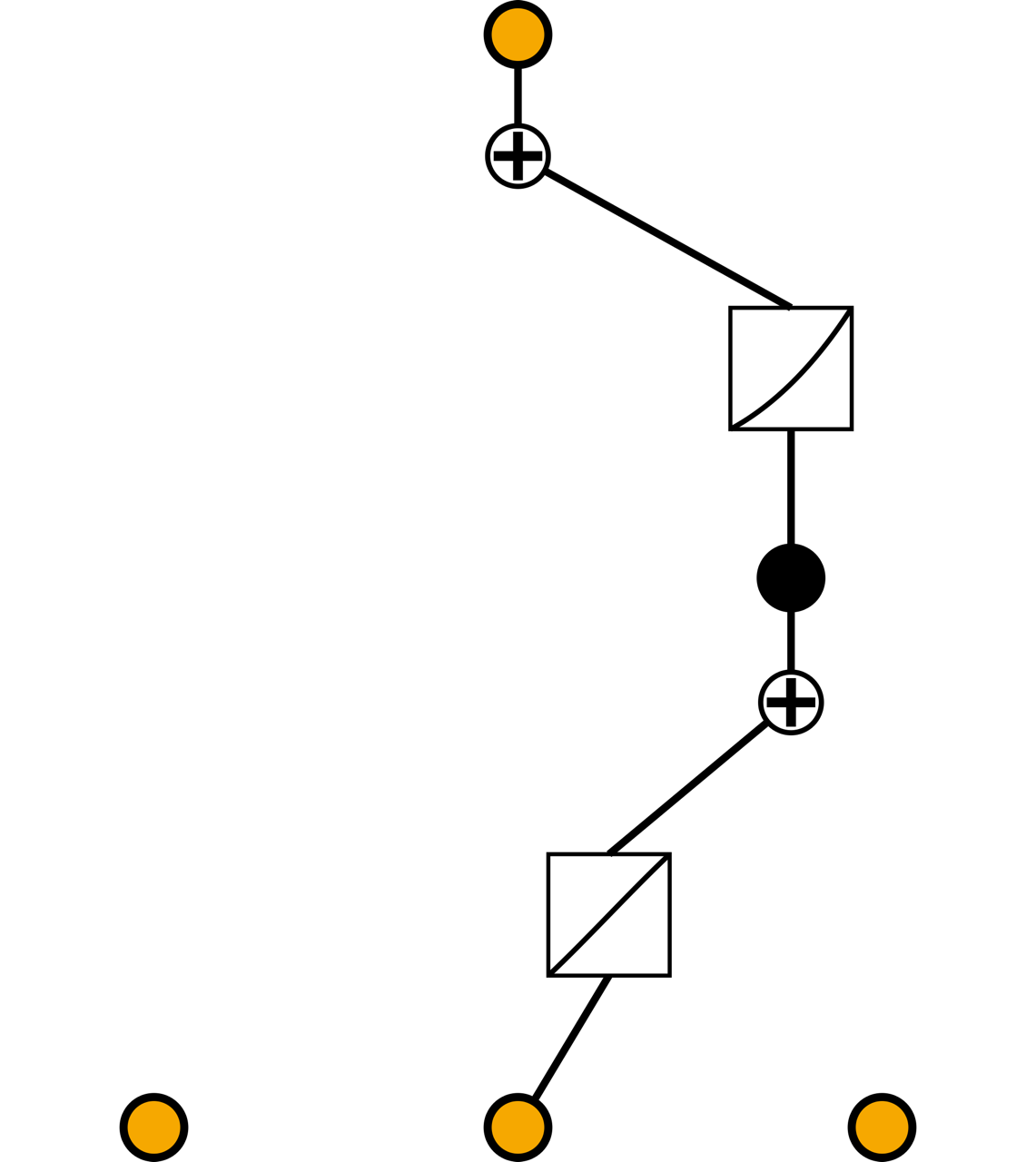}
    \put(45,15){\includegraphics[height=0.15\textwidth]{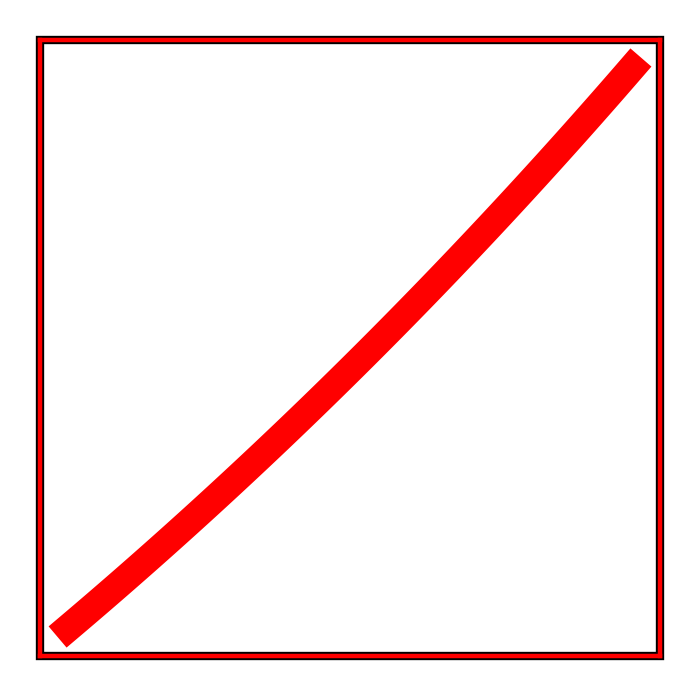}}
    \put(61,62){\includegraphics[height=0.15\textwidth]{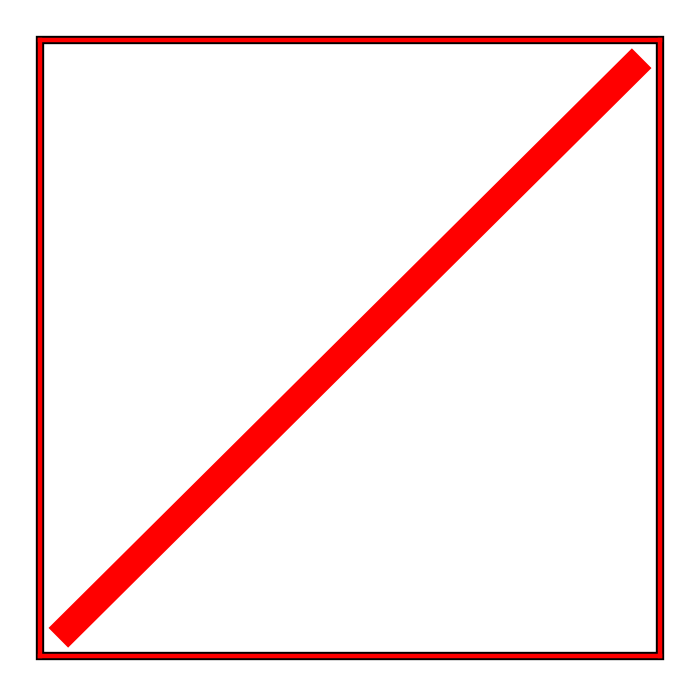}}
    \put(10,-14){$\hat{I}_1^{^1\bar{\vec{\Sigma}}}$}
    \put(42,-14){$\hat{J}_2^{^1\bar{\vec{\Sigma}}}$}
    \put(74,-14){$\hat{J}_3^{^1\bar{\vec{\Sigma}}}$}
    \put(38,107){$^1\omega^\text{KAN}$}
\end{overpic}
\vspace{1.5em}
\caption{$^1\omega(^1\bar{\vec{\Sigma}})$}
\end{subfigure}
\vspace{.5em}
\hspace{1em}
\begin{subfigure}{0.2\textwidth}
\centering
\begin{overpic}[height=\textwidth]{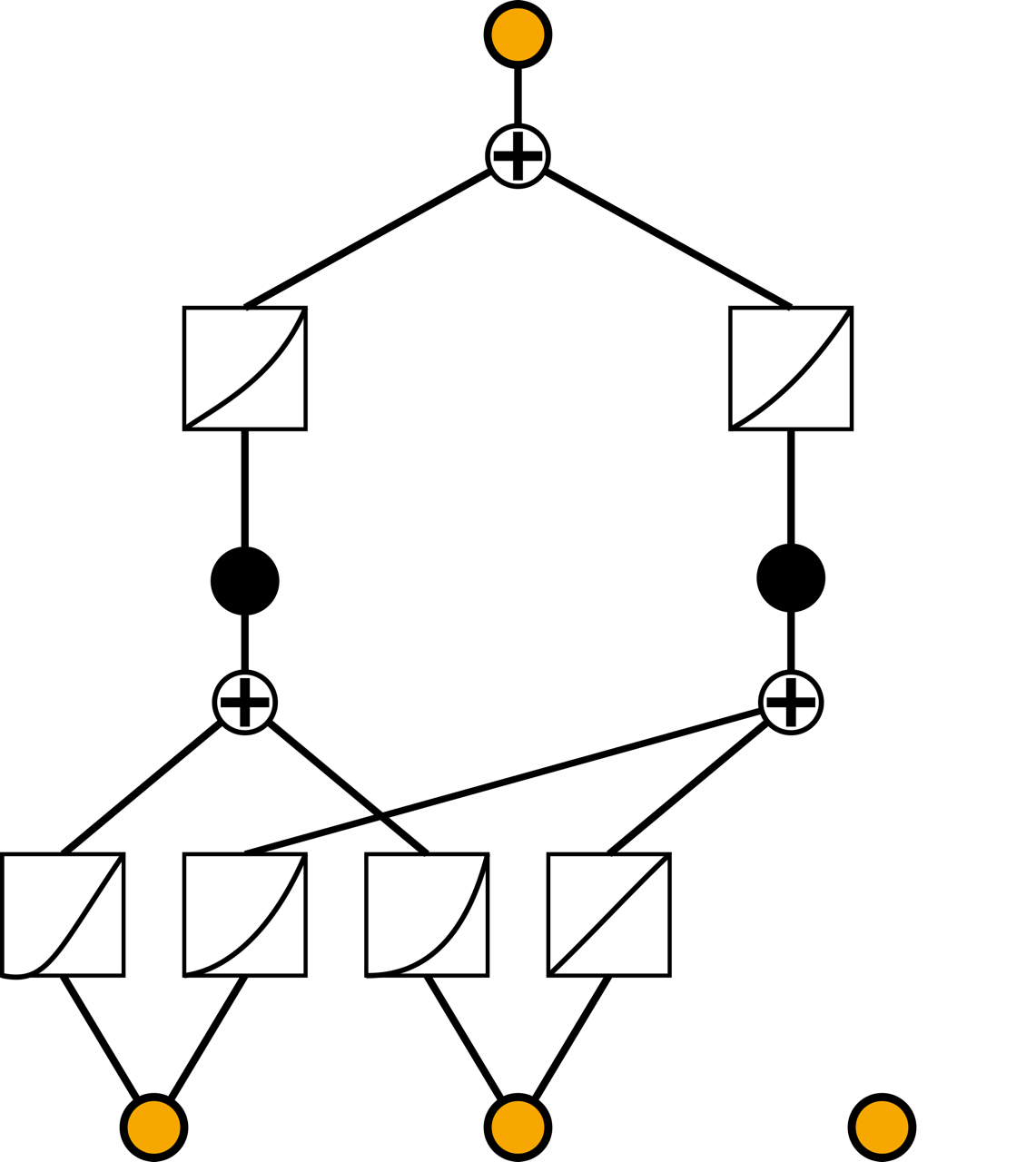}
    \put(-1,15){\includegraphics[height=0.15\textwidth]{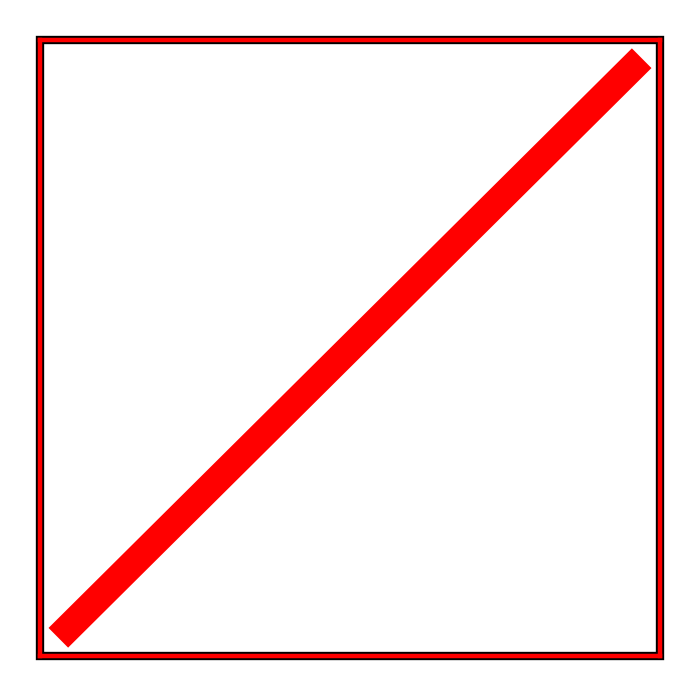}}
    \put(15,15){\includegraphics[height=0.15\textwidth]{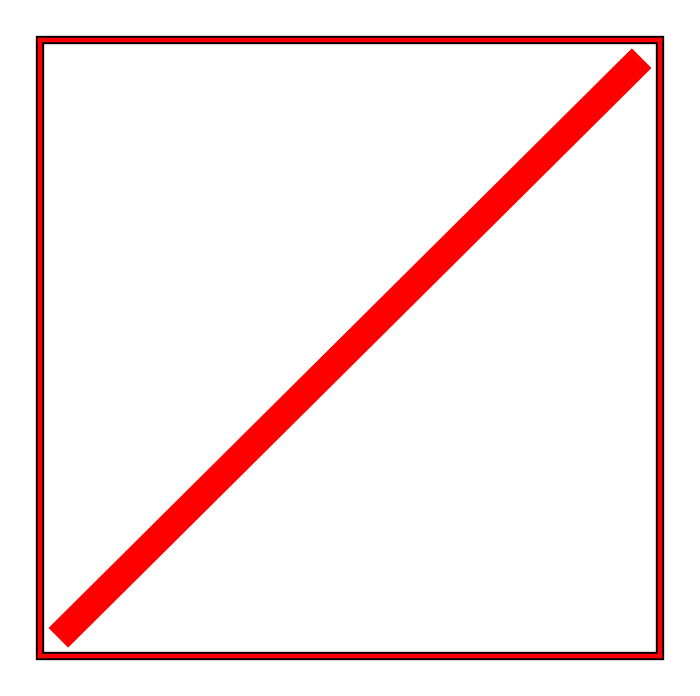}}
    \put(30,15){\includegraphics[height=0.15\textwidth]{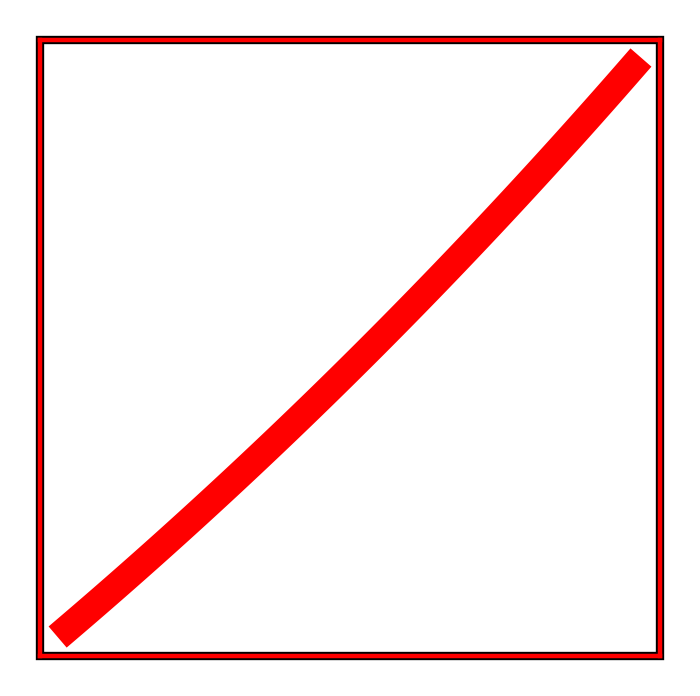}}
    \put(45,15){\includegraphics[height=0.15\textwidth]{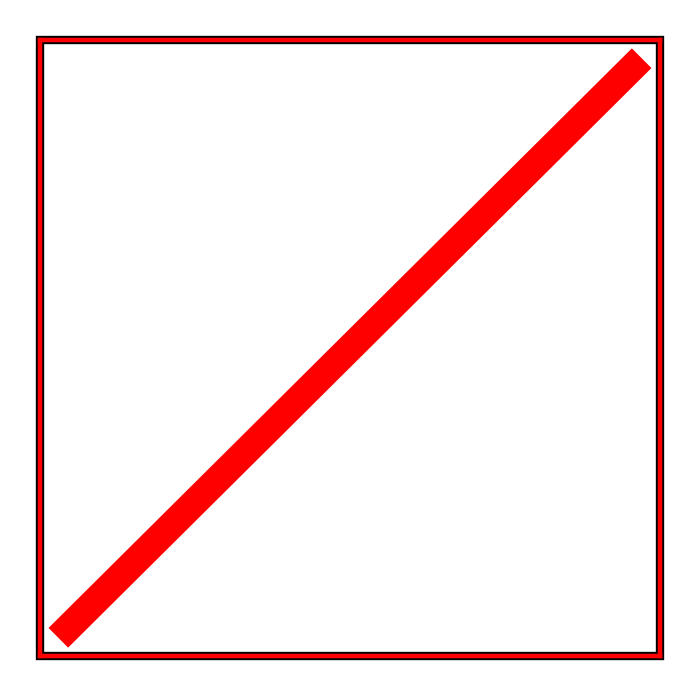}}
    \put(15,62){\includegraphics[height=0.15\textwidth]{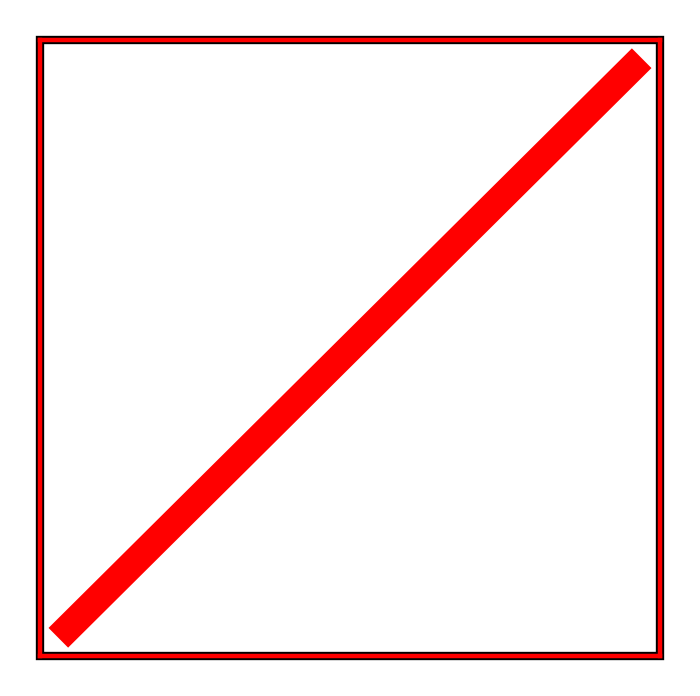}}
    \put(61,62){\includegraphics[height=0.15\textwidth]{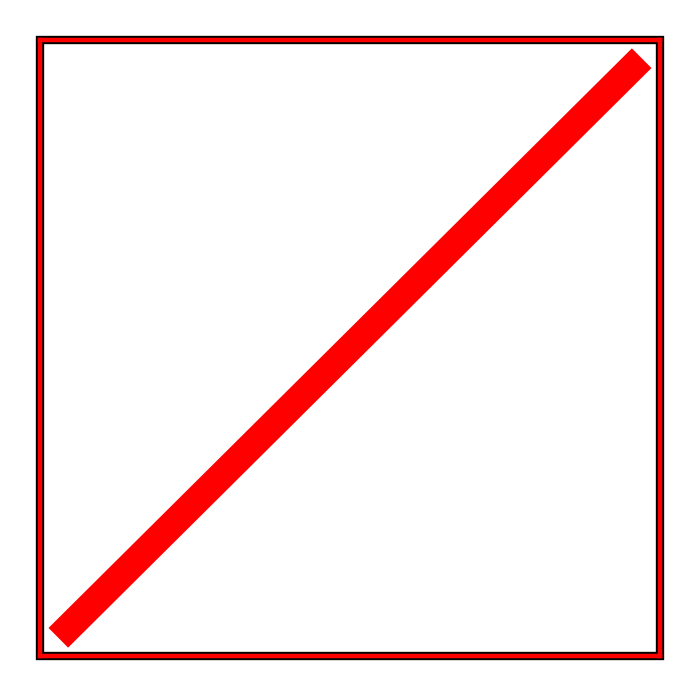}}
    \put(10,-14){$\hat{I}_1^{^2\bar{\vec{C}}_e}$}
    \put(42,-14){$\hat{I}_2^{^2\bar{\vec{C}}_e}$}
    \put(74,-14){$\hat{I}_3^{^2\bar{\vec{C}}_e}$}
    \put(38,107){$^2\psi^\text{KAN}$}
\end{overpic}
\vspace{1.5em}
\caption{$^2\psi(^2\bar{\vec{C}}_e)$}
\end{subfigure}
\hspace{.5em}
\begin{subfigure}{.3\textwidth}
    \centering\small
    \includegraphics[width=.9\textwidth]{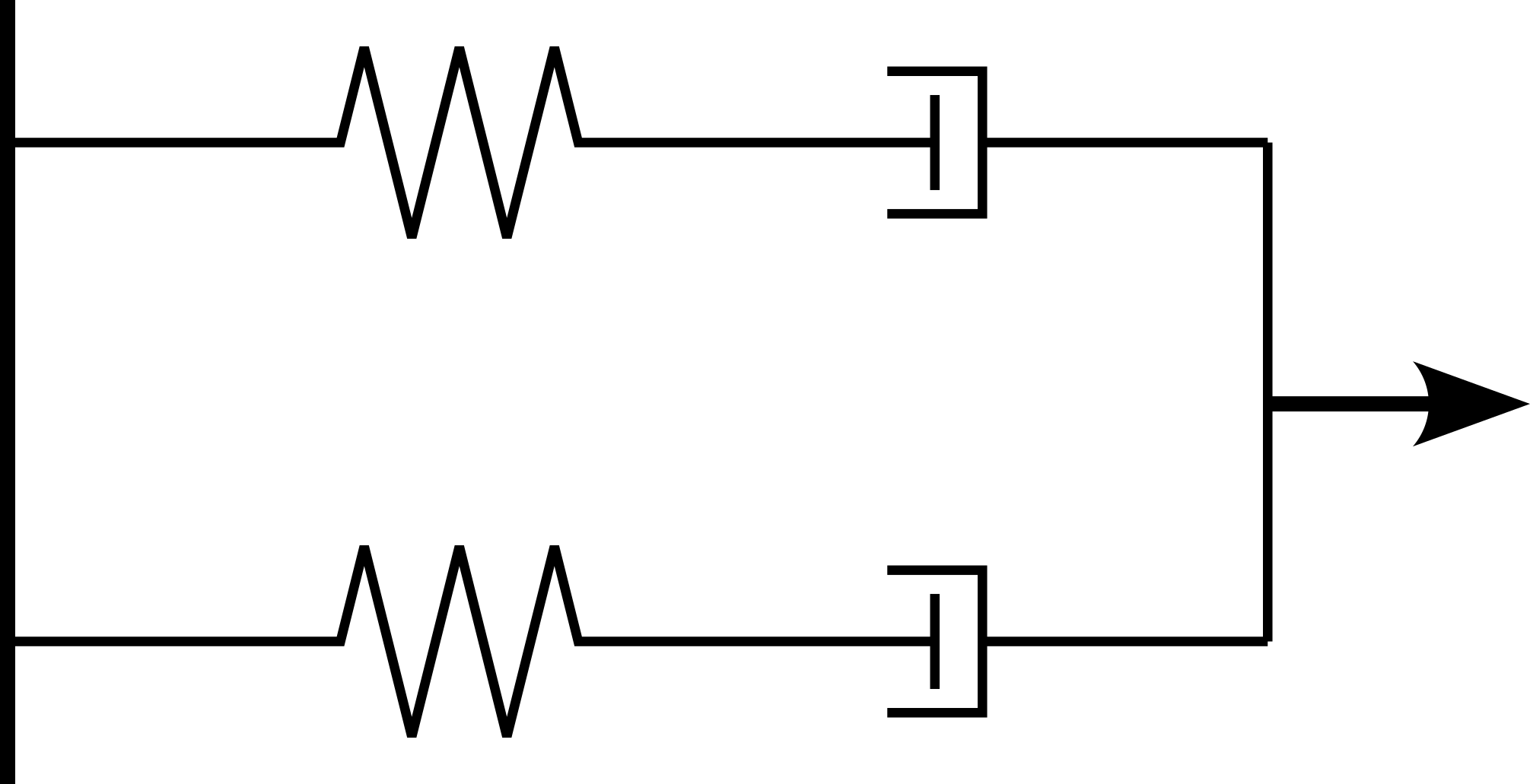}   
    \put(-105,68){$^1\psi(^1\bar{\vec{C}}_e)$}
    \put(-105,25){$^2\psi(^2\bar{\vec{C}}_e)$}
    \put(-60,68){$^1\omega(^1\bar{\vec{\Sigma}})$}
    \put(-60,25){$^2\omega(^2\bar{\vec{\Sigma}})$}
    \put(-4,55){$\vec{P}=$}
    \put(-17,40){${}^1\vec{P} + {}^2\vec{P}$}
    \vspace{4em}

\end{subfigure}

%% file: table_symb_VHB4910.tex
\begin{tabularx}{\textwidth}{p{4cm}X}
\hline
Free energy &
$^1\psi(^1\bar{\vec{C}}_e) = a \cdot \exp(b\cdot \hat{I}_{1}^{^1\bar{\vec{C}}_e})  + c \cdot \hat{I}_{2}^{^1\bar{\vec{C}}_e}  $ \\
\hline
\end{tabularx}
\begin{tabularx}{\textwidth}{p{4cm}p{1cm}p{1cm}p{1cm}}
\multirow{2}{*}{Material parameters} &
$a$ & $b$ &$c$ \\
 & 34.654 & 0.1204 & 3.7902 \\
\hline
\end{tabularx}

\vspace{4mm}

\begin{tabularx}{\textwidth}{p{4cm}X}
\hline
Dissipation potential &
$^1\omega(^1\bar{\vec{\Sigma}}) = \mathcal{H}(a \cdot \exp(b\cdot\exp(c\cdot\hat{J}_2^{^1\bar{\vec{\Sigma}}})))$ \\
\hline
\end{tabularx}
\begin{tabularx}{\textwidth}{p{4cm}p{1cm}p{1cm}p{1cm}}
\multirow{2}{*}{Material parameters} &
$a$ & $b$ & $c$ \\
 & 0.0195 & 0.0002 & 0.036 \\
\hline
\end{tabularx}

\vspace{4mm}

\begin{tabularx}{\textwidth}{p{4cm}X}
\hline
Free energy &
$^2\psi(^2\bar{\vec{C}}_e) = a \cdot \hat{I}_{1}^{^2\bar{\vec{C}}_e} + b\cdot\hat{I}_{2}^{^2\bar{\vec{C}}_e}+ c \cdot \exp(d\cdot\hat{I}_{1}^{^2\bar{\vec{C}}_e})$ \\
\hline
\end{tabularx}
\begin{tabularx}{\textwidth}{p{4cm}p{1cm}p{1cm}p{1cm}p{1cm}}
\multirow{2}{*}{Material parameters} &
$a$ & $b$ & $c$ & $d$\\
 & 2.5911 & 1.8997 & 68.73 & 0.0299 \\
\hline
\end{tabularx}

%% file: standalone_iCKAN_results_VHB4905.tex
\centering

\subfloat[Model of $^1\psi(^1\bar{\vec{C}}_e)$\label{subfig-1:psi1}]{
\vspace{.5em}
\begin{overpic}[height=0.2\textwidth]{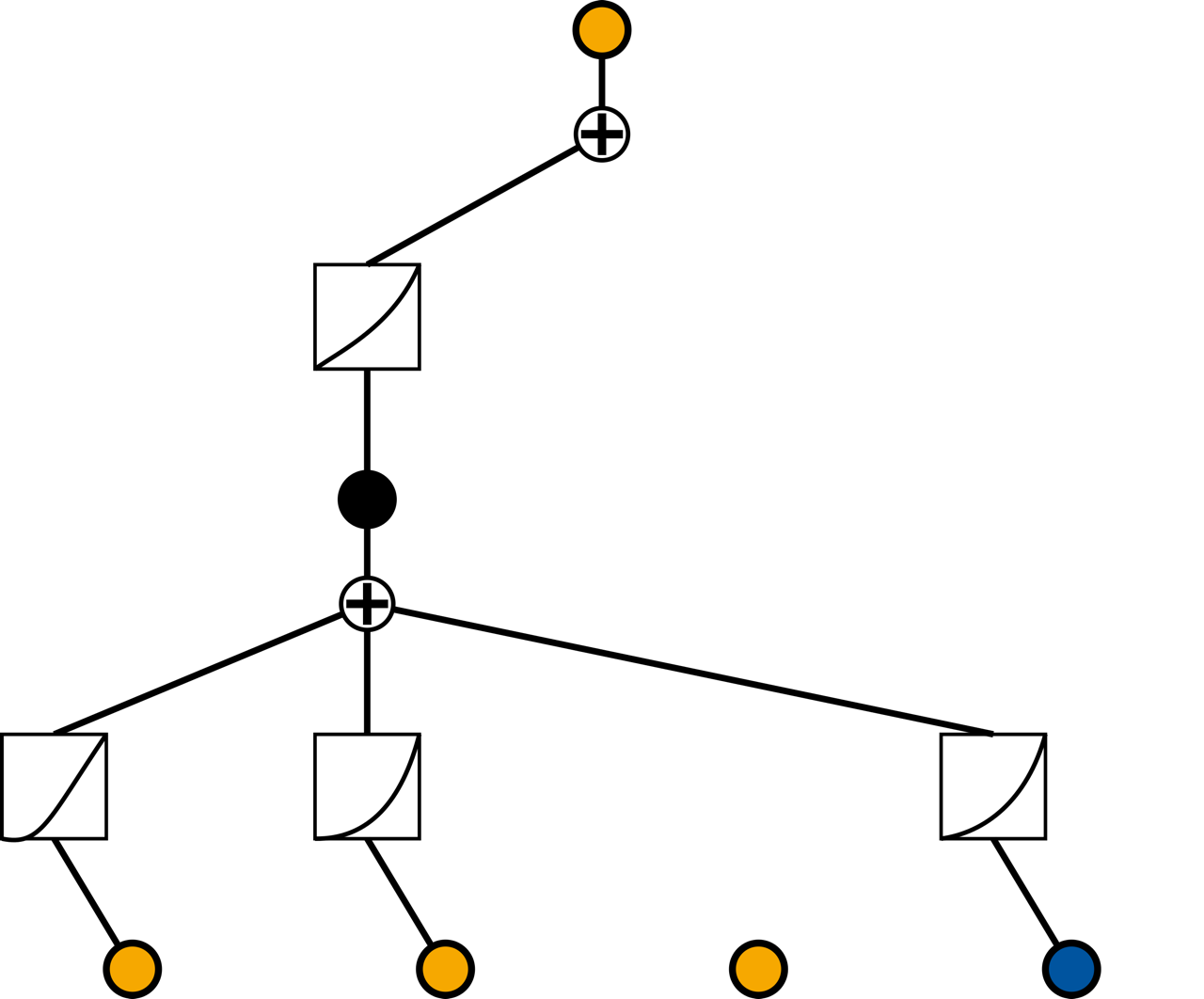}
    \put(-1,12){\includegraphics[height=0.03\textwidth]{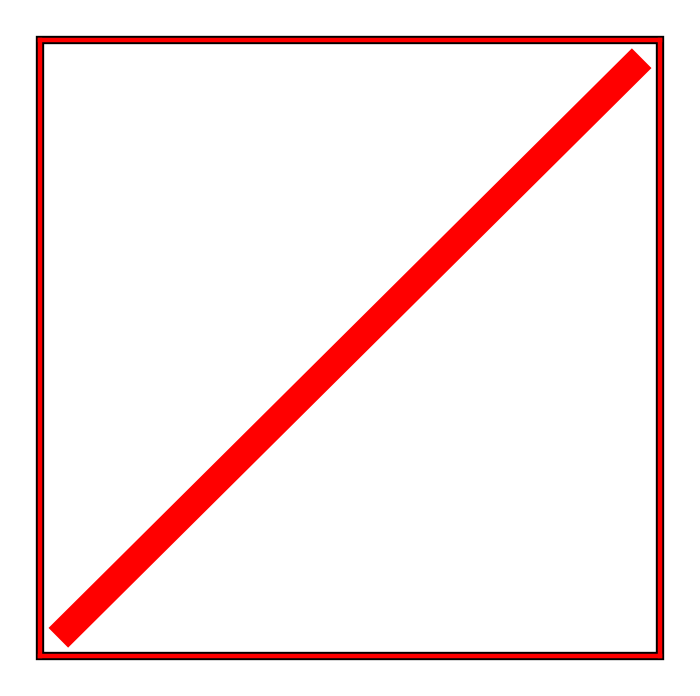}}
    \put(25,12){\includegraphics[height=0.03\textwidth]{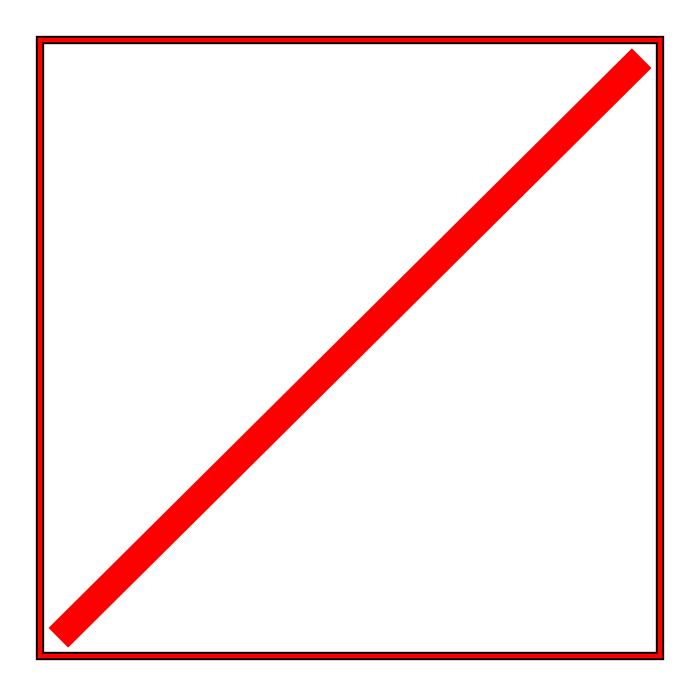}}
    \put(77,12){\includegraphics[height=0.03\textwidth]{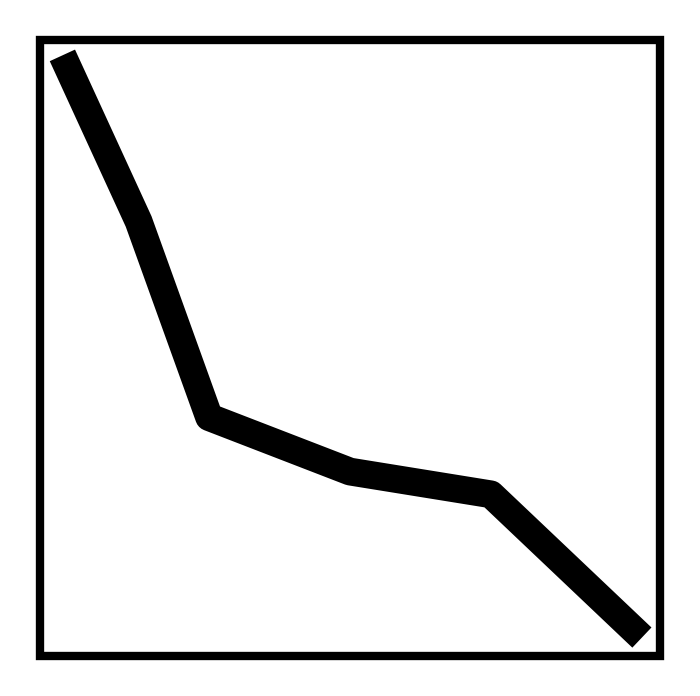}}
    \put(25,52){\includegraphics[height=0.03\textwidth]{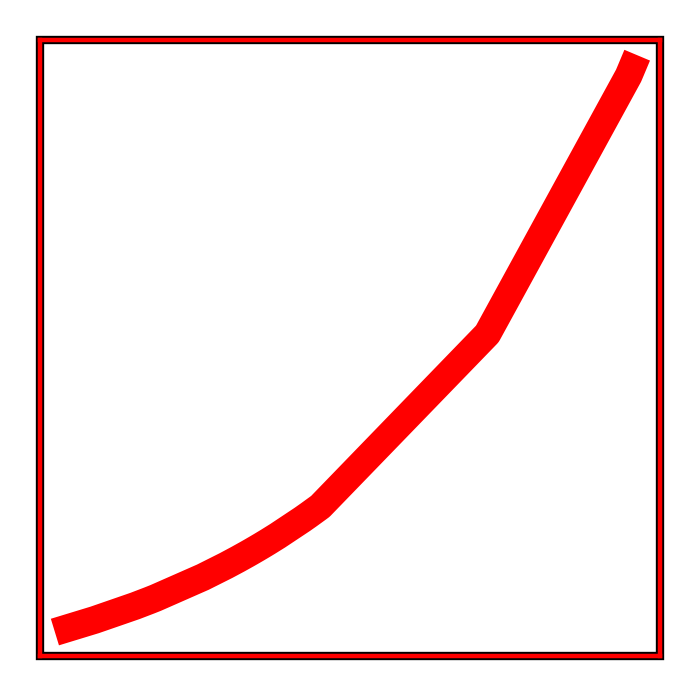}}
    \put(10,-13){$\hat{I}_1^{^1\bar{\vec{C}}_e}$}
    \put(33,-13){$\hat{I}_2^{^1\bar{\vec{C}}_e}$}
    \put(60,-13){$\hat{I}_3^{^1\bar{\vec{C}}_e}$}
    \put(88,-13){$\theta$}
    \put(45,88){$^1\psi^\text{KAN}$}
\end{overpic}
\vspace{1.2em}
}
\hfill
\subfloat[Model of $^1\omega(^1\bar{\vec{\Sigma}})$\label{subfig-2:omega1}]{
\vspace{.5em}
\begin{overpic}[height=0.2\textwidth]{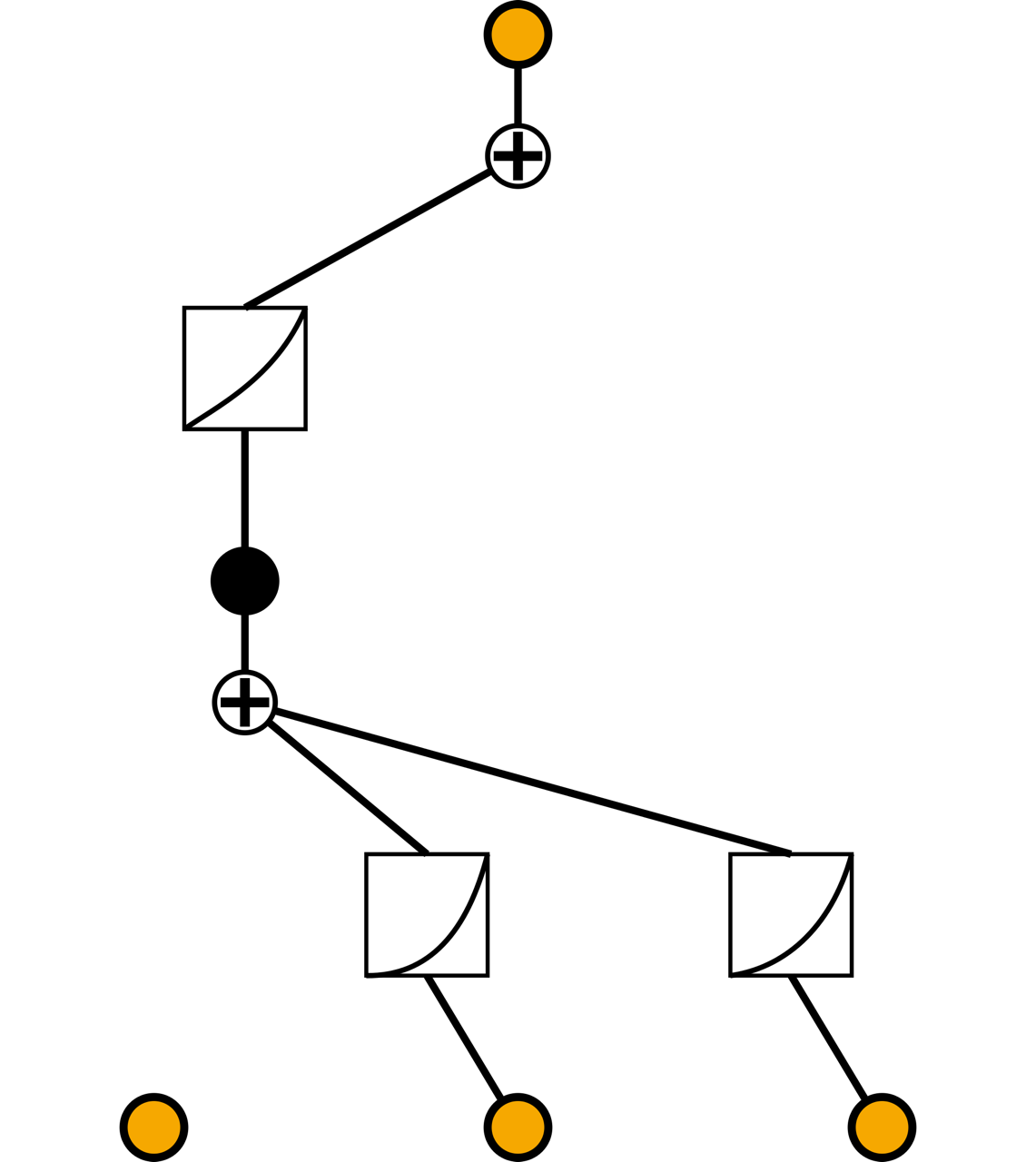}
    \put(30,15){\includegraphics[height=0.03\textwidth]{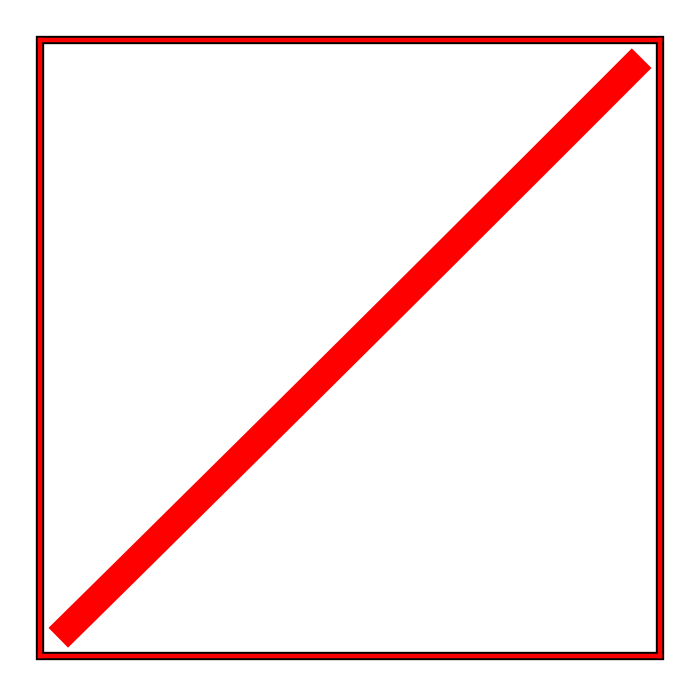}}
    \put(62,15){\includegraphics[height=0.03\textwidth]{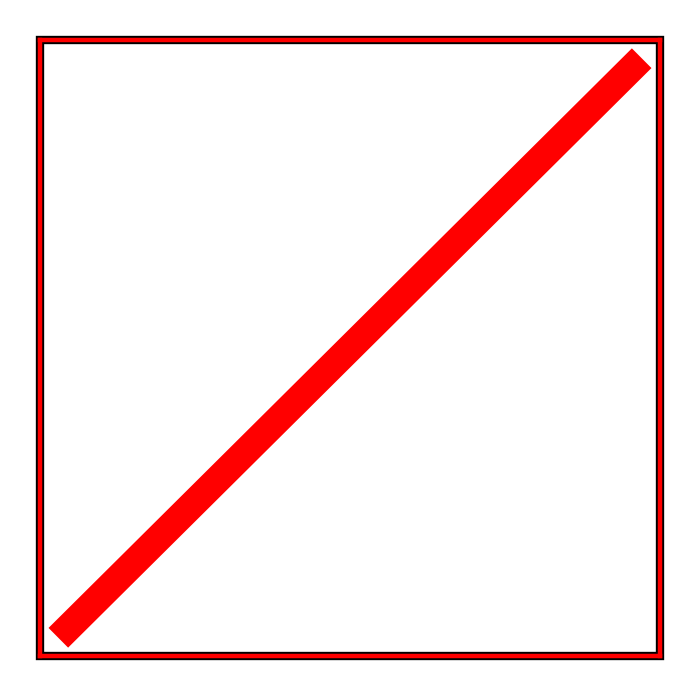}}
    \put(14.5,62){\includegraphics[height=0.03\textwidth]{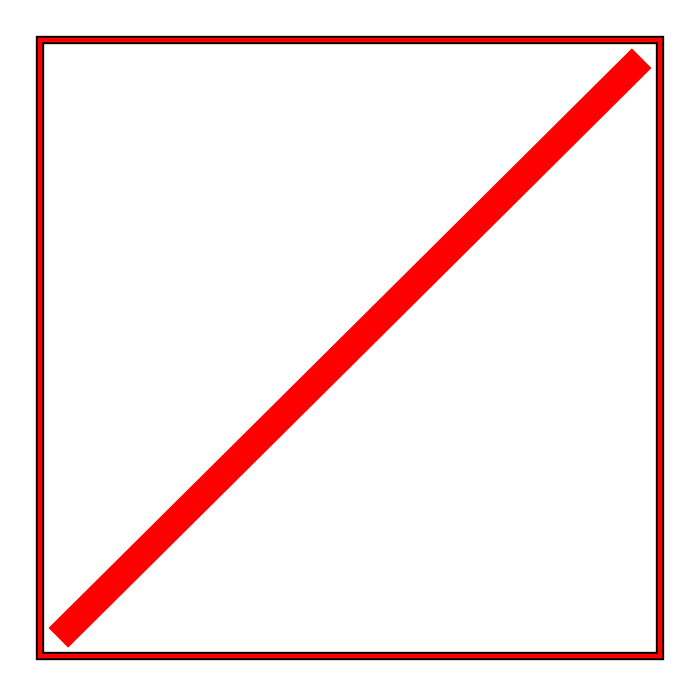}}
    \put(10,-15){$\hat{I}_1^{^1\bar{\vec{\Sigma}}}$}
    \put(42,-15){$\hat{J}_2^{^1\bar{\vec{\Sigma}}}$}
    \put(74,-15){$\hat{J}_3^{^1\bar{\vec{\Sigma}}}$}
    \put(38,105){$^1\omega^\text{KAN}$}
\end{overpic}
\vspace{1.2em}
}
\hfill
\vspace{3em}
\subfloat[Model of $^2\psi(^2\bar{\vec{C}}_e)$\label{subfig-3:psi2}]{
\vspace{.5em}
\begin{overpic}[height=0.2\textwidth]{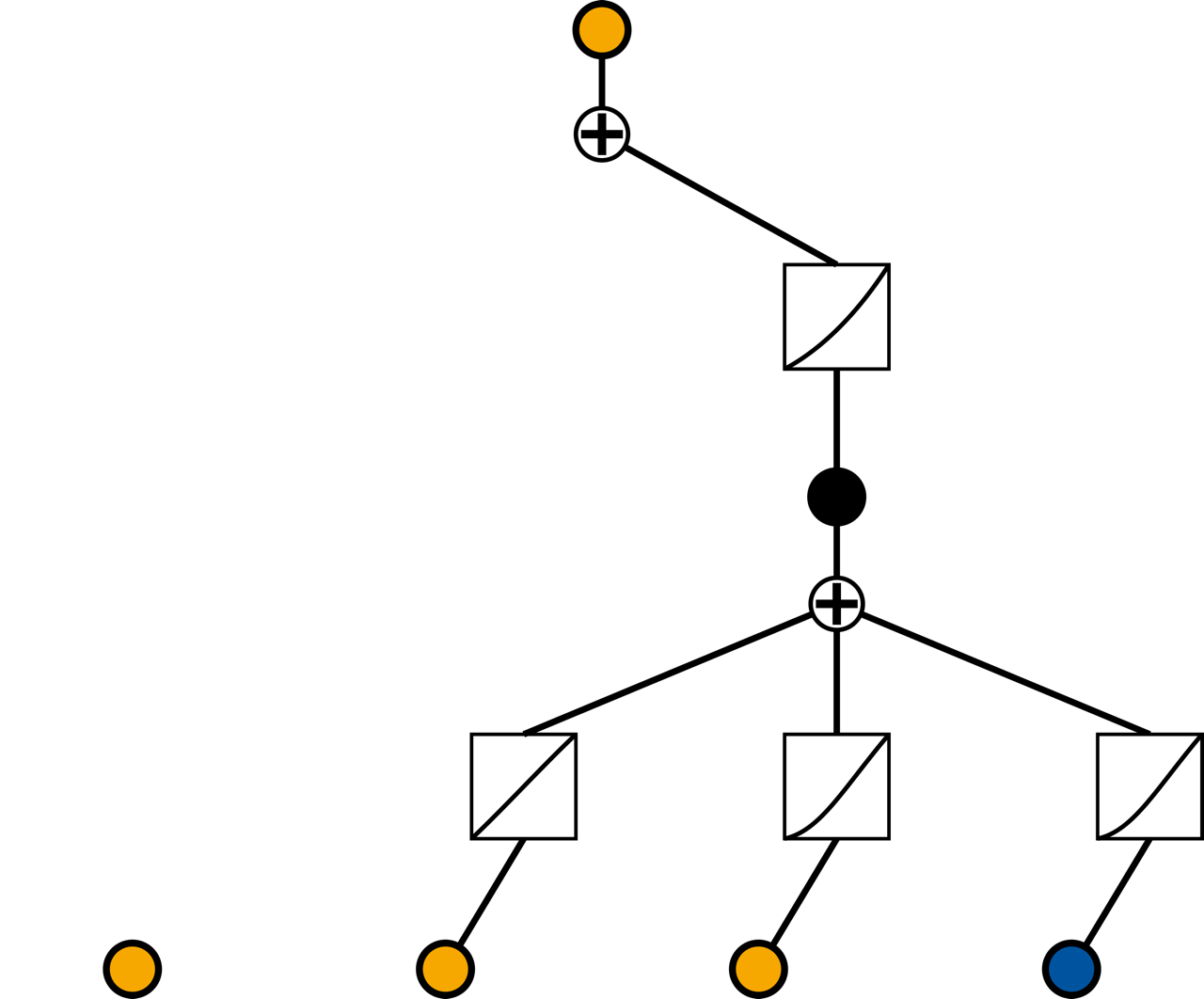}
    \put(38,12){\includegraphics[height=0.03\textwidth]{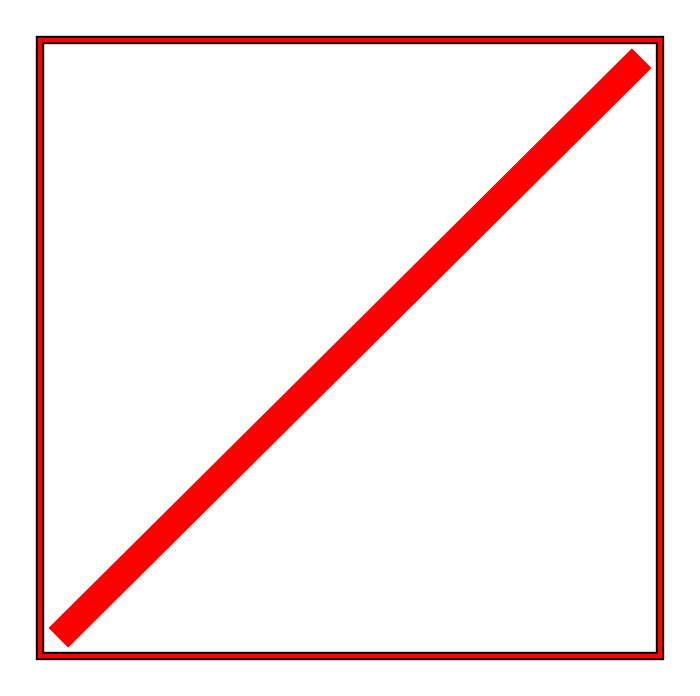}}
    \put(65,12){\includegraphics[height=0.03\textwidth]{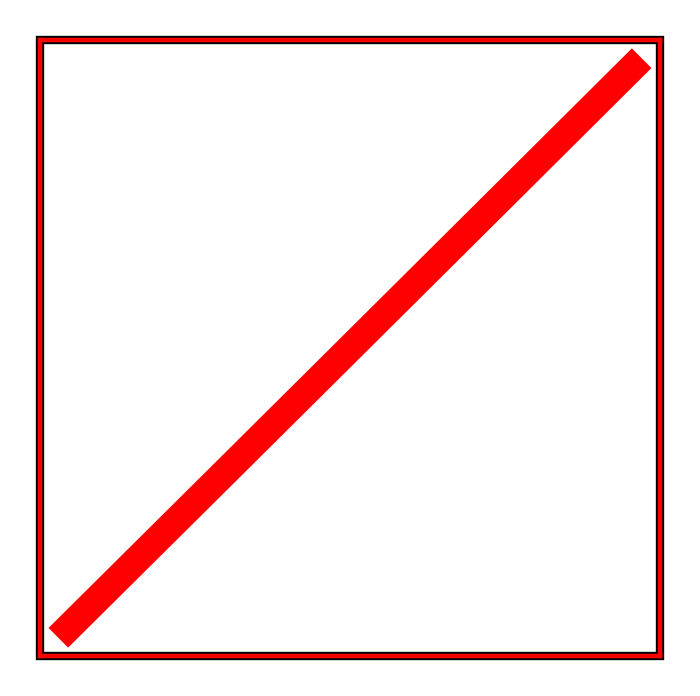}}
    \put(90,12){\includegraphics[height=0.03\textwidth]{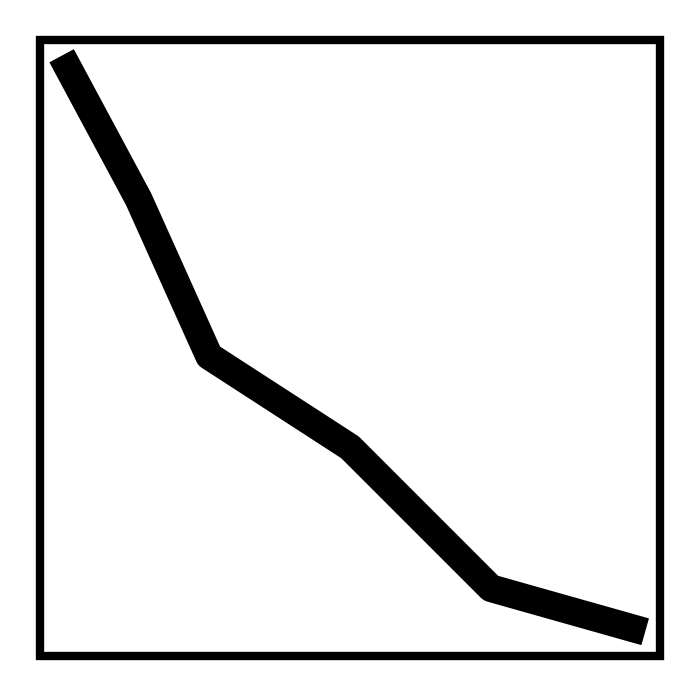}}
    \put(65,52){\includegraphics[height=0.03\textwidth]{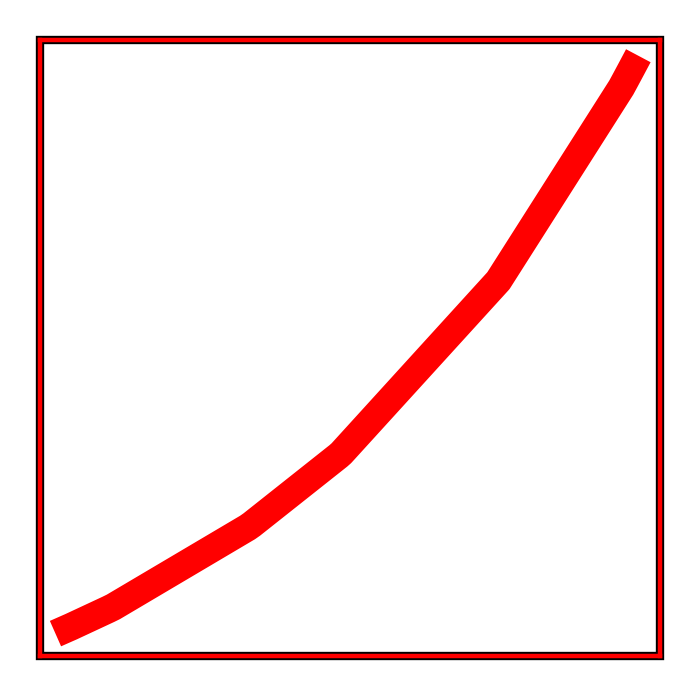}}
    \put(10,-13){$\hat{I}_1^{^2\bar{\vec{C}}_e}$}
    \put(33,-13){$\hat{I}_2^{^2\bar{\vec{C}}_e}$}
    \put(60,-13){$\hat{I}_3^{^2\bar{\vec{C}}_e}$}
    \put(88,-13){$\theta$}
    \put(45,88){$^2\psi^\text{KAN}$}
\end{overpic}
\vspace{1.2em}
}
\hfill
\subfloat[Model of $^2\omega(^2\bar{\vec{\Sigma}})$\label{subfig-2:omega2}]{
\vspace{.5em}
\begin{overpic}[height=0.2\textwidth]{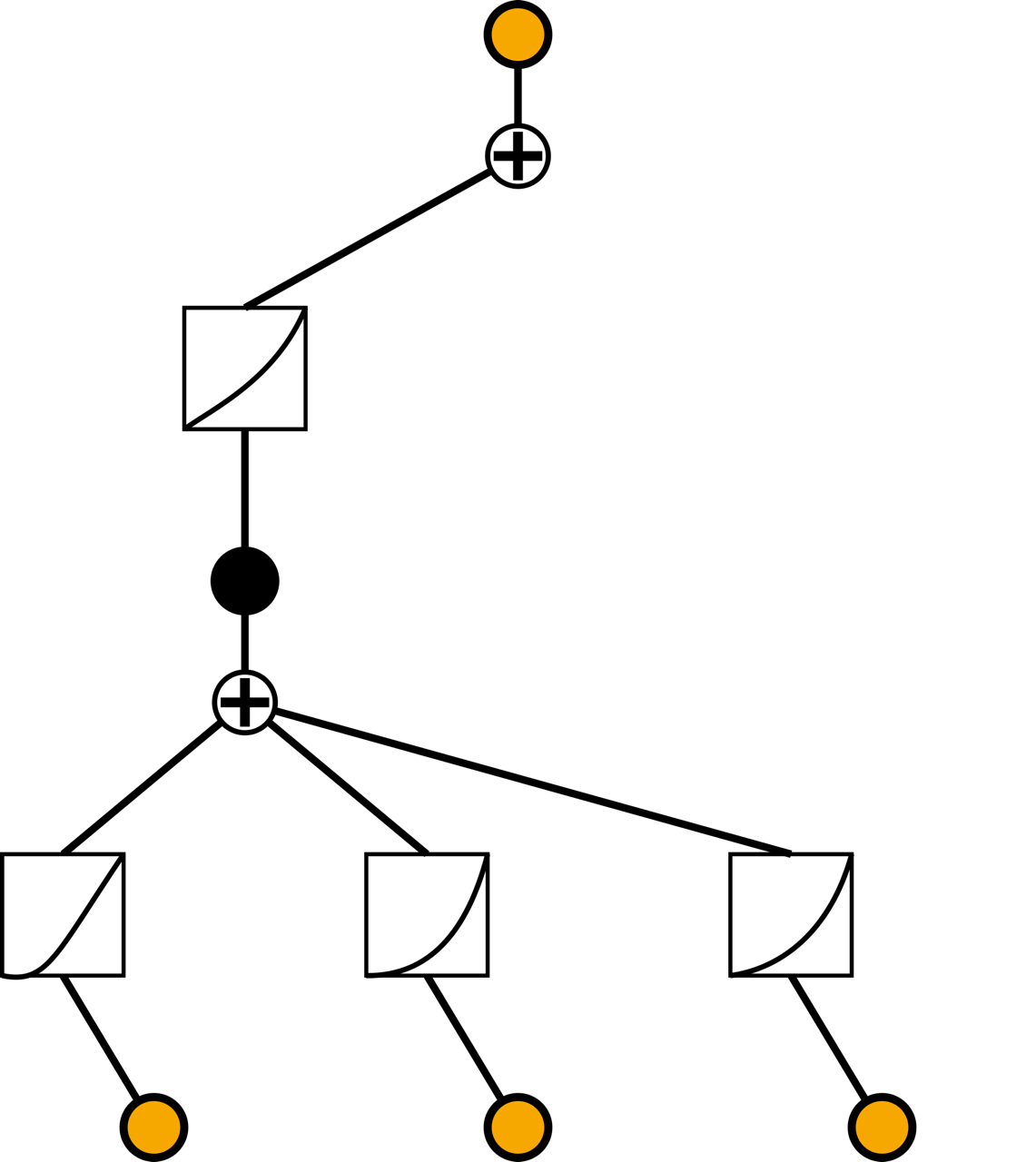}
    \put(-1,15){\includegraphics[height=0.03\textwidth]{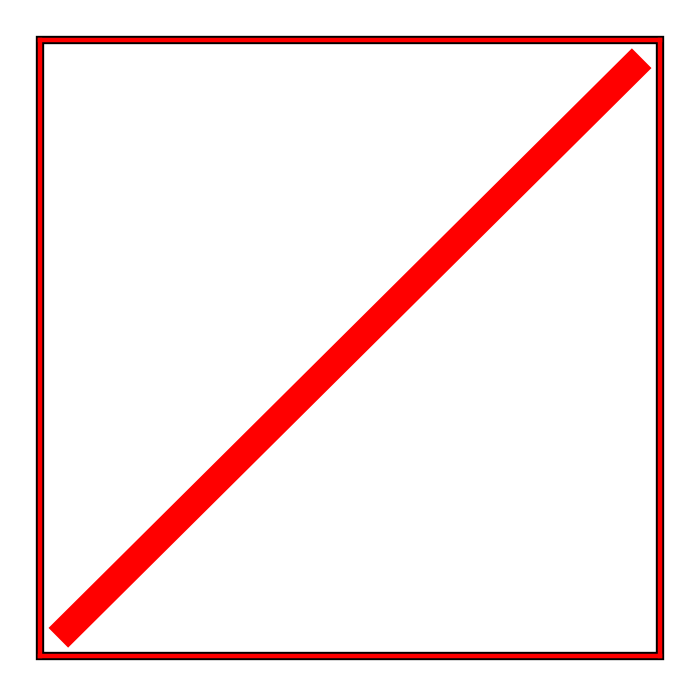}}
    \put(30,15){\includegraphics[height=0.03\textwidth]{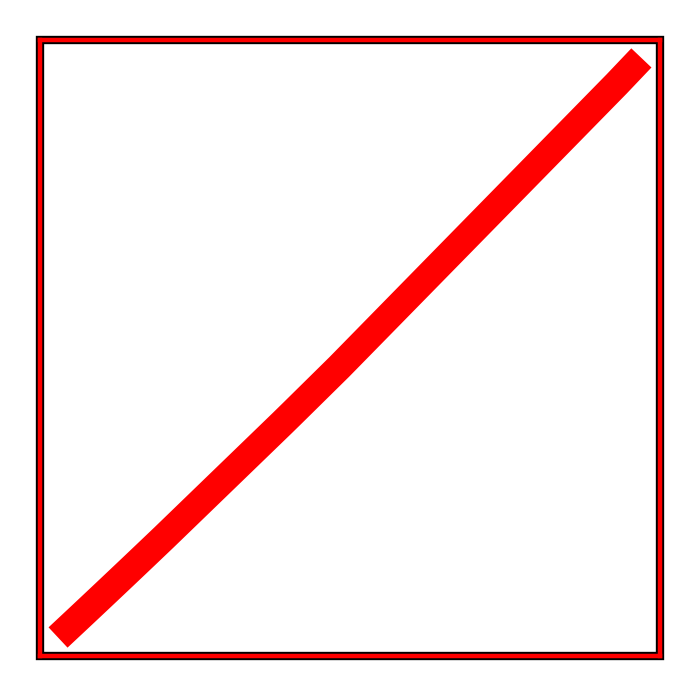}}
    \put(62,15){\includegraphics[height=0.03\textwidth]{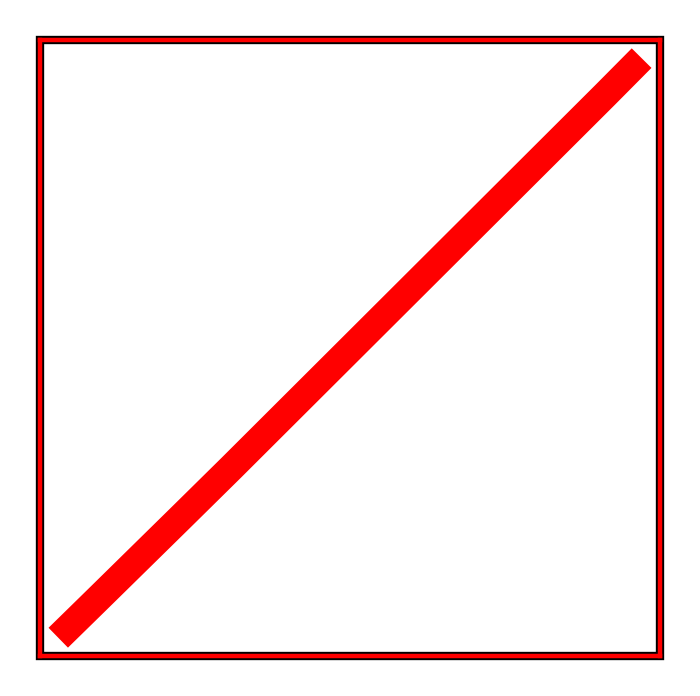}}
    \put(14.5,62){\includegraphics[height=0.03\textwidth]{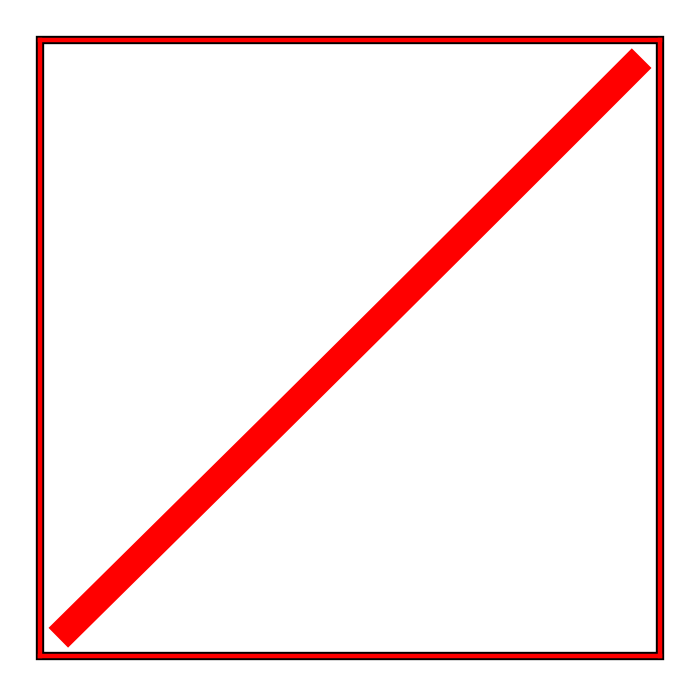}}
    \put(10,-15){$\hat{I}_1^{^2\bar{\vec{\Sigma}}}$}
    \put(42,-15){$\hat{J}_2^{^2\bar{\vec{\Sigma}}}$}
    \put(74,-15){$\hat{J}_3^{^2\bar{\vec{\Sigma}}}$}
    \put(38,105){$^2\omega^\text{KAN}$}
\end{overpic}
\vspace{1.2em}
}
\vspace{-3.5em}

%% file: table_symb_VHB4905.tex
\begin{tabularx}{\textwidth}{p{4cm}X}
\hline
Free energy &
$^1\psi(^1\bar{\vec{C}}_e) = a\cdot \exp(b\cdot \hat{I}_{1}^{^1\bar{\vec{C}}_e} + c\cdot \hat{I}_{2}^{^1\bar{\vec{C}}_e} +  ¹g(\theta) ) $ \\
\hline
\end{tabularx}
\begin{tabularx}{\textwidth}{p{4cm}p{1cm}p{1cm}p{1cm}}
\multirow{2}{*}{Material parameters} &
$a$ & $b$ & $c$\\
 & 0.0349 & 0.0855 & 0.0864\\\hline
Temperature dependency & \multicolumn{3}{l}{$^1g(\theta) = 10.0746\cdot(0.5239 - 0.0111 \cdot \theta)^3$}\\
\hline
\end{tabularx}

\vspace{4mm}

\begin{tabularx}{\textwidth}{p{4cm}X}
\hline
Dissipation potential &
$^1\omega(^1\bar{\vec{\Sigma}}) = \mathcal{H}(a \cdot \exp(b\cdot\exp(c \cdot \hat{J}_2^{^1\bar{\vec{\Sigma}}} )+ d\cdot\exp(e\cdot \hat{J}_3^{^1\bar{\vec{\Sigma}}})))$ \\
\hline
\end{tabularx}
\begin{tabularx}{\textwidth}{p{4cm}p{1cm}p{1cm}p{1cm}p{1cm}p{1cm}}
\multirow{2}{*}{Material parameters} &
$a$ & $b$ & $c$ &$d$&$e$\\
 & 31.918 & 0.1077 & 4.7 & 0.0382 & 10.0  \\
\hline
\end{tabularx}

\vspace{4mm}

\begin{tabularx}{\textwidth}{p{4cm}X}
\hline
Free energy &
$^2\psi(^2\bar{\vec{C}}_e) = a \cdot \exp (b\cdot\hat{I}_{2}^{^2\bar{\vec{C}}_e} + c\cdot\exp(d\cdot\hat{I}_{3}^{^2\bar{\vec{C}}_e}) +^2g(\theta))$ \\
\hline
\end{tabularx}
\begin{tabularx}{\textwidth}{p{4cm}p{1cm}p{1cm}p{1cm}p{1cm}}
\multirow{2}{*}{Material parameters} &
$a$ & $b$ & $c$ & $d$ \\
 & 0.319 & 0.0261 & 0.0921 & 10.0 \\\hline
Temperature dependency & \multicolumn{4}{l}{$^2g(\theta) = 
\begin{cases}
-1.7498 \cdot 10^{-4} \cdot \theta^{2} + 3.256 \cdot 10^{-3} \cdot \theta - 1.8712 \cdot 10^{-1}, & 0 \leq \theta < 40 \\[6pt]
-3.51 \cdot 10^{-5} \cdot \theta^{2} + 2.344 \cdot 10^{-3} \cdot \theta - 7.9019 \cdot 10^{-2}, & 40 \leq \theta \leq 80
\end{cases}$}\\
\hline
\end{tabularx}

\vspace{4mm}

\begin{tabularx}{\textwidth}{p{4cm}X}
\hline
Dissipation potential &
$^2\omega(^2\bar{\vec{\Sigma}}) = \mathcal{H}(a \cdot \exp(b\cdot \hat{I}_{1}^{^2\bar{\vec{\Sigma}}} + c\cdot\exp(d\cdot\hat{J}_2^{^2\bar{\vec{\Sigma}}})))$ \\
\hline
\end{tabularx}
\begin{tabularx}{\textwidth}{p{4cm}p{1cm}p{1cm}p{1cm}p{1cm}}
\multirow{2}{*}{Material parameters} &
$a$ & $b$ & $c$ &$d$\\
 & 0.0589 & 0.2543 & 0.0264 & 2.8013 \\
\hline
\end{tabularx}

%% file: AA_section_Conclusion.tex
\section{Discussion}
\label{sec:conclusion}

Developing accurate and interpretable constitutive models for materials undergoing nonlinear inelastic deformation, particularly at finite strains, remains a central challenge in computational mechanics. Although recent advances in data-driven modeling have accelerated model discovery, a persistent trade-off exists between predictive accuracy and physical interpretability. Thermodynamically consistent approaches often rely on highly flexible neural networks, which can lack transparency, or on symbolic methods with limited expressiveness \citep{holthusen2025generalized, flaschel2023automated2}. This work bridges this divide by introducing inelastic Constitutive Kolmogorov–Arnold Networks (iCKANs), which combine a generalized inelastic constitutive framework with the interpretability of Kolmogorov–Arnold Networks (KANs). This integration enables iCKANs to achieve predictive accuracy, interpretability, and robust extrapolation simultaneously. Their flexible architecture also facilitates symbolic extraction of learned potentials. Building upon previous KAN-based frameworks for hyperelastic materials, iCKANs establish a new paradigm for interpretable, data-driven discovery of inelastic material models \citep{abdolazizi2025constitutive, thakolkaran2025can}. 

\textbf{Interpretability.} The iCKAN framework shares the constitutive structure of inelastic Constitutive Neural Networks (iCANNs) \citep{holthusen2025generalized}, but differs fundamentally in its representation of the elastic and dissipation potentials. Whereas iCANNs employ custom-built convex neural networks, iCKANs use input-convex KANs, enabling a more compact and expressive formulation. The use of trainable B-spline activations enables iCKANs to achieve high expressiveness with substantially reduced network complexity. A major advantage of KANs is their built-in support for activation-level symbolic regression, which facilitates the direct extraction of interpretable, closed-form analytical expressions for both elastic and dissipation potentials. 
\color{author} 
In the present work, interpretability refers to the ability to express the constitutive relations as compact symbolic formulas whose individual terms can be inspected and related to established constitutive modeling concepts. While alternative definitions of interpretability exist, this notion is consistent with the objective of deriving human-readable constitutive laws that resemble classical phenomenological models.
\color{black} 
This capability preserves flexibility in data-driven discovery while yielding transparent, physically meaningful models with strong predictive performance \citep{abdolazizi2025constitutive}. Furthermore, iCKANs naturally accommodate additional non-mechanical features, such as temperature, whose effects can be extracted as explicit mathematical expressions, providing valuable physical insight.


\textbf{Extrapolability.} The combination of a thermodynamically consistent constitutive framework and convex network architecture ensures that the learned elastic and dissipation potentials satisfy essential physical constraints across their entire input domain. Importantly, these constraints are retained in the closed-form symbolic expressions obtained through the KAN symbolification process. This structure enhances robustness and enables physically consistent extrapolation beyond the training regime \citep{abdolazizi2025constitutive}. As demonstrated in synthetic benchmarks (Section~\ref{sec:results_synthetic}), models trained on limited deformation data can still provide reliable predictions for unseen, larger deformations using both the spline-based KAN and its symbolic form.

\textbf{Computational efficiency.} By promoting sparsity in the KAN architecture, the symbolic regression process yields concise, closed-form expressions for both elastic and dissipation potentials. These analytical formulations can be seamlessly integrated into finite element solvers, eliminating the need for network evaluation during simulation. As a result, the iCKAN approach introduces no additional computational overhead at deployment \citep{abdolazizi2025constitutive}.

\textbf{Choice of input representation.} This study employs the principal invariants of the elastic right Cauchy–Green tensor as inputs for the free energy, but alternative objective representations, such as principal stretches, are equally valid in the iCKAN formulation \citep{abdolazizi2025constitutive}. For the dissipation potential, invariant-based stress measures are used, with the option to augment input sets by combining stress and deformation invariants. The selection of input representation influences both predictive accuracy and interpretability, and a systematic comparison across different material classes is a promising direction for future research \citep{holthusen2025generalized}.

\textbf{Predictive performance.} The results demonstrate that iCKANs deliver strong predictive performance on both synthetic and experimental viscoelastic datasets, including cases with feature-dependent material behavior.
\color{author}
However, in certain scenarios, conventional neural-network-based approaches may still attain lower prediction errors \citep{abdolazizi2026thermodynamically,kalina2026physics}. While the universal approximation capability of KANs is not guaranteed under all conditions \citep{samadi2024smooth}, potential limitations of KAN-based potentials may exist. Yet, our investigations indicate that predictive accuracy can often be improved by increasing the number of branches in the iCKAN, choosing a wider network architecture, or relaxing the sparsity constraints imposed on the network. While this increases the expressive capacity of the model, it also leads to more complex symbolic expressions and thus reduces interpretability. In this sense, the iCKAN framework deliberately prioritizes interpretability and physical structure, which may result in modest accuracy trade-offs.

The proposed iCKAN framework intentionally balances predictive accuracy, thermodynamic consistency, and interpretability. Consequently, modest reductions in accuracy may occur when compared to highly expressive black-box neural-network models. 
\color{black}
Closing this gap remains an open challenge, motivating further work on training strategies, architectural optimization, and advanced regularization techniques to enhance predictive performance while retaining interpretability.

\textbf{Limitations and future work.} While this work primarily addresses viscoelastic material behavior, the iCKAN framework is broadly applicable to other classes of inelasticity, including plasticity, damage, growth and remodeling \citep{boes2024accounting,holthusen2025generalized,holthusen2025automated}. Extensions to anisotropic materials and structural-scale implementations are also feasible \citep{holthusen2026complement,kalina2025neural,thakolkaran2025can}. These directions open new possibilities for feature-aware material characterization and data-driven design.

From a numerical standpoint, KANs require the specification of spline grid ranges, which makes initialization more involved than in standard multilayer perceptrons. Although the present approximation strategy proved effective, future research should focus on developing robust adaptive grids and training techniques. Additional opportunities include investigating partially input-convex architectures, advanced pruning and convexification strategies, and automated hyperparameter selection, potentially leveraging large language models or other AI-driven tools \citep{hospedales2021meta,tacke2025constitutive}. Efficient, standardized integration into finite element solvers for both forward and inverse boundary value problems represents another key step for broad adoption.

\textbf{Conclusion.} Reliable and interpretable constitutive models for inelastic materials under finite deformations are essential for fields where complex, feature-dependent responses arise, such as biomedical engineering. The iCKAN framework presented in this work provides a robust, physically consistent, and data-driven solution for automated discovery of inelastic material models. By incorporating input-convex Kolmogorov–Arnold Networks, the approach achieves a unique combination of flexibility, predictive accuracy, and transparent symbolic representations, advancing the state of the art in interpretable inelastic  material modeling.

%% file: Appendix_A.tex
\section{Appendix}
\label{app:A}

\color{author}
\subsection{Thermodynamic consistency}
\label{app:limit_analysis}

In the following, we analyze the reduced dissipation inequality
$\mathcal{D}(\bar{\vec{\Sigma}}) = \bar{\vec{\Sigma}} : \dfrac{\partial \omega}{ \bar{\vec{\Sigma}}}
$,
where the dissipation potential is given by $\omega
=
\omega\!\left(
I_1^{\bar{\vec{\Sigma}}},
\sqrt{J_2^{\bar{\vec{\Sigma}}}},
\sqrt[3]{J_3^{\bar{\vec{\Sigma}}}}
\right)$.
By the chain rule, the dissipation can be written as
\begin{equation}
\mathcal D(\bar{\vec{\Sigma}}) =
\bar{\vec{\Sigma}}:
\Bigg[
\frac{\partial\omega}{\partial I_1}\,\vec I
+
\frac{\partial\omega}{\partial\sqrt{J_2}}
\frac{\partial\sqrt{J_2}}{\partial \bar{\vec{\Sigma}}} 
+
\frac{\partial\omega}{\partial\sqrt[3]{J_3}}
\frac{\partial\sqrt[3]{J_3}}{\partial \bar{\vec{\Sigma}}}
\Bigg].
\end{equation}
As $\bar{\vec{\Sigma}} \in \sym(3)$, this can then be rewritten as
\begin{equation}
\mathcal D(\bar{\vec{\Sigma}}) =
\frac{\partial\omega}{\partial I_1} \bar{\vec{\Sigma}}:\vec I + 
\frac{\partial\omega}{\partial\sqrt{J_2}} \frac{\partial\sqrt{J_2}}{\partial \bar{\vec{\Sigma}}}  : \bar{\vec{\Sigma}}  + 
\frac{\partial\omega}{\partial\sqrt[3]{J_3}} \frac{\partial\sqrt[3]{J_3}}{\partial \bar{\vec{\Sigma}}} : \bar{\vec{\Sigma}} .
\label{eq:dissipation2}
\end{equation}
So simplicity, we denote 
\begin{equation}
    f_1(\bar{\vec{\Sigma}}) = I_1^{\bar{\vec{\Sigma}}} = \tr (\bar{\vec{\Sigma}}), \quad f_2(\bar{\vec{\Sigma}}) = \sqrt{J_2^{\bar{\vec{\Sigma}}}} = \sqrt{\frac{1}{2}\,\mathrm{tr}\!\left(\mathrm{dev}(\bar{\vec{\Sigma}})^2\right)}, \quad f_3(\bar{\vec{\Sigma}}) = \sqrt[3]{J_3^{\bar{\vec{\Sigma}}}} = \sqrt[3]{
\frac{1}{3}\,
\mathrm{tr}\!\left(\mathrm{dev}(\bar{\vec{\Sigma}})^3\right)
}\,.
\end{equation}
Since $f_2$ and $f_3$ are positively homogeneous of degree one, it follows from Euler's theorem for positively homogeneous functions that
\begin{equation}
    \frac{\partial f_2}{\partial \bar{\vec{\Sigma}}} : \bar{\vec{\Sigma}} = f_2(\bar{\vec{\Sigma}}), \quad \frac{\partial f_3}{\partial \bar{\vec{\Sigma}}} : \bar{\vec{\Sigma}} = f_3(\bar{\vec{\Sigma}})\,. 
\end{equation}
Thus, \autoref{eq:dissipation2} can be rewritten as
\begin{equation}
    \mathcal D(\bar{\vec{\Sigma}}) = 
\frac{\partial\omega}{\partial I_1} \tr(\bar{\vec{\Sigma}}) + 
\frac{\partial\omega}{\partial\sqrt{J_2}}\sqrt{J_2}  + 
\frac{\partial\omega}{\partial\sqrt[3]{J_3}} \sqrt[3]{J_3} .
\label{eq:dissipation3}
\end{equation}
Since $\omega$ is convex, nonnegative, and zero-valued at the origin, it attains its minimum at the origin. Therefore,
\begin{equation}
\frac{\partial\omega}{\partial \mathbf{z}} \cdot \mathbf{z} \geq 0, \quad \text{where} \quad \mathbf{z}={(I_1,\sqrt{J_2},\sqrt[3]{J_3})^T}.
\end{equation}
Therefore, the \autoref{eq:dissipation3} is fulfilled.

\subsubsection{Limit analysis} 
However, the above analysis relies on the assumption that the derivatives of $f_2$ and $f_3$ are well-defined for all $\bar{\vec{\Sigma}} \in \sym(3)$. In particular, the derivatives of $f_2$ and $f_3$ are not defined at $\bar{\vec{\Sigma}} = \vec{0}$. To ensure that the dissipation remains non-negative even at $\bar{\vec{\Sigma}} = \vec{0}$, it is necessary to analyze the limit behavior of the terms involving $f_2$ and $f_3$ as $\bar{\vec{\Sigma}}$ approaches zero. Specifically, we need to ensure that
\begin{equation}
    \frac{\partial\sqrt{J_2}}{\partial \bar{\vec{\Sigma}}}  : \bar{\vec{\Sigma}}  \geq 0, \quad \frac{\partial\sqrt[3]{J_3}}{\partial \bar{\vec{\Sigma}}} : \bar{\vec{\Sigma}}  \geq 0\,.
\end{equation}
The following proposition investigates the corresponding limit behavior.
\begin{proposition}
Let
\begin{equation}
f_2(\bar{\vec{\Sigma}})
=
\sqrt{\frac{1}{2}\,\mathrm{tr}\!\left(\mathrm{dev}(\bar{\vec{\Sigma}})^2\right)},
\end{equation}
where
\begin{equation}
\mathrm{dev}(\bar{\vec{\Sigma}})
=
\bar{\vec{\Sigma}}
-\frac{1}{3}\mathrm{tr}(\bar{\vec{\Sigma}})\vec{I}.
\end{equation}
For every direction $\bar{\vec{\Sigma}}\neq\vec{0}$,
\begin{equation}
\frac{\partial f_2}{\partial \bar{\vec{\Sigma}}}
=
\frac{\mathrm{dev}(\bar{\vec{\Sigma}})}
{\sqrt{2\,\mathrm{tr}\!\left(\mathrm{dev}(\bar{\vec{\Sigma}})^2\right)}}.
\end{equation}
Then
\begin{equation}
\lim_{\bar{\vec{\Sigma}}\rightarrow\vec{0}}
\bar{\vec{\Sigma}}:\frac{\partial f_2}{\partial \bar{\vec{\Sigma}}}
=
0.
\end{equation}
\end{proposition}

\begin{proof}
To investigate the limit at the origin, consider an arbitrary direction
$\vec{H}\in\mathrm{Sym}(3)\setminus\{\vec{0}\}$ 
and define
\begin{equation}
\bar{\vec{\Sigma}}=t\vec{H},
\qquad t\in\mathbb{R}.
\end{equation}
For $t\neq 0$, the linearity of the deviatoric operator yields
\begin{equation}
\mathrm{dev}(t\vec{H})
=
t\,\mathrm{dev}(\vec{H}).
\end{equation}
Hence,
\begin{equation}
\frac{\partial f_2}{\partial \bar{\vec{\Sigma}}}(t\vec{H})
=
\frac{t\,\mathrm{dev}(\vec{H})}
{\sqrt{2\,t^2\,\mathrm{tr}\!\left(\mathrm{dev}(\vec{H})^2\right)}}.
\end{equation}
Consequently,
\begin{equation}
(t\vec{H}):
\frac{\partial f_2}{\partial \bar{\vec{\Sigma}}}(t\vec{H})
=
\frac{t\vec{H}:t\,\mathrm{dev}(\vec{H})}
{\sqrt{2\,t^2\,\mathrm{tr}\!\left(\mathrm{dev}(\vec{H})^2\right)}}.
\end{equation}
Since
\begin{equation}
\vec{H}
=
\mathrm{dev}(\vec{H})
+
\frac{1}{3}\mathrm{tr}(\vec{H})\vec{I},
\end{equation}
and
\begin{equation}
\mathrm{tr}\!\left(\mathrm{dev}(\vec{H})\right)=0,
\end{equation}
it follows that
\begin{equation}
\vec{H}:\mathrm{dev}(\vec{H})
=
\mathrm{dev}(\vec{H}):\mathrm{dev}(\vec{H})
=
\mathrm{tr}\!\left(\mathrm{dev}(\vec{H})^2\right).
\end{equation}
Therefore,
\begin{equation}
(t\vec{H}):
\frac{\partial f_2}{\partial \bar{\vec{\Sigma}}}(t\vec{H})
=
\frac{|t|\,\mathrm{tr}\!\left(\mathrm{dev}(\vec{H})^2\right)}
{\sqrt{2\,\mathrm{tr}\!\left(\mathrm{dev}(\vec{H})^2\right)}}.
\end{equation}
After simplification,
\begin{equation}
(t\vec{H}):
\frac{\partial f_2}{\partial \bar{\vec{\Sigma}}}(t\vec{H})
=
|t|
\sqrt{\frac{1}{2}\,\mathrm{tr}\!\left(\mathrm{dev}(\vec{H})^2\right)}.
\end{equation}
Hence,
\begin{equation}
\lim_{t\to 0}
(t\vec{H}):
\frac{\partial f_2}{\partial \bar{\vec{\Sigma}}}(t\vec{H})
=
\sqrt{\frac{1}{2}\,\mathrm{tr}\!\left(\mathrm{dev}(\vec{H})^2\right)}
\lim_{t\to 0}|t|
=
0.
\end{equation}
Since the result holds for an arbitrary direction $\vec{H}$, we conclude that
\begin{equation}
\lim_{\bar{\vec{\Sigma}}\rightarrow\vec{0}}
\bar{\vec{\Sigma}}:\frac{\partial f_2}{\partial \bar{\vec{\Sigma}}}
=
0.
\end{equation}
\end{proof}

\textbf{Remark.} An analogous conclusion holds for
\begin{equation}
f_3(\bar{\vec{\Sigma}})
=
\left(
\frac{1}{3}\,
\mathrm{tr}\!\left(\mathrm{dev}(\bar{\vec{\Sigma}})^3\right)
\right)^{1/3},
\end{equation}
which is also positively homogeneous of degree one and continuous at
\(\bar{\vec{\Sigma}}=\vec{0}\). Consequently,
\begin{equation}
\lim_{\bar{\vec{\Sigma}}\to\vec{0}}
\bar{\vec{\Sigma}}:
\frac{\partial f_3}{\partial \bar{\vec{\Sigma}}}
=
0.
\end{equation}
The same result can alternatively be obtained from Euler's theorem for positively homogeneous functions.

Consequently, although the individual derivatives are not defined at
\(\bar{\vec{\Sigma}}=\vec{0}\), the corresponding dissipation contributions vanish in the limit.

\subsubsection{Modified square root and cubic root}
\label{app:modified_roots}
To ensure that the derivative of $f_2$ and $f_3$ is well-defined at $\bar{\vec{\Sigma}} = \vec{0}$, we consider modified versions of these functions that are continuously differentiable everywhere, i.e.,
\begin{equation}
    \tilde{f}_2(\bar{\vec{\Sigma}}) = \frac{J_2^{\bar{\vec{\Sigma}}}}{(J_2^{\bar{\vec{\Sigma}}} + \epsilon_2)^{1/2}}, \quad \tilde{f}_3(\bar{\vec{\Sigma}}) = \frac{J_3^{\bar{\vec{\Sigma}}}}{(J_3^{\bar{\vec{\Sigma}}} + \epsilon_3)^{2/3}}\,,
\end{equation}
where $\epsilon_2, \epsilon_3 > 0$ are small regularization parameters. With this modification, the derivatives of $\tilde{f}_2$ and $\tilde{f}_3$ are well-defined for all $\bar{\vec{\Sigma}} \in \sym(3)$, including at $\bar{\vec{\Sigma}} = \vec{0}$. This is shown in the following proposition.
\begin{proposition}
Let
\begin{equation}
\tilde{f}_2(\bar{\vec{\Sigma}})
=
\frac{J_2^{\bar{\vec{\Sigma}}}}
{\bigl(J_2^{\bar{\vec{\Sigma}}}+\epsilon\bigr)^{1/2}},
\qquad
\tilde{f}_3(\bar{\vec{\Sigma}})
=
\frac{J_3^{\bar{\vec{\Sigma}}}}
{\bigl(J_3^{\bar{\vec{\Sigma}}}+\epsilon\bigr)^{2/3}},
\end{equation}
with $\epsilon>0$. Then
\begin{equation}
\lim_{\bar{\vec{\Sigma}}\to\vec{0}}
\frac{\partial \tilde{f}_i}{\partial\bar{\vec{\Sigma}}}
=
\vec{0},
\qquad i=2,3.
\end{equation}
\end{proposition}

\begin{proof}
Let
\begin{equation}
\phi_2(x)=\frac{x}{(x+\epsilon)^{1/2}},
\qquad
\phi_3(x)=\frac{x}{(x+\epsilon)^{2/3}}.
\end{equation}
Since $\epsilon>0$, both $\phi_2$ and $\phi_3$ are continuously differentiable at $x=0$, with finite derivatives
\begin{equation}
\phi_2'(0)=\epsilon^{-1/2},
\qquad
\phi_3'(0)=\epsilon^{-2/3}.
\end{equation}
By the chain rule,
\begin{equation}
\frac{\partial \tilde{f}_i}{\partial\bar{\vec{\Sigma}}}
=
\phi_i'(J_i(\bar{\vec{\Sigma}}))
\frac{\partial J_i}{\partial\bar{\vec{\Sigma}}},
\qquad i=2,3.
\end{equation}
Since
\begin{equation}
J_i(\bar{\vec{\Sigma}})\to0
\qquad\text{and}\qquad
\frac{\partial J_i}{\partial\bar{\vec{\Sigma}}}\to\vec0
\quad
(\bar{\vec{\Sigma}}\to\vec0),
\end{equation}
it follows that
\begin{equation}
\lim_{\bar{\vec{\Sigma}}\to\vec0}
\frac{\partial \tilde{f}_i}{\partial\bar{\vec{\Sigma}}}
=
\vec0,
\qquad i=2,3.
\end{equation}
\end{proof}

\textbf{Remark.} This regularization is consistent with the thermodynamic principles, as $\bar{\vec{\Sigma}} = \vec{0}$, there should be no dissipation. In fact, this regularization is physically motivated and is required to ensure the thermodynamic consistency at the origin. Moreover, as $\epsilon_2, \epsilon_3 \to 0$, the modified functions $\tilde{f}_2$ and $\tilde{f}_3$ converge to the original functions $f_2$ and $f_3$, respectively, and the corresponding dissipation contributions converge to the original contributions, ensuring that the thermodynamic consistency is preserved in the limit.
\color{black}

\subsection{Generation of a convex function from a convex, non-decreasing function}
\label{app:convexification}

Let $f(\vec{x})$ be a multivariate function of $\vec{x}$, which is convex and non-decreasing with respect to each component of $\vec{x}$. We prove that
\begin{equation}
    \tilde{f}(\vec{x}) = f(\vec{x}) - f(\vec{\alpha}) - \left.\ \vec{\nabla} f(\vec{x})\right|_{\vec{x} = \vec{\alpha}}  \cdot (\vec{x}-\vec{\alpha})
\end{equation}
is convex with respect to each component of $\vec{x}$, nonnegative and has zero value and zero-valued derivative at $\vec{x} = \vec{\alpha}$.

\textbf{Convexity.} Suppose $f(\vec{x})$ is a twice-diffentiable function, it holds that the second partial derivate with respect to each argument must be nonnegative, i.e., $\partial^2 f / \partial x_i^2 \geq 0$ for all $i$. The first partial derivative of $\tilde{f}(\vec{x})$ with respect to $x_i$ is given as
\begin{equation}
    \frac{\partial \tilde{f}(\vec{x})}{\partial x_i} = \frac{\partial f(\vec{x})}{\partial x_i} - \left.\frac{\partial f(\vec{x})}{\partial x_i}\right|_{\vec{x} = \vec{\alpha}}\,,
\end{equation}
and the second partial derivative is
\begin{equation}
    \frac{\partial^2 \tilde{f}(\vec{x})}{\partial x_i^2} = \frac{\partial^2 f(\vec{x})}{\partial x_i^2} \geq 0 \quad \text{for all } i\,.
\end{equation}
Therefore, $\tilde{f}(\vec{x})$ is convex with respect to each component of $\vec{x}$.

\textbf{Nonnegativity.} Since $f(\vec{x})$ is convex and non-decreasing with respect to each component of $\vec{x}$, it must hold that for any $\vec{x}$ and $\vec{\alpha}$,
\begin{equation}
    f(\vec{x}) \geq f(\vec{\alpha}) + \left.\ \vec{\nabla} f(\vec{x})\right|_{\vec{x} = \vec{\alpha}}  \cdot (\vec{x}-\vec{\alpha})\,.  
\end{equation}
Rearranging this inequality gives
\begin{equation}    
    f(\vec{x}) - f(\vec{\alpha}) - \left.\ \vec{\nabla} f(\vec{x})\right|_{\vec{x} = \vec{\alpha}}  \cdot (\vec{x}-\vec{\alpha}) \geq 0\,,
\end{equation}
which implies that $\tilde{f}(\vec{x}) \geq 0$ for all $\vec{x}$.

\textbf{Zero value and zero derivative at $\vec{x} = \vec{\alpha}$.} Evaluating $\tilde{f}(\vec{x})$ at $\vec{x} = \vec{\alpha}$ gives
\begin{equation}
    \tilde{f}(\vec{\alpha}) = f(\vec{\alpha}) - f(\vec{\alpha}) - \left.\ \vec{\nabla} f(\vec{x})\right|_{\vec{x} = \vec{\alpha}}  \cdot (\vec{\alpha}-\vec{\alpha}) = 0\,.
\end{equation}
The gradient of $\tilde{f}(\vec{x})$ at $\vec{x} = \vec{\alpha}$ is
\begin{equation}
    \left.\ \vec{\nabla} \tilde{f}(\vec{x})\right|_{\vec{x} = \vec{\alpha}} = \left.\ \vec{\nabla} f(\vec{x})\right|_{\vec{x} = \vec{\alpha}} - \left.\ \vec{\nabla} f(\vec{x})\right|_{\vec{x} = \vec{\alpha}} = \vec{0}\,.
\end{equation}
Thus, $\tilde{f}(\vec{x})$ has a zero value and derivative at $\vec{x} = \vec{\alpha}$.

\subsection{Stress-free undeformed state}
\label{app:psi_undeformed}

To ensure that the undeformed state is correctly represented, that is, that there should be no energy or stress at zero elastic deformation \citep{holzapfel2000nonlinear}, the free energy function must satisfy the following conditions
\begin{equation}
    \psi(\bar{\vec{C}}_e=\vec{I}) = 0 \quad \text{and} \quad \left.\dfrac{\partial \psi}{\partial \bar{\vec{C}}_e}\right|_{\bar{\vec{C}}_e=\vec{I}} = \vec{0}\,.
\end{equation}
This can be achieved by adding a correction term for stress and potential, $\psi^\sigma$ and $\psi^\epsilon$, respectively, to the output free energy of the KAN, $\psi^{\text{KAN}}$, i.e., $\psi = \psi^{\text{KAN}} + \psi^\sigma(J) + \psi^\epsilon$. These two terms are calculated dependent on the output free energy at the undeformed state. For more detailed explanation, please refer to Appendix D in \citep{abdolazizi2025constitutive}.
In our case, the derivatives of the modified invariants w.r.t. $\bar{\vec{C}}_e$ are zero at the undeformed state \citep{thakolkaran2025can}, i.e., 
\begin{equation}
    \left.\dfrac{\partial \hat{I}_1^{\bar{\vec{C}}_e}}{\partial \bar{\vec{C}}_e}\right|_{\bar{\vec{C}}_e=\vec{I}} = \left.\dfrac{\partial \hat{I}_2^{\bar{\vec{C}}_e}}{\partial \bar{\vec{C}}_e}\right|_{\bar{\vec{C}}_e=\vec{I}} = \left.\dfrac{\partial \hat{I}_3^{\bar{\vec{C}}_e}}{\partial \bar{\vec{C}}_e}\right|_{\bar{\vec{C}}_e=\vec{I}} = \vec{0}\,.
\end{equation} 
Therefore, the stress-free undeformed state is automatically ensured. The potential-free undeformed state is achieved by subtracting the output free energy evaluated at in the undeformed state,
\begin{equation}
    \psi = \psi^{\text{KAN}} - \left.\psi^{\text{KAN}}\right|_{\bar{\vec{C}}_e=\vec{I}}\,.
\end{equation}

\subsection{Alternative choices of argument for dissipation potential}
\label{app:Potential_arguments}

\textbf{Adding negative stress invariants.} For more network expressibility, the negative invariants can be added additionally to the positive invariants to enhance the input arguments of KAN for the dissipation potential, i.e., 
\begin{equation}
\omega = \omega\left(\hat{I}_1^{\bar{\vec{\Sigma}}}, \hat{J}_2^{\bar{\vec{\Sigma}}}, \hat{J}_3^{\bar{\vec{\Sigma}}}, -\hat{I}_1^{\bar{\vec{\Sigma}}}, -\hat{J}_2^{\bar{\vec{\Sigma}}}, -\hat{J}_3^{\bar{\vec{\Sigma}}}\right)\,.
\label{eq:omega_6_arguments}
\end{equation}
This approach increases the expressivity of the dissipation potential network by additional activations that are monotonic decreasing and convex with respect to the invariants. In this way, the convexity with respect to the invariants is preserved. To ensure that the output dissipation potential is zero-valued at its origin and non-negative, the same transformation step in \autoref{eq:convex_omega} should be carried out.

\textbf{Adding principal invariants.} Additional stress invariants can be added to the input arguments of the dissipation potential network to enhance the expressivity, e.g., the classical invariants $I_2^{\bar{\vec{\Sigma}}}$ and $I_3^{\bar{\vec{\Sigma}}}$. These invariants can be expressed in terms of the stress invariants as 
\begin{equation}
    I_2^{\bar{\vec{\Sigma}}} = \frac{(I_1^{\bar{\vec{\Sigma}}})^2}{6} + J_2^{\bar{\vec{\Sigma}}} \,, \quad I_3^{\bar{\vec{\Sigma}}} = \frac{(I_1^{\bar{\vec{\Sigma}}})^3}{27} + \frac{2}{3} I_1^{\bar{\vec{\Sigma}}} J_2^{\bar{\vec{\Sigma}}} + J_3^{\bar{\vec{\Sigma}}}\,.
\end{equation}
Analogously to the modification of the stress invariants, we take the square root and cubic root of the second and third principal invariants, respectively, and achieve the following modified invariants,
\begin{equation}
    \hat{I}_2^{\bar{\vec{\Sigma}}} = \sqrt{I_2^{\bar{\vec{\Sigma}}}}\,,\quad \hat{I}_3^{\bar{\vec{\Sigma}}} = \sqrt[3]{I_3^{\bar{\vec{\Sigma}}}}\,.
\end{equation}
Thus, the expression of the dissipation potential can be reformulated as
\begin{equation}
\omega = \omega\left(I_1^{\bar{\vec{\Sigma}}},  \hat{J}_2^{\bar{\vec{\Sigma}}}, \hat{J}_3^{\bar{\vec{\Sigma}}}, \hat{I}_2^{\bar{\vec{\Sigma}}} ( I_1^{\bar{\vec{\Sigma}}}, J_2^{\bar{\vec{\Sigma}}} ),\hat{I}_3^{\bar{\vec{\Sigma}}} (I_1^{\bar{\vec{\Sigma}}}, J_2^{\bar{\vec{\Sigma}}}, J_3^{\bar{\vec{\Sigma}}}) \right)\,.
\label{eq:omega_5_arguments}
\end{equation}
The guarantee of thermodynamic consistency is proven in \citet{holthusen2025generalized}. Analogously, the transformation step in \autoref{eq:convex_omega} should follow.

\subsection{Estimation for initial grid range of iCKANs}
\label{app:init_grid_range}

The activation functions in KANs are represented by B-splines defined on fixed grids. In this work, we adopt the spline grid initialization strategy proposed by \citet{thakolkaran2025can}. Accordingly, each input dimension requires a prescribed grid range that specifies the interval over which the activation functions are explicitly represented. Outside this interval, linear extrapolation is applied, which may lead to numerical instabilities if the inputs extend significantly beyond the grid bounds. Consequently, a reliable choice of the grid range for the first KAN layer is essential for stable training and evaluation. Although one could assume a generally sufficiently large range, this would lead to a reduction of training efficiency and resolution.

In contrast to prior work focusing on hyperelasticity \citep{abdolazizi2025constitutive, thakolkaran2025can}, the effective network inputs in the present iCKAN formulation are not directly given by the dataset $(\vec{F},\vec{P})$, but instead arise within a recurrent constitutive update and depend on intermediate model predictions. Their distribution is therefore not known a priori and cannot be determined through a simple preprocessing step. To address this difficulty, we introduce an estimation procedure for suitable grid ranges based on invariant measures computed from the input and output tensors using the training data.

Specifically, the initial grid range for the KAN model representing the Helmholtz free energy $\psi$ is defined based on the invariants of the right Cauchy-Green deformation tensor $\vec{C} = \vec{F}^\mathrm{T}\vec{F}$, while the initial grid range for the KAN model representing the dissipation potential $\omega$ is defined based on the invariants of the product of the right Cauchy-Green deformation tensor $\vec{C}$ and the second Piola-Kirchhoff stress tensor $\vec{S} = \vec{F}^{-1}\vec{P}$. In the following, we provide a justification for the plausibility of this estimation scheme.

\subsubsection{Approximation for initial grid range of KAN for dissipation potential}

The initial grid range of the KAN for dissipation potential should cover the range of the invariants of the symmetric elastic Mandel-like stress tensor $\bar{\vec{\Sigma}}$. The elastic Mandel-like stress $\bar{\vec{\Sigma}}$ and the second Piola-Kirchhoff stress tensor $\vec{S}$ are defined by
\begin{equation}
    \bar{\vec{\Sigma}} = 2 \, \bar{\vec{C}}_e \dfrac{\partial \psi}{\partial \bar{\vec{C}}_e} \quad \text{and} \quad \vec{S} = 2 \, \vec{U}_i^{-1} \dfrac{\partial \psi}{\partial \bar{\vec{C}}_e}  \vec{U}_i^{-1}\,,
\end{equation}
leading to the relation $\bar{\vec{\Sigma}} = \vec{U}_i^{-1} \vec{C} \vec{S} \vec{U}_i$. We thus define $\hat{\vec{\Sigma}} = \vec{C} \vec{S}$ and proove that $\bar{\vec{\Sigma}} = \vec{U}_i^{-1} \hat{\vec{\Sigma}} \vec{U}_i $ and $\hat{\vec{\Sigma}}$ share the same invariants.

\textbf{First stress invariant.} 
\begin{equation}
        I_1^{\bar{\vec{\Sigma}}} = \tr(\bar{\vec{\Sigma}}) = \tr(\vec{U}_i^{-1} \hat{\vec{\Sigma}} \vec{U}_i) = \tr(\hat{\vec{\Sigma}}\vec{U}_i \vec{U}_i^{-1}) = \tr(\hat{\vec{\Sigma}}) = I_1^{\hat{\Sigma}}\,.
\end{equation}
\textbf{Second stress invariant.} 
\begin{equation}
    \begin{split}
        J_2^{\bar{\vec{\Sigma}}} &= \dfrac{1}{2} \tr(\dev(\bar{\vec{\Sigma}})^2) = \dfrac{1}{2} \tr\left(\bar{\vec{\Sigma}}^2 - \dfrac{2}{3} I_1^{\bar{\vec{\Sigma}}} \bar{\vec{\Sigma}} + \dfrac{1}{9} (I_1^{\bar{\vec{\Sigma}}})^2 \vec{I}\right) \\
        &= \dfrac{1}{2} (\tr(\bar{\vec{\Sigma}}^2) - \dfrac{5}{9} (I_1^{\bar{\vec{\Sigma}}})^2) = \dfrac{1}{2} (\tr(\hat{\vec{\Sigma}}^2) - \dfrac{5}{9} (I_1^{\hat{\Sigma}})^2) \\
        &= \dfrac{1}{2} \tr(\dev(\bar{\vec{\Sigma}})^2) = J_2^{\hat{\Sigma}} \\
        &\text{with} \quad \tr(\bar{\vec{\Sigma}}^2) = \tr(\vec{U}_i^{-1} \hat{\vec{\Sigma}} \vec{U}_i\vec{U}_i^{-1} \hat{\vec{\Sigma}} \vec{U}_i) =\tr(\vec{U}_i^{-1} \hat{\vec{\Sigma}} \hat{\vec{\Sigma}} \vec{U}_i)= \tr(\hat{\vec{\Sigma}}^2)\,.
    \end{split}
\end{equation}
\textbf{Third stress invariant.}
\begin{equation}
    \begin{split}
        J_3^{\bar{\vec{\Sigma}}} &= \dfrac{1}{3} \tr(\dev(\bar{\vec{\Sigma}})^3) = \dfrac{1}{3} \tr\left(\bar{\vec{\Sigma}}^3 - I_1^{\bar{\vec{\Sigma}}} \bar{\vec{\Sigma}}^2 + \dfrac{1}{3} (I_1^{\bar{\vec{\Sigma}}})^2 \bar{\vec{\Sigma}} - \dfrac{1}{27} (I_1^{\bar{\vec{\Sigma}}})^3 \vec{I}\right) \\
        &= \dfrac{1}{3} (\tr(\bar{\vec{\Sigma}}^3) - I_1^{\bar{\vec{\Sigma}}} \tr(\bar{\vec{\Sigma}}^2) + \dfrac{8}{27} (I_1^{\bar{\vec{\Sigma}}})^3) = \dfrac{1}{3} (\tr(\hat{\vec{\Sigma}}^3) - I_1^{\hat{\Sigma}} \tr(\hat{\vec{\Sigma}}^2) + \dfrac{8}{27} (I_1^{\hat{\Sigma}})^3) \\
        &= \dfrac{1}{3} \tr(\dev(\hat{\vec{\Sigma}})^3) = J_3^{\hat{\Sigma}} \\
        &\text{with} \quad \tr(\bar{\vec{\Sigma}}^3) = \tr(\vec{U}_i^{-1} \hat{\vec{\Sigma}} \vec{U}_i \vec{U}_i^{-1} \hat{\vec{\Sigma}} \vec{U}_i \vec{U}_i^{-1} \hat{\vec{\Sigma}} \vec{U}_i) =\tr(\vec{U}_i^{-1} \hat{\vec{\Sigma}} \hat{\vec{\Sigma}} \hat{\vec{\Sigma}} \vec{U}_i)= \tr(\hat{\vec{\Sigma}}^3)\,.
    \end{split}
\end{equation}

In conclusion, the approximated grid range using the invariants of $\hat{\vec{\Sigma}} = \vec{C} \vec{S}$ is equivalent to the grid range using the invariants of $\bar{\vec{\Sigma}}$.

\subsubsection{Approximation for initial grid range of KAN for free energy} 

The initial grid range of the KAN for free energy should cover the range of the invariants of the elastic right Cauchy-Green deformation tensor $\bar{\vec{C}}_e$. We approximate the grid range using the invariants of the right Cauchy-Green deformation tensor $\vec{C}$. For one cyclic tension test, the elastic deformation is always smaller than the total deformation, i.e., $\bar{\vec{C}}_e \leq \vec{C}$. Hence, the invariants of $\bar{\vec{C}}_e$ are always smaller than or equal to the invariants of $\vec{C}$, i.e., $I_1^{\bar{\vec{C}}_e} \leq I_1^{\vec{C}}$, $I_2^{\bar{\vec{C}}_e} \leq I_2^{\vec{C}}$ and $I_3^{\bar{\vec{C}}_e} \leq I_3^{\vec{C}}$. Therefore, the grid range using the invariants of $\vec{C}$ always covers the grid range using the invariants of $\bar{\vec{C}}_e$.

It is to note, this only holds for compressible cases. For incompressible materials, we assume this approach is sufficient.

\subsection{Implicit time integration scheme}
\label{app:implicit_iCKAN}

Following the implicit integration scheme, the following nonlinear equation hast to be solved at each timestep $t$ to get the current inelastic stretch tensor $\vec{U}_{i,t}$,
\begin{equation}
    \vec{r} := \vec{C}_{i,t-1} - \vec{U}_{i,t} \exp(- 2 \,\Delta t\, \vec{D}_{i,t}) \vec{U}_{i,t} \overset{!}{=} \vec{0}\,,
    \label{eq:implicit_evo}
\end{equation}
where $\vec{D}_{i,t}$ depends nonlinearly on $\vec{U}_{i,t}$ \citep{holthusen2026complement}. Classically, the initial guess for $\vec{U}_{i,t}^{(0)}$ is chosen as the previous timestep value $\vec{U}_{i,t-1}$ and the equation is solved iteratively using, e.g., the Newton-Raphson method, until the residual $\vec{r}$ is sufficiently small. In comparison to the explicit time integration scheme, this approach is more stable but however time consuming. 

\subsubsection{Helper network for solving implicit evolution equation} 

To overcome the drawbacks of the implicit scheme and accelerate the training, the iterative solver can be replaced by a helper neural network $\mathcal{N}_f$ to directly predict $\vec{U}_{i,t}$ \citep{as2023mechanics, rosenkranz2024viscoelasticty}. This helper network can be represented by a Liquid Time-Constant Network (LTC) \citep{hasani2021liquid} to predict the current inelastic stretch \citep{holthusen2026complement}. For instance, we need to compute a trial value of $\vec{D}_{i,t}$, $\vec{D}_{i,t}^{\text{trial}}$, which is $\vec{D}_i$ evaluated with the current $\vec{C}_t$ but the last $\vec{U}_{i,t-1}$. Note that the outcome of the network generally not statisfy the symmetric positive definite condition of $\vec{U}_{i,t}$. To enforce this condition, it is of advantage to work with the lower triangular matrix $\vec{L}_{i}$ of the Cholesky decomposition of $\vec{U}_{i}$, i.e., $\vec{U}_{i} = \vec{L}_{i} \vec{L}_{i}^\mathrm{T}$ \citep{benoit1924note}. The helper network is then designed to predict $\vec{L}_{i,t}$ as
\begin{equation}
    \vec{L}_{i,t} = \mathcal{N}_f(\vec{L}_{i,t-1}, \vec{C}_t, \bar{\vec{D}}_{i,t}^{\text{trial}})\,,
\end{equation}
following with $\vec{U}_{i,t} = \vec{L}_{i,t} \vec{L}_{i,t}^\mathrm{T}$. 

To evaluate the accuracy of the predicted $\vec{U}_{i,t}$, the loss function is enhanced by a term corresponding to the evolution equation, namely the residual of the implicit scheme in \autoref{eq:implicit_evo},
\begin{equation}
  \mathcal{L}_\text{total} = \mathcal{L}_\text{stress} + \lambda_\text{evo} \cdot \mathcal{L}_\text{evo} + \lambda_\text{L1} \cdot \mathcal{L}_{\text{L1}} \quad\text{with}\quad \mathcal{L}_\text{evo} = \mathrm{MSE}(\vec{r})
    \label{eq:loss_implicit}
\end{equation}
with $\lambda_\text{evo}$ being a hyperparameter to balance the two loss parts to a similar scale. The overall architecture of the implicit iCKAN with a helper LTC is illustrated in \autoref{fig:iCKAN_implicit}. The algorithm for the implicit iCKANs with helper network is summarized in Appendix \ref{app:implicit_iCKAN}.

\begin{figure}[H]
    \centering
    \includestandalone{standalone_iCKAN_implicit}
    \caption{Implicit iCKAN architecture with helper LTC to solve the evolution equation at timestep $t$. The inputs at each step are the current deformation gradient and time increment $(\vec{F}_t,\Delta t)$ together with the state variables $(\vec{C}_{t-1},\vec{U}_{i,t-1})$ from the previous step. The updated state variables is achieved by solving \autoref{eq:implicit_evo} using a helper LTC network. The updated state is then propagated to the next time step. In addition, the output stress $\vec{P}$ is computed by evaluating the free energy function $\psi$ with the updated state variables.}
    \label{fig:iCKAN_implicit}
\end{figure}

\textbf{Liquid time constant networks.} LTCs are a type of recurrent neural networks that is designed to model temporal dependencies in sequential data. A LTC is based on the following ordinary differential equation (ODE)
\begin{equation}
    \dfrac{d\vec{h}}{dt} = f(\vec{x}, \vec{h}) - \alpha (\vec{x}, \vec{h}) \, \vec{h}(t)\,,
\end{equation}
where $\vec{h}$ is the hidden state, $\vec{x}$ is the input, and $\alpha$ is a learnable parameter that controls the decay rate of the hidden state. Consequently, a LTC consists of two networks, the source network $\mathcal{N}_f$ to learn the update function $f$, and a second network $\mathcal{N}_{\alpha}$ to learn the state-dependent time constant $\alpha$, which is positive and scaled by the timestep $\Delta t$. This parameter controls the dynamics of the network, specifically, the larger $\alpha$, the faster the state changes. Thus, the original forward Euler scheme
\begin{equation}
    \vec{h}_{t+1} = \vec{h}_t + \Delta t \cdot f(\vec{h}_t, \vec{x})
\end{equation}
is modified into a semi-explicit scheme, which reads
\begin{equation}
    \vec{h}_{t+1} = \vec{h}_t + \Delta t \cdot \left( f(\vec{h}_t, \vec{x}) - \alpha_t \, \vec{h}_t \right)\,.
\end{equation}
In the present framework, the current hidden state is the previous lower Cholesky decompsition of the inelastic stretch tensor $\vec{L}_{i,t-1}$ and the current state is the current right Cauchy-Green tensor $\vec{C}_t$ and the trial inelastic strain rate $\vec{D}_{i,t}^{\text{trial}}$, i.e.,  
\begin{equation}
    \vec{L}_{i,t} = \vec{L}_{i,t-1} + \Delta t \cdot \left( \mathcal{N}_f(\vec{L}_{i,t-1}, \vec{C}_t, \bar{\vec{D}}_{i,t}^{\text{trial}}) - \mathcal{N}_{\alpha}(\vec{L}_{i,t-1}, \vec{C}_t, \bar{\vec{D}}_{i,t}^{\text{trial}}) \vec{L}_{i,t-1} \right)\,.
\end{equation}

\subsubsection{Pseudo code for iCKAN with implicit time integration and helper network}

\begin{algorithm}[H]
\caption{iCKAN with implicit time integration and helper network}
\label{alg:implicit_iCKAN}
\begin{algorithmic}[1]
\REQUIRE $\Delta t_t$, $\vec{F}_t$, $\vec{f}$
\STATE Retrieve $\vec{C}_{t-1}$, $\vec{U}_{i,t-1}$ from history

\STATE Get trial value of evolution equation $\vec{D}_{i,t}^{\text{trial}}$
\STATE $\vec{C}_{t} \gets \vec{F}_t^T \vec{F}_t$
\STATE $\bar{\vec{C}}_{e,t}^{\text{trial}} \gets \vec{U}_{i,t-1}^{-1} \vec{C}_{t} \vec{U}_{i,t-1}$
\STATE $\psi \gets \mathrm{cKAN}(I_{1}^{\bar{\vec{C}}_{e,t}^{\text{trial}}}, I_{2}^{\bar{\vec{C}}_{e,t}^{\text{trial}}}, I_{3}^{\bar{\vec{C}}_{e,t}^{\text{trial}}})$ \hfill \textit{KAN for free energy}
\STATE $\bar{\vec{\Sigma}}^{\text{trial}} \gets 2 \bar{\vec{C}}_{e,t}^{\text{trial}} \dfrac{\partial \psi}{\partial \bar{\vec{C}}_{e,t}^{\text{trial}}}$
\STATE $\omega \gets \mathrm{cKAN}(I_1^{\bar{\vec{\Sigma}}^{\text{trial}}}, \sqrt{J_2^{\bar{\vec{\Sigma}}^{\text{trial}}}}, \sqrt[3]{J_3^{\bar{\vec{\Sigma}}^{\text{trial}}}})$ \hfill \textit{KAN for dissipation potential}
\STATE $\vec{D}_i^{\text{trial}} \gets \dfrac{\partial \omega}{\partial \bar{\vec{\Sigma}}^{\text{trial}}}$
\STATE Predict $\vec{L}_{i,t} \gets f_{\text{NN}}(\vec{U}_{i,t-1}, \vec{C}_t, \vec{D}_{i,t}^{\text{trial}})$
\STATE $\vec{U}_{i,t} \gets \vec{L}_i \vec{L}_i^\mathrm{T}$
\STATE Get predicted value of evolution equation $\vec{D}_{i,t}$
\STATE $\vec{D}_i \gets \dfrac{\partial \omega}{\partial \bar{\vec{\Sigma}}}$
\STATE Get residual of evolution equation:
    $\vec{r} := \vec{C}_{i,t-1} - \vec{U}_{i,t} \exp(- 2 \,\Delta t\, \vec{D}_{i,t}) \vec{U}_{i,t} \overset{!}{=} \vec{0}$
\STATE $\bar{\vec{C}}_{e,t} \gets \vec{U}_{i,t}^{-1} \vec{C}_t \vec{U}_{i,t}$
\STATE Compute invariants: $I_{1,e}, I_{2,e}, I_{3,e} \gets \mathrm{Invariants}(\bar{\vec{C}}_{e,t})$
\STATE $\psi \gets \mathrm{cKAN}(I_{1,e}, I_{2,e}, I_{3,e})$ \hfill \textit{KAN for free energy}
\STATE $\vec{S}_t \gets 2 \vec{U}_i^{-1} \dfrac{\partial \psi}{\partial \bar{\vec{C}}_{e,t}} \vec{U}_i^{-1}$
\STATE $\vec{P}_t \gets \vec{F}_t \vec{S}_t $
\STATE Update history: store $\vec{C}_t$, $\vec{U}_{i,t}$
\RETURN $\vec{P}_t$
\end{algorithmic}
\end{algorithm}

\subsection{Candidate functions for symbolification}
\label{app:candidate_functions}

Candidate functions that are convex and non-decreasing on the interval $[0, \infty)$ \citep{thakolkaran2025can}:
\begin{center}
\begin{tabularx}{\textwidth}{>{\centering\arraybackslash}X}
\hline
$x$, $x^{1.5}$, $x^2$, $x^{2.5}$, $x^3$, $x^{3.5}$, $x^4$, $x^{4.5}$, $x^5$\\
$e^x$, $\log(e^x)$, $\log(e^x)^2$, $\log(e^x)^3$\\\hline
\end{tabularx}
\end{center}

Candidate functions that are convex, zero-valued and non-negative on the interval $(-\infty, \infty)$:
\begin{center}
\begin{tabularx}{\textwidth}{>{\centering\arraybackslash}X}
\hline
$|x|$, $|x|^{1.5}$, $x^2$, $|x|^{2.5}$, $|x|^3$, $|x|^{3.5}$, $x^4$, $|x|^{4.5}$, $|x|^5$\\
$e^{x^2}-1$, $e^{x^4}-1$\\
$\mathrm{cosh}(x)-1$, $\mathrm{cosh}(x^2)-1$\\
$\log(\mathrm{cosh}(x))$, $\log(\mathrm{cosh}(x^2))$\\\hline
\end{tabularx}
\end{center}

\subsection{Additional results of iCKANs on synthetic data}
\label{app:results_synthetic}

In this section, we show the additional results of the iCKANs trained on synthetic data. Using the explicit time integration, we examine the options where the input arguments of the KAN for dissipation potential is augmented with the negative stress invariants or additional principal invariants. We further examine the performance of iCKANs with implicit time intergration scheme Appendix~\ref{app:implicit_iCKAN}. We provide the details on the performance of each variant and compare the convergence between the explicit and implicit iCKANs.

\textbf{Explicit iCKAN: Full result.} The prediction of the symbolified version of the explicit iCKAN on the entire dataset is shown in \autoref{fig:results_uniaxial_cycle_31_31_L1}

\begin{figure}[H]
    \centering
    \includegraphics{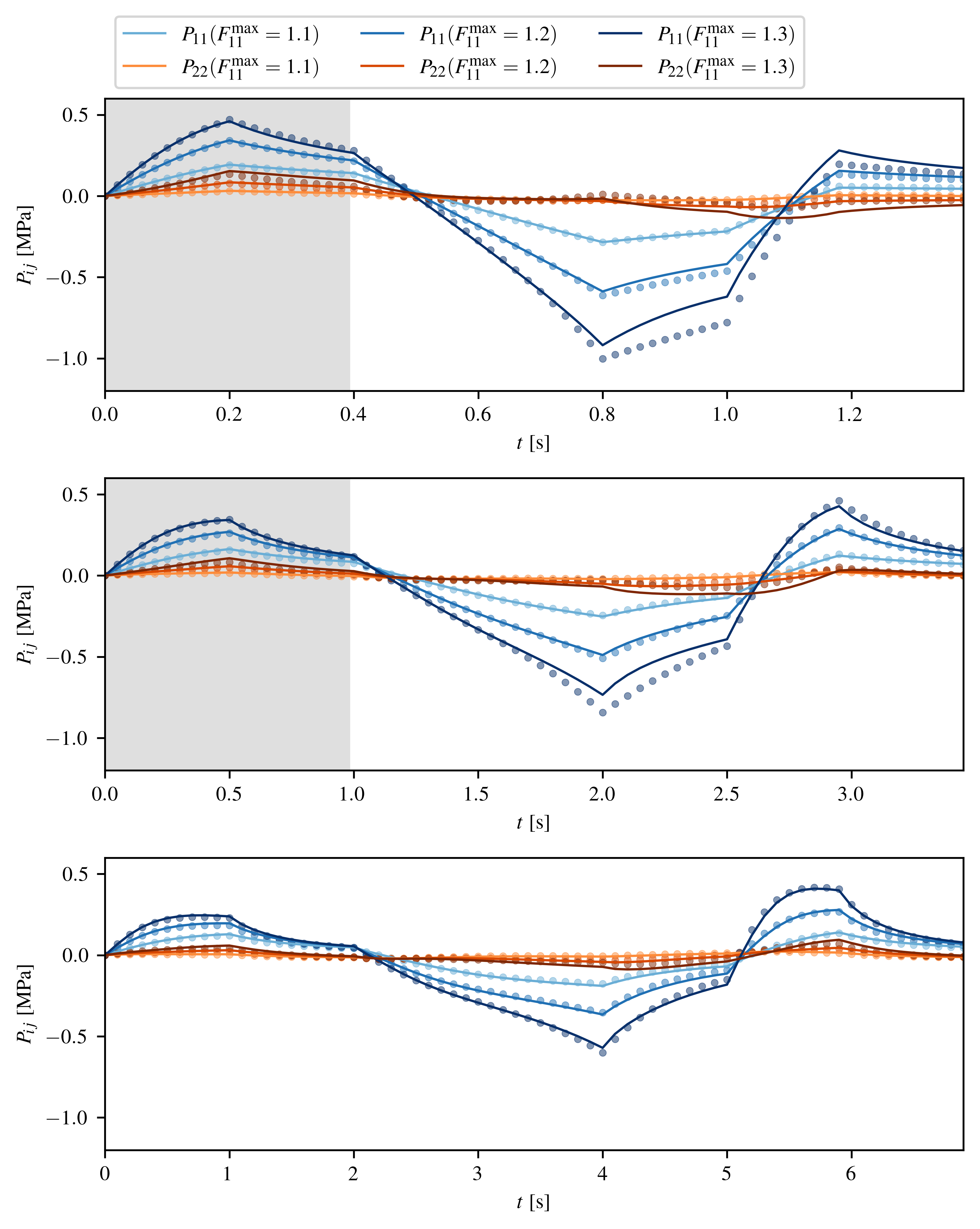}
    \put(-70,378){ \fbox{$t_\text{load} = 0.2$ s}}
    \put(-70,212){ \fbox{$t_\text{load} = 0.5$ s}}
    \put(-67,50){ \fbox{$t_\text{load} = 1$ s}}
    \put(-379,365){\includegraphics{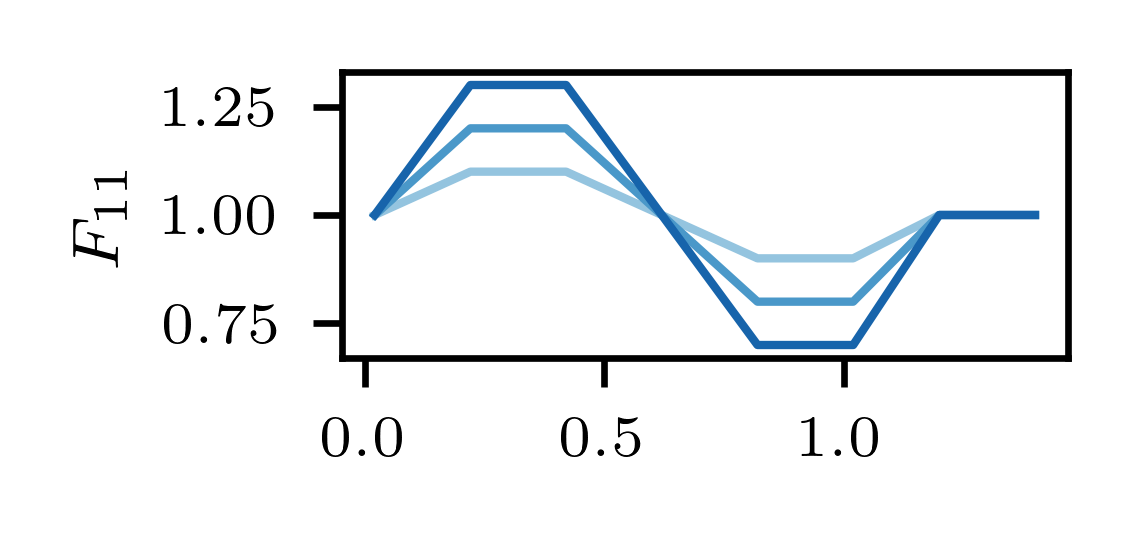}}
    \put(-379,200){\includegraphics{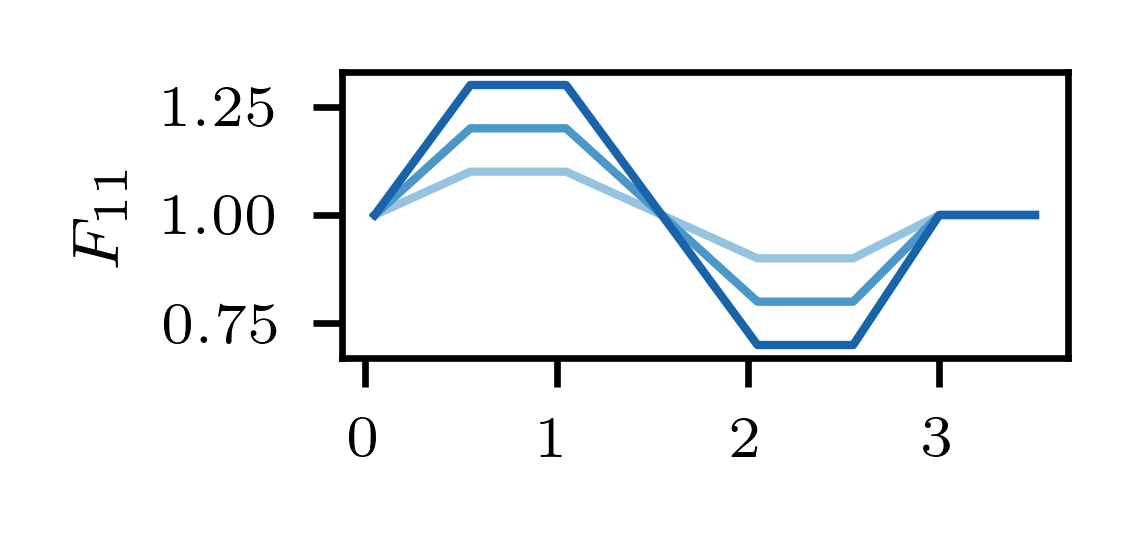}}
    \put(-379,30){\includegraphics{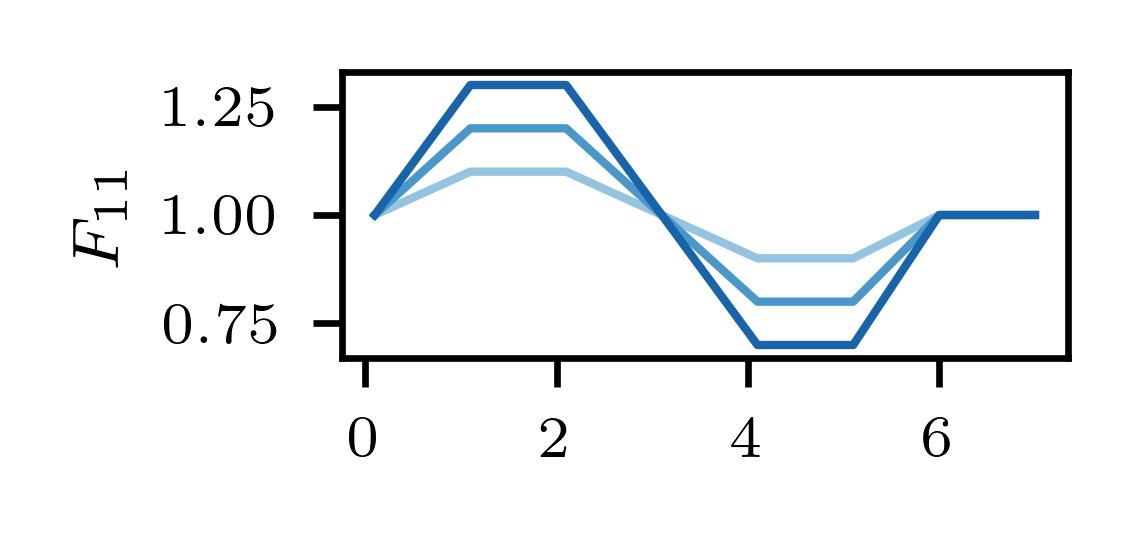}}
\caption{Predicted first Piola-Kirchhoff stress by the symbolified explicit iCKAN. The gray area shows the range of the training data and the dots with corresponding color indicate the reference data.}
\label{fig:results_uniaxial_cycle_31_31_L1}
\end{figure}

\textbf{Explicit iCKAN: Adding negative stress invariants to arguments of inelasitc potential.} We could enhance the arguments for the dissipation potential with the negative stress invariants, as shown in \autoref{eq:omega_6_arguments}. This approach provides higher expressibility of the network, but unfortunately shows lower stability in prediction.

To achieve more interpretable symbolifc expressions of the inelasitc potential we could symbolify the combined activations for the positive and negative parts of each input argument together, i.e. 
\begin{equation}
    \begin{split}
    &y^{\text{symb}}_1 = y_1^+(\hat{I}_1^{\bar{\vec{\Sigma}}})+y_1^-(-\hat{I}_1^{\bar{\vec{\Sigma}}})\,,\\
    &y^{\text{symb}}_2 = y_2^+(\hat{J}_2^{\bar{\vec{\Sigma}}})+y_2^-(-\hat{J}_2^{\bar{\vec{\Sigma}}})\,,\\
    &y^{\text{symb}}_3 = y_3^+(\hat{J}_3^{\bar{\vec{\Sigma}}})+y_3^-(-\hat{J}_3^{\bar{\vec{\Sigma}}}) \,.
    \end{split}
\end{equation} 
The candidate functions are general convex functions that are not necessarily non-decreasing. The architecture of the KAN for dissipation potential before and after the symbolification is shown in \autoref{fig:KAN_61}. The last three activations are pruned to be zero since their contribution to the output is included in the first three activations. It is to note, that the $\mathcal{H}$-transformation has to be applied to the network outcome to achieve convex, non-negative and zero-valued at origin dissipation potential.

\begin{figure}[H]
    \vspace{2em}
    \hspace{4em}
    \small
    \subfloat[B-spline expression of KAN for dissipation potential\label{subfig-1:omega_KAN_61_spline}]{
    \begin{overpic}[width=.3\textwidth]{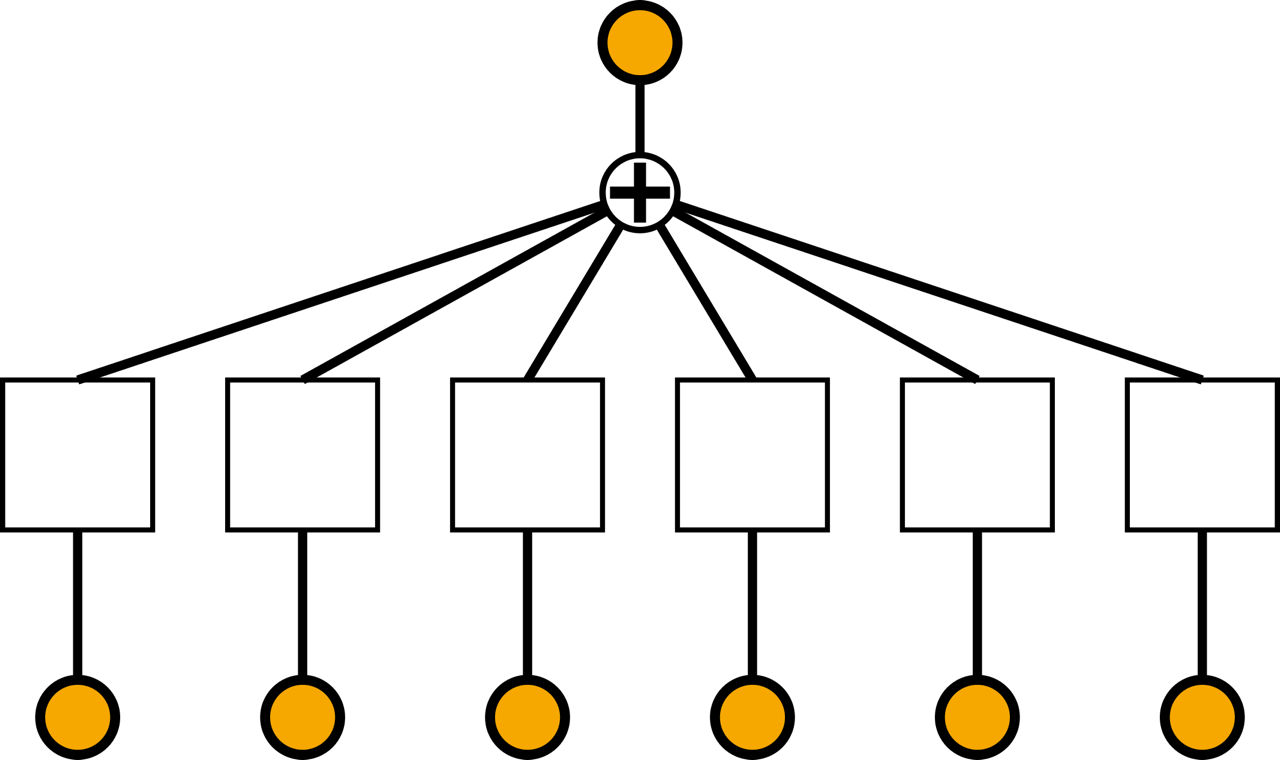}
        \put(-1,17){\includegraphics[width=0.045\textwidth]{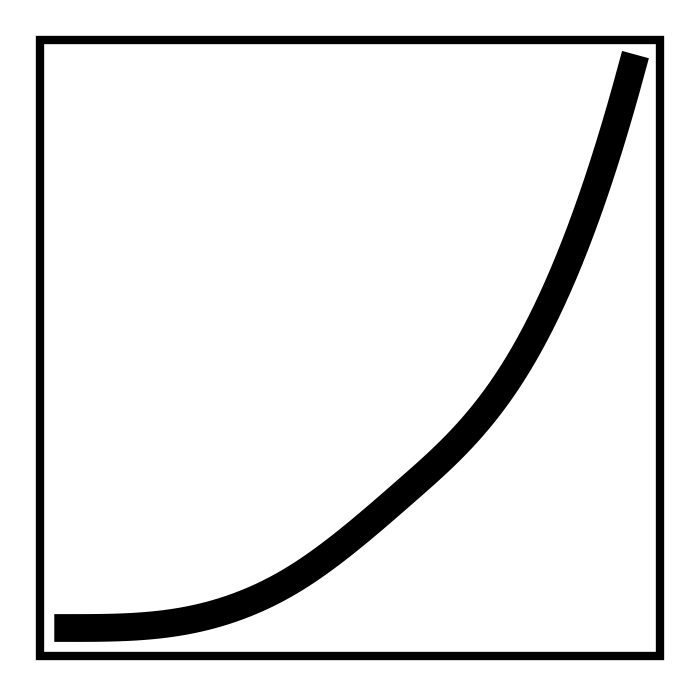}}
        \put(16,17){\includegraphics[width=0.045\textwidth]{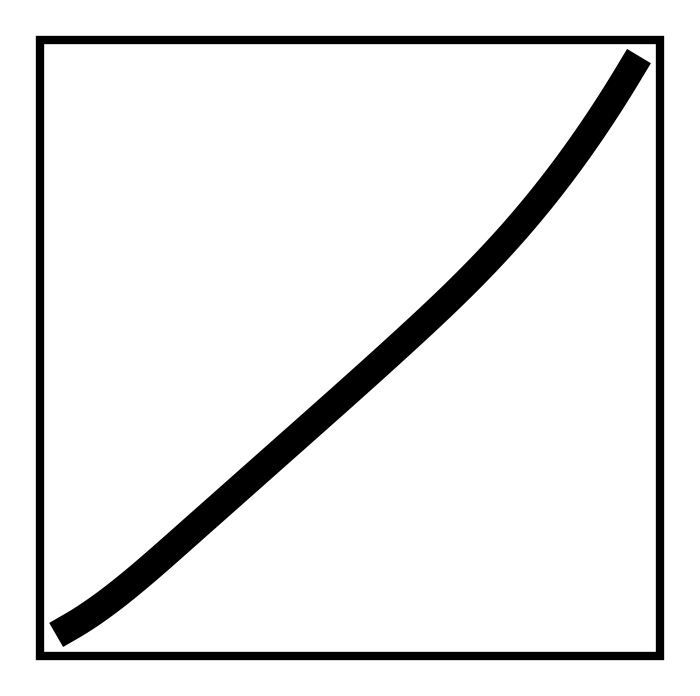}}
        \put(33,17){\includegraphics[width=0.045\textwidth]{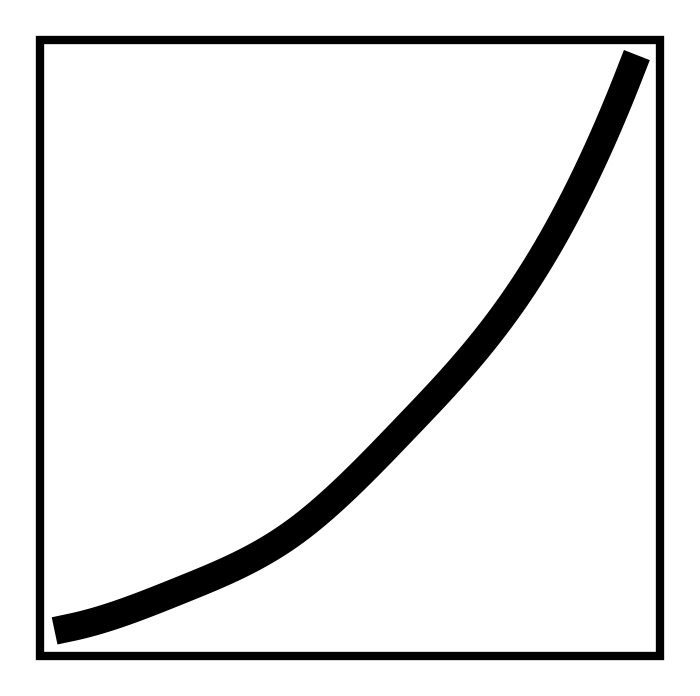}}
        \put(51,17){\includegraphics[width=0.045\textwidth]{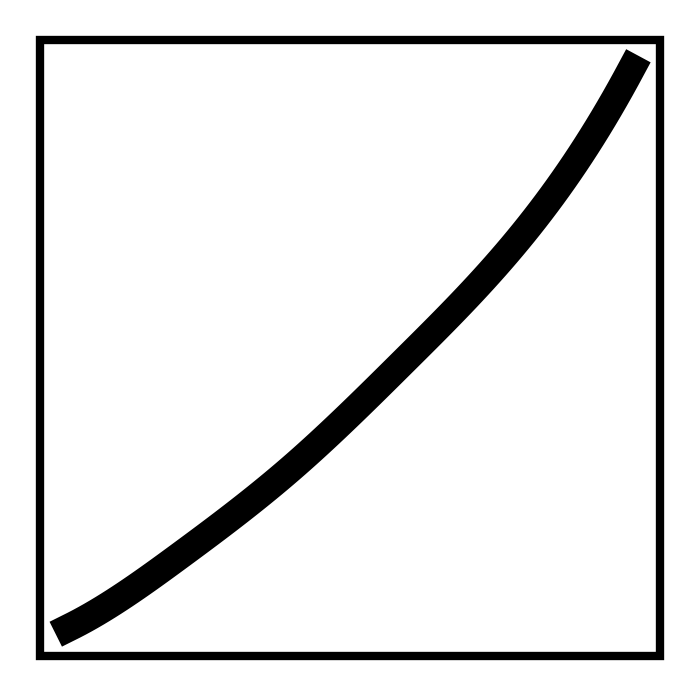}}
        \put(68,17){\includegraphics[width=0.045\textwidth]{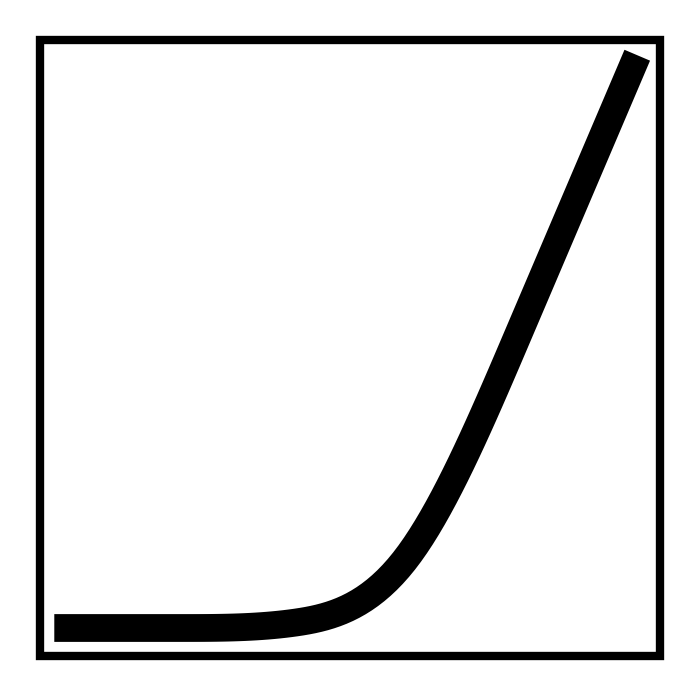}}
        \put(86,17){\includegraphics[width=0.045\textwidth]{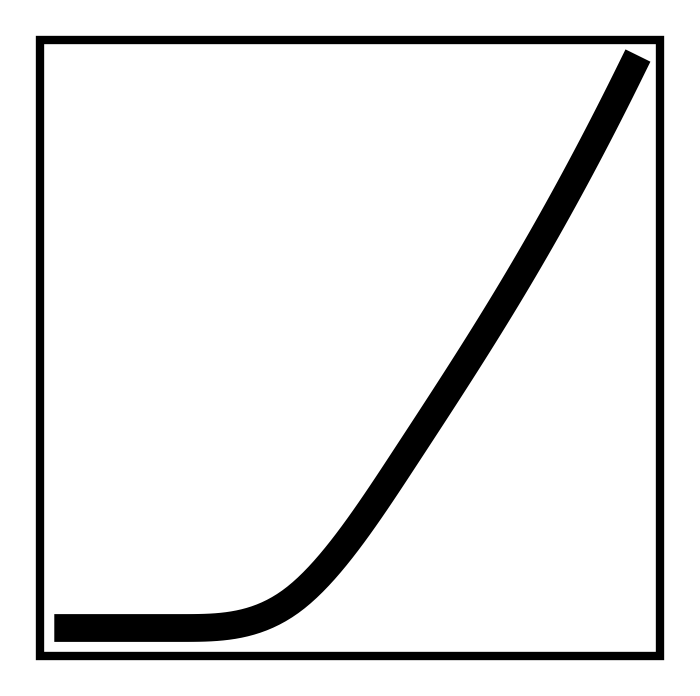}}
        \put(3,-10){$\hat{I}_1^{\bar{\vec{\Sigma}}}$}
        \put(18,-10){$\hat{J}_2^{\bar{\vec{\Sigma}}}$}
        \put(35,-10){$\hat{J}_3^{\bar{\vec{\Sigma}}}$}
        \put(51,-10){$-\hat{I}_1^{\bar{\vec{\Sigma}}}$}
        \put(70,-10){$-\hat{J}_2^{\bar{\vec{\Sigma}}}$}
        \put(89,-10){$-\hat{J}_3^{\bar{\vec{\Sigma}}}$}
        \put(50,64){$\omega^\text{KAN}$}
    \end{overpic}
    \vspace{2em}
    }
    \hfill
    \begin{minipage}[c]{0.1\textwidth}
        \centering
        \vspace{-8em}
        \raisebox{2.5em}{\large\boldmath$\overset{\text{Symb.}}{\longrightarrow}$}
    \end{minipage}
    \hfill
    \subfloat[Symbolic expression of KAN for dissipation potential\label{subfig-1:omega_KAN_61_symb}]{
    \begin{overpic}[width=.3\textwidth]{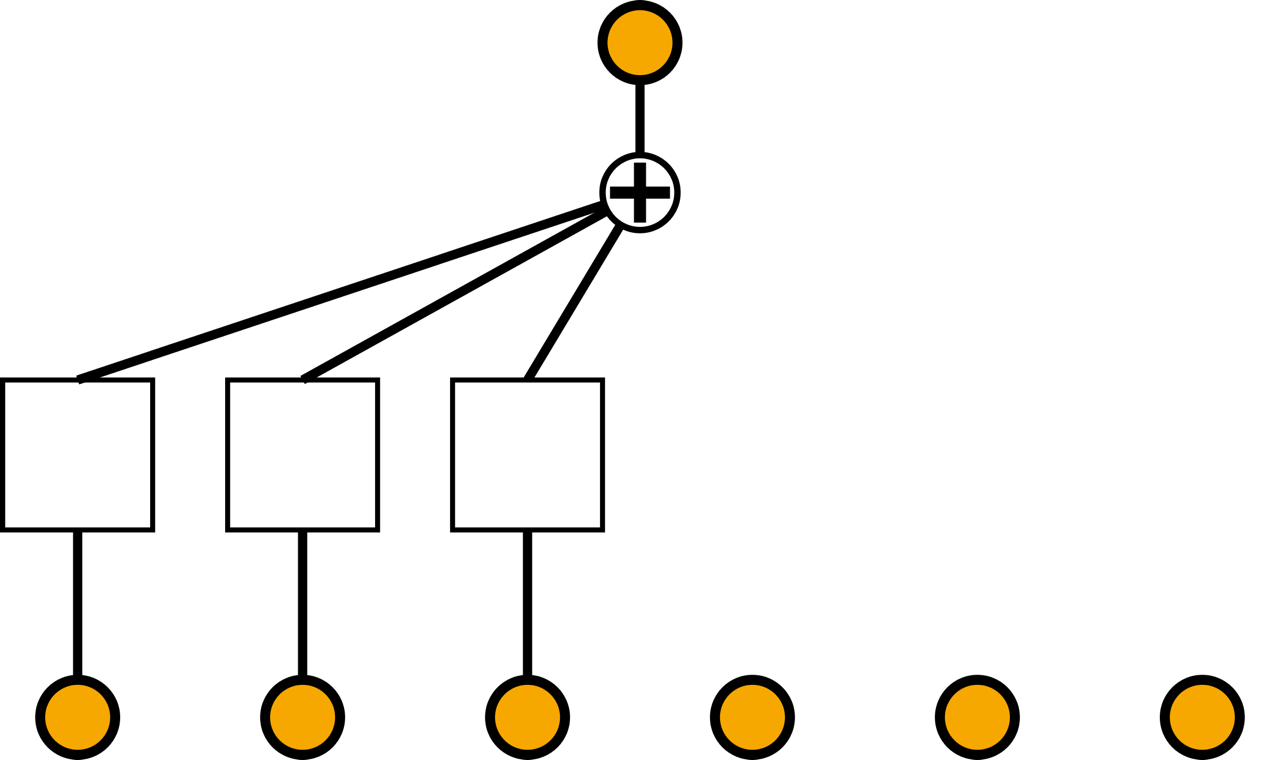}
        \put(-1,17){\includegraphics[width=0.045\textwidth]{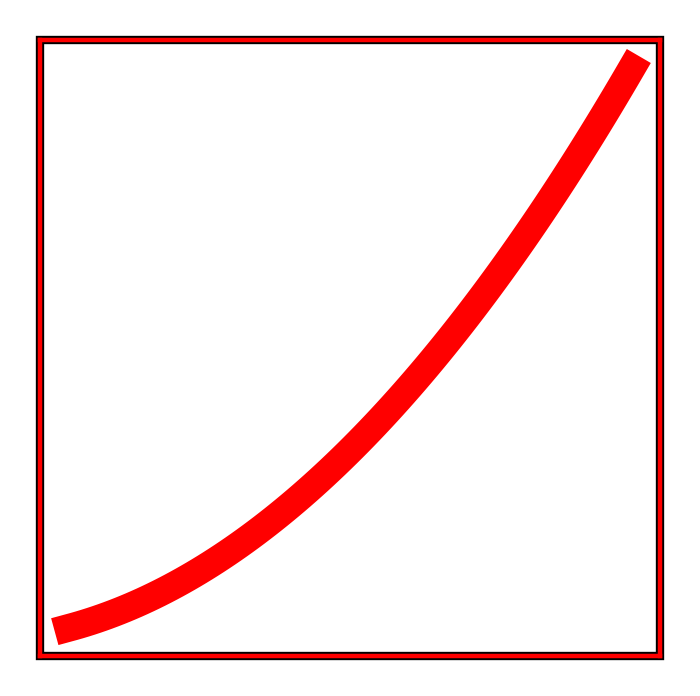}}
        \put(17,17){\includegraphics[width=0.045\textwidth]{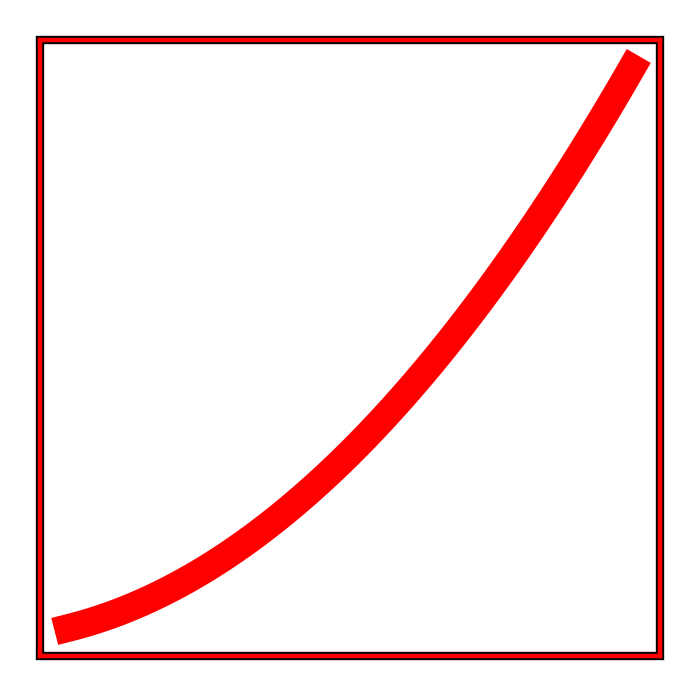}}
        \put(35,17){\includegraphics[width=0.045\textwidth]{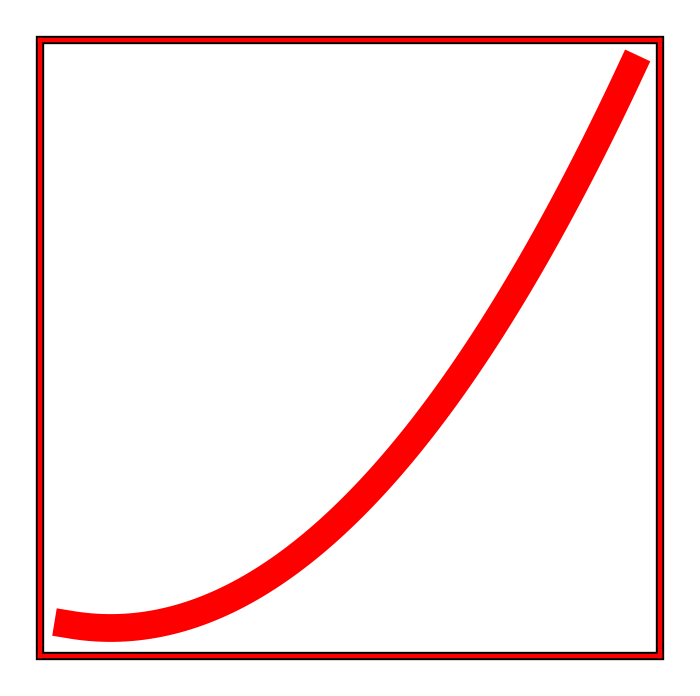}}
        \put(3,-10){$\hat{I}_1^{\bar{\vec{\Sigma}}}$}
        \put(18,-10){$\hat{J}_2^{\bar{\vec{\Sigma}}}$}
        \put(35,-10){$\hat{J}_3^{\bar{\vec{\Sigma}}}$}
        \put(51,-10){$-\hat{I}_1^{\bar{\vec{\Sigma}}}$}
        \put(70,-10){$-\hat{J}_2^{\bar{\vec{\Sigma}}}$}
        \put(89,-10){$-\hat{J}_3^{\bar{\vec{\Sigma}}}$}
        \put(50,64){$\omega^\text{KAN}$}
    \end{overpic}
    \vspace{2em}
    }
    \hspace{4em}

    \caption{KAN for dissipation potential using both positive and negative stress invariants as arguments (a) before and (b) after the symbolification of the activations.}
    \label{fig:KAN_61}
\end{figure}


\textbf{Explicit iCKAN: Adding principal invariants to arguments of dissipation potential.} We could enhance the arguments of the KAN for dissipation potential with the principal invariants, as shown in \autoref{eq:omega_5_arguments}. 


\textbf{Implicit iCKAN.} We consider the implicit iCKAN variant using three modified stress invariants as inputs to the KAN for dissipation potential (\autoref{eq:omega_3_arguments}). The accuracy of the inelastic stretch predicted by the helper LTC network is compared with the Newton–Raphson solution of the implicit evolution equation in \autoref{fig:PK2_S11_maxwell_implicit_helper} for the case of $F_{11}^{\max} = 1.3$ and $\dot{F}_{11} = 0.6$ s${-1}$. Although trained only on the gray-marked range, the LTC shows good generalization to larger deformation levels. 

\begin{figure}[H]
    \centering
    \includegraphics{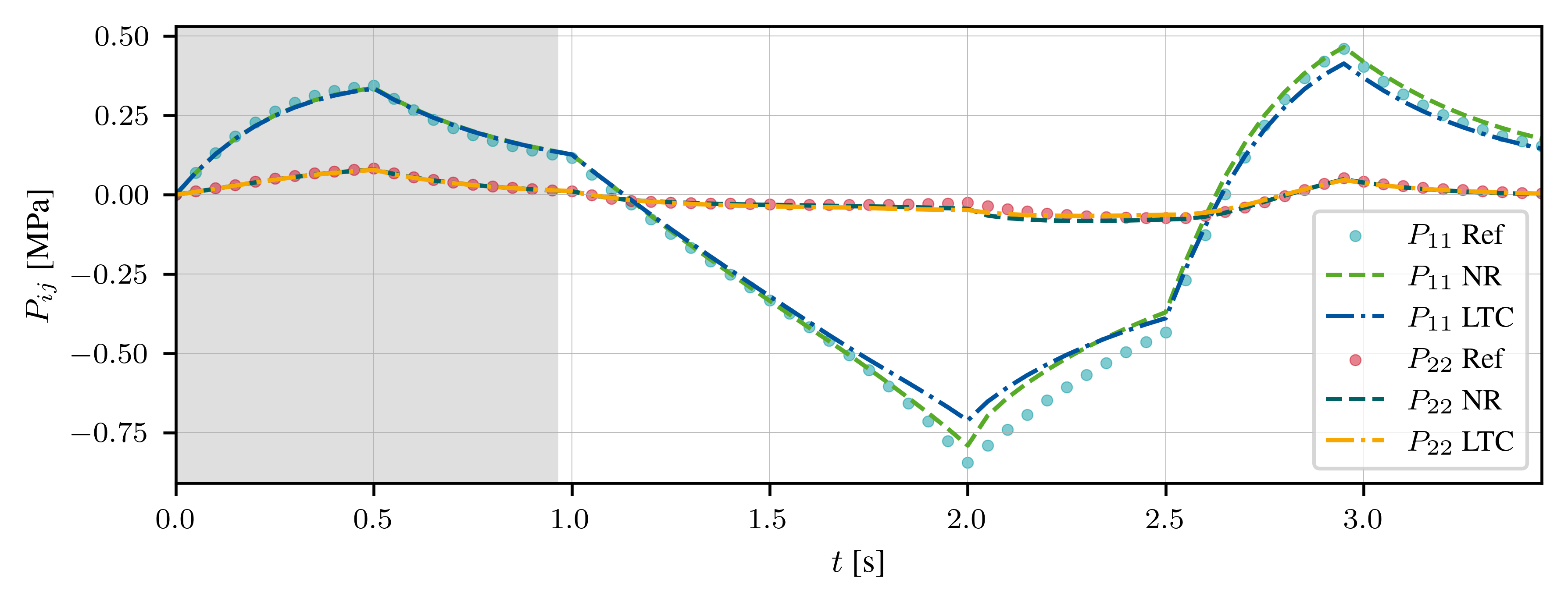}
    \put(-381,35){\includegraphics{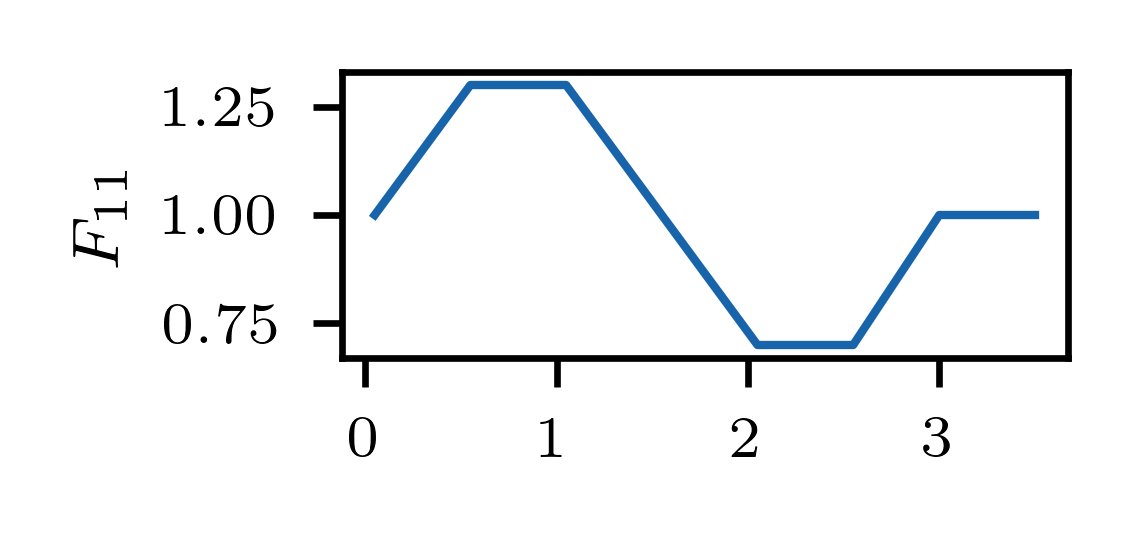}}
    \put(-310,37){\footnotesize $t$ [s]}
    \caption{Stress prediction of the implicit iCKAN. During training, the evolution equation is solved with a trainable LTC helper network. The gray region indicate the range of training data. Predictions using the trained LTC are compared with solutions obtained by a Newton–Raphson solver.}
\label{fig:PK2_S11_maxwell_implicit_helper}
\end{figure}

\textbf{Comparison between the explicit and implicit iCKAN.} For the case where three modified stress invariants are used as input arguments for the KAN for dissipation potential, the convergence behavior of the explicit and implicit time integration schemes is compared in \autoref{fig:loss_maxwell_en31_pot31}. 
\begin{figure}[H]
\centering
   \includegraphics[width=\linewidth]{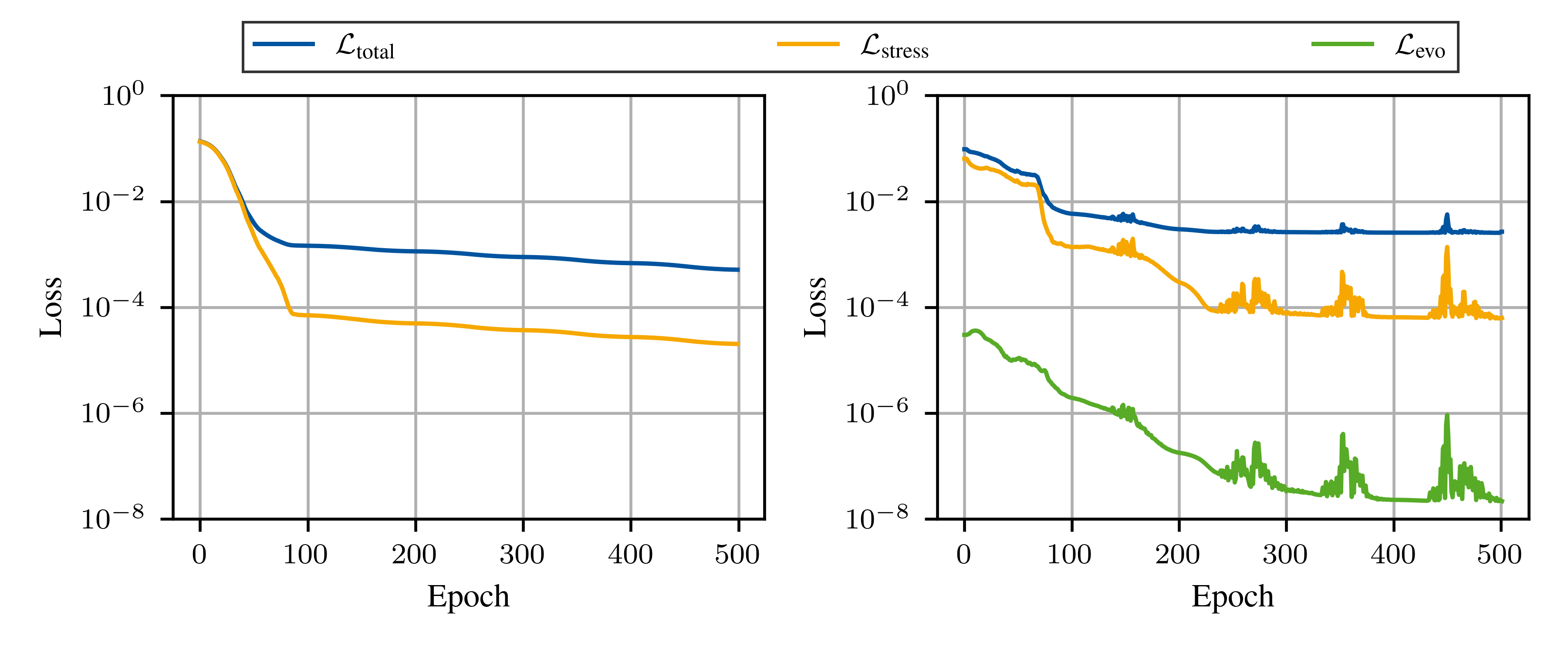}
    \put(-375,-10){(a) Explicit training}
    \put(-190,-10){(b) Implicit training with helper LTC}
    \caption{Training loss of the iCKAN model with L1 regularization using (a) the explicit integration scheme and (b) the implicit integration scheme with helper LTC to solve the evolution equation.}
    \label{fig:loss_maxwell_en31_pot31}
\end{figure}

\textbf{Detailed results.} \autoref{tab:results_synthetic} shows the detailed results for both explicit and implicit iCKAN before and after symbolification. The quantities $\mathcal{L}_{\text{stress}}$ and $\mathcal{L}_{\text{test}}$ denote the NMSE of the predicted stress on the training and full datasets, respectively. $\mathcal{L}_{\text{symb}}$ is the NMSE of the symbolified network on the full dataset, and $\mathcal{L}_{\text{evo}}$ is the loss of the implicit evolution equation when a helper LTC is used during training.

\begin{table}[H]
\caption{Detailed results of training and testing of iCKANs on the synthetic dataset.}
\label{tab:results_synthetic}
\centering\small
\begin{tabular}{lccccc}
    \hline iCKAN & $\omega^{\text{KAN}}(\bullet)$ &
    $\mathcal{L}_{\text{stress}}$
    & $\mathcal{L}_{\text{evo}}$
    & $\mathcal{L}_{\text{test}}$& $\mathcal{L}_{\text{symb}}$  \\
    \hline
    Explicit & $(\hat{I}_1^{\bar{\vec{\Sigma}}},\hat{J}_2^{\bar{\vec{\Sigma}}},\hat{J}_3^{\bar{\vec{\Sigma}}})$
    & $2.0 \cdot 10^{-5}$ & - & $6.3 \cdot 10^{-4}$ & $3.8 \cdot 10^{-4}$ \\
    Explicit & $(\hat{I}_1^{\bar{\vec{\Sigma}}},\hat{J}_2^{\bar{\vec{\Sigma}}},\hat{J}_3^{\bar{\vec{\Sigma}}}, \hat{I}_2^{\bar{\vec{\Sigma}}},\hat{I}_3^{\bar{\vec{\Sigma}}})$ 
    & $4.2 \cdot 10^{-4}$   & - & $9.7 \cdot 10^{-4}$  &-\\
    Implicit (LTC) & $(\hat{I}_1^{\bar{\vec{\Sigma}}},\hat{J}_2^{\bar{\vec{\Sigma}}},\hat{J}_3^{\bar{\vec{\Sigma}}})$ 
    & $6.9 \cdot 10^{-5}$ & $2.9 \cdot 10^{-8}$ & $7.3 \cdot 10^{-4}$  &-\\
    Implicit (NR) & $(\hat{I}_1^{\bar{\vec{\Sigma}}},\hat{J}_2^{\bar{\vec{\Sigma}}},\hat{J}_3^{\bar{\vec{\Sigma}}})$ 
    & -      & -     & $3.4 \cdot 10^{-4}$ & $3.9 \cdot 10^{-4}$ \\\hline
\end{tabular}
\end{table}

\color{author}
\subsection{Time increment sensitivity}
\label{app:sensitivity_analysis_time_increment}

Furthermore, a sensitivity analysis with respect to the time increment is performed on the synthetic dataset. An explicit iCKAN model is trained using the three modified stress invariants as inputs to the dissipation potential (\autoref{eq:omega_3_arguments}). The analysis considers the loading case $F_{11}^{\max}=1.3$ and $\dot{F}_{11}=0.6~\mathrm{s}^{-1}$. During training, a time increment of $\Delta t = 0.5~\mathrm{s}$ is employed.

To investigate the influence of the temporal discretization, the trained model is evaluated using time increments of \mbox{$\Delta t = \{0.01, 0.05, 0.1\}~\mathrm{s}$}. For each time increment, the evolution equation is solved using both explicit and implicit time-integration schemes. The resulting stress predictions are shown in \autoref{fig:PK2_S11_maxwell_time_increment}.

As expected, the explicit integration scheme exhibits an increased sensitivity to larger time increments, whereas the implicit scheme remains comparatively robust. The observed error variations are therefore primarily attributed to the numerical properties of the time integration schemes rather than limitations of the iCKAN architecture itself. Since the time increment $\Delta t$ is provided as an input to the model, the trained iCKAN can consistently account for different temporal resolutions, provided that the integration scheme remains numerically stable.

\begin{figure}[H]
    \centering
    \includegraphics{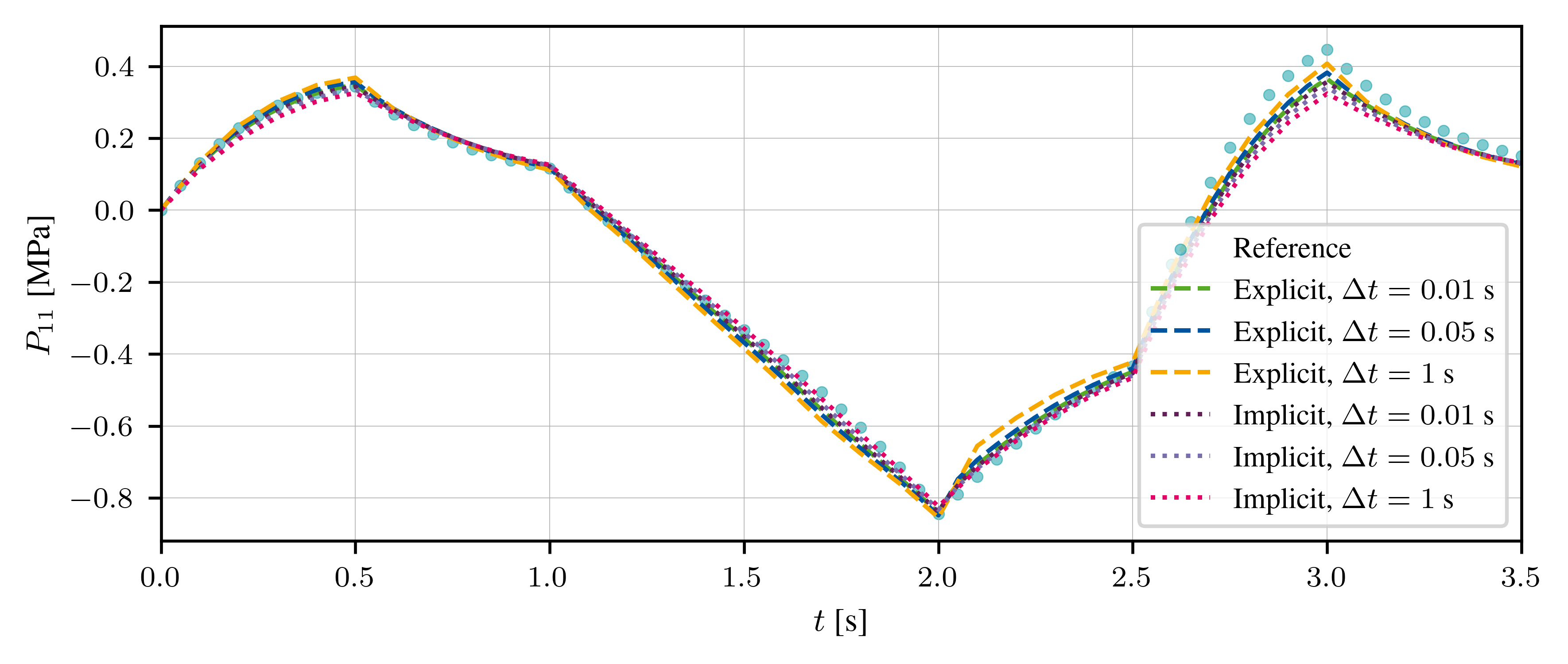}
    \put(-381,35){\includegraphics{results_synthetic_Uniaxial_Cycle_F11_vs_time_t05_f03.png}}
    \put(-310,37){\footnotesize $t$ [s]}
    \caption{Time-increment sensitivity analysis of the iCKAN. Stress predictions obtained with explicit and implicit time-integration schemes for different time increments $\Delta t$, demonstrating the influence of temporal resolution on prediction accuracy and numerical stability.}
\label{fig:PK2_S11_maxwell_time_increment}
\end{figure}
\color{black}

\subsection{Additional data for training iCKANs on experimental data}
\label{app:info_experiment}

\textbf{Hyperparameters.} \autoref{tab:hyperparameters} lists the hyperparameters of the iCKANs from Section \ref{sec:results} for the material model discovery of VHB 4910 and VHB 4905.
\begin{table}[H]
    \centering
    \small
    \begin{tabularx}{\textwidth}{lXXXX}
        \hline
        \multirow{3}{*}{Hyperparameter} & \multicolumn{4}{c}{Value} \\
        \cline{2-5}
        & \multicolumn{2}{c}{VHB 4910} & \multicolumn{2}{c}{VHB 4905} \\
        \cline{2-5}
        & free energy & dissipation potential & free energy & dissipation potential\\
        \hline
        Topology  &[3,2,1] & [3,2,1]& [4,2,1] & [3,2,1]\\
        Order of splines & 3& 3& 3 &3\\
        Grid intervals & 1&1 &1 &1\\
        \hline
        Optimizer & \multicolumn{2}{c}{AMSGrad} &\multicolumn{2}{c}{AMSGrad} \\
        Scheduler & \multicolumn{2}{c}{Cyclic LR} &\multicolumn{2}{c}{Cyclic LR} \\ 
        Base Learning rate & \multicolumn{2}{c}{$5 \cdot 10^{-3}$} & \multicolumn{2}{c}{$1 \cdot 10^{-4}$}\\
        Max Learning rate & \multicolumn{2}{c}{$5 \cdot 10^{-2}$} & \multicolumn{2}{c}{$1 \cdot 10^{-3}$}\\
        Clip gradient norm & \multicolumn{2}{c}{$0.1$} & \multicolumn{2}{c}{$0.1$}\\
        L1 regularization magnitude & \multicolumn{2}{c}{$1 \cdot 10^{-5}$}&\multicolumn{2}{c}{$1 \cdot 10^{-4}$} \\
        \hline
    \end{tabularx}
    \caption{Hyperparameters for the iCKAN models used in the numerical examples in Section \ref{sec:results} for VHB 4910 and VHB 4905.}
    \label{tab:hyperparameters}
\end{table}

\textbf{Detailed results.} \autoref{tab:results_VHB} shows the detailed results of iCKANs performance on VHB 4910 and VHB 4905 polymer. The stress loss before symbolification is denoted as $\mathcal{L}$ and the stress loss after symbolification is denoted as $\hat{\mathcal{L}}$. \autoref{fig:duo_iCKAN_VHB4910_4905_explicit_loss} shows the loss convergence of training iCKANs from Section \ref{sec:results} for the material model discovery of VHB 4910 and VHB 4905.

\begin{table}[H]
\caption{Detailed results of training and testing of iCKANs on VHB 4910 and VHB 4905 dataset.}
\label{tab:results_VHB}
\centering\small
\begin{tabular}{lccccc}
    \hline  & 
    $\mathcal{L}_{\text{train}}$
    & $\mathcal{L}_{\text{test}}$& $\hat{\mathcal{L}}_\text{train}$ & $\hat{\mathcal{L}}_{\text{test}}$  \\
    \hline
    VHB 4910 &  $7.6 \cdot 10^{-5}$   & $1.3 \cdot 10^{-3}$ & $1.9 \cdot 10^{-4}$ &$1.9 \cdot 10^{-3}$ \\
    VHB 4905 for $\theta \in [0,80]\, ^\circ \mathrm{C}$ &  $2.2 \cdot 10^{-4}$ & $6.5 \cdot 10^{-4}$ & $4.1 \cdot 10^{-4}$  & $6.9 \cdot 10^{-4}$\\
    VHB 4905 at $\theta = 20\, ^\circ \mathrm{C}$ &  $9.98 \cdot 10^{-5}$ & $1.5 \cdot 10^{-4}$ & -- & -- \\\hline
    
\end{tabular}
\end{table}

\begin{figure}[H]
    \centering
    \small
    \begin{subfigure}[t]{0.49\textwidth}
        \centering
        \includegraphics[width=\textwidth]{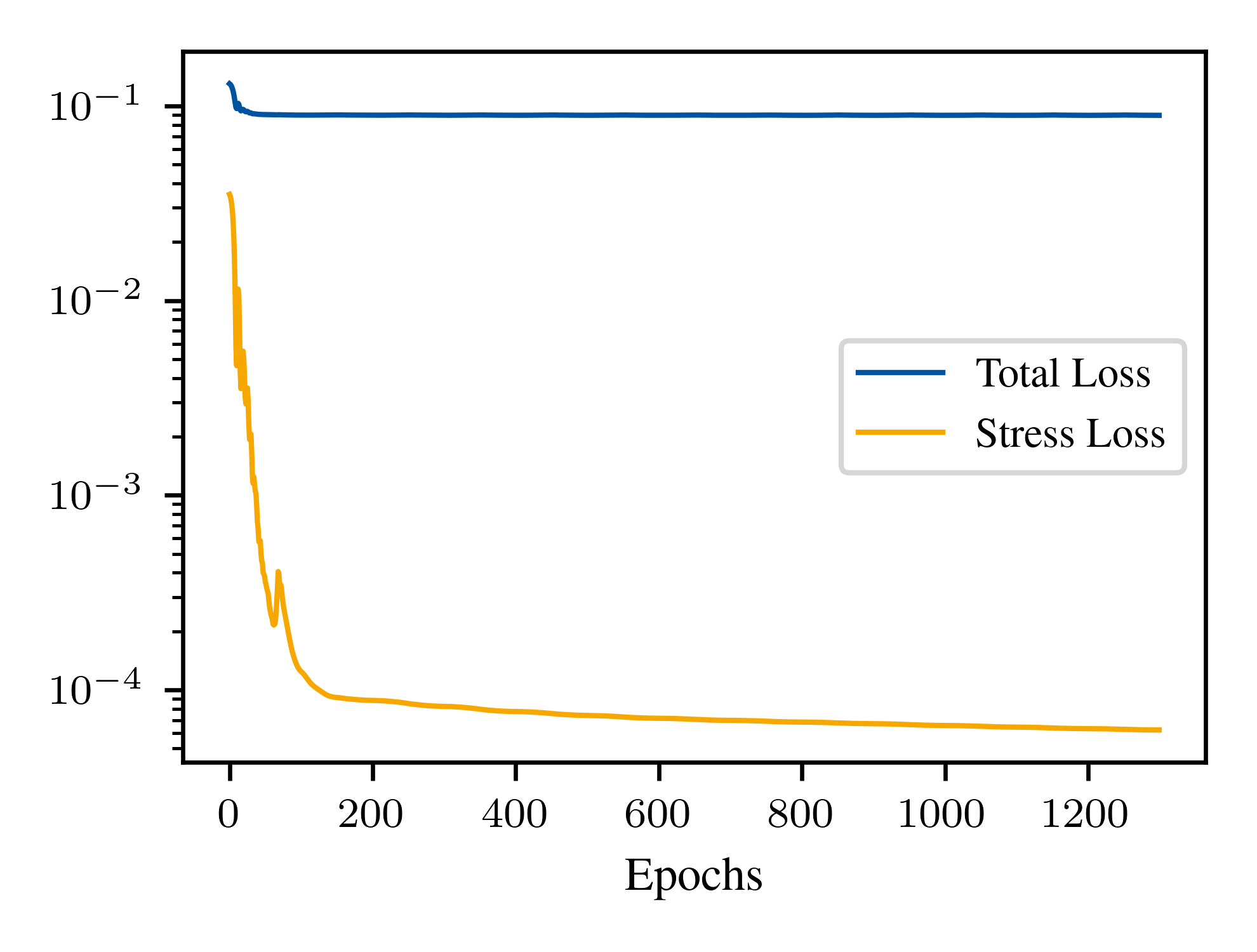}
        \put(-225,80){\rotatebox{90}{Loss}}
        \caption{VHB 4910}
    \end{subfigure}
    \hfill
    \begin{subfigure}[t]{0.49\textwidth}
        \centering
        \includegraphics[width=\textwidth]{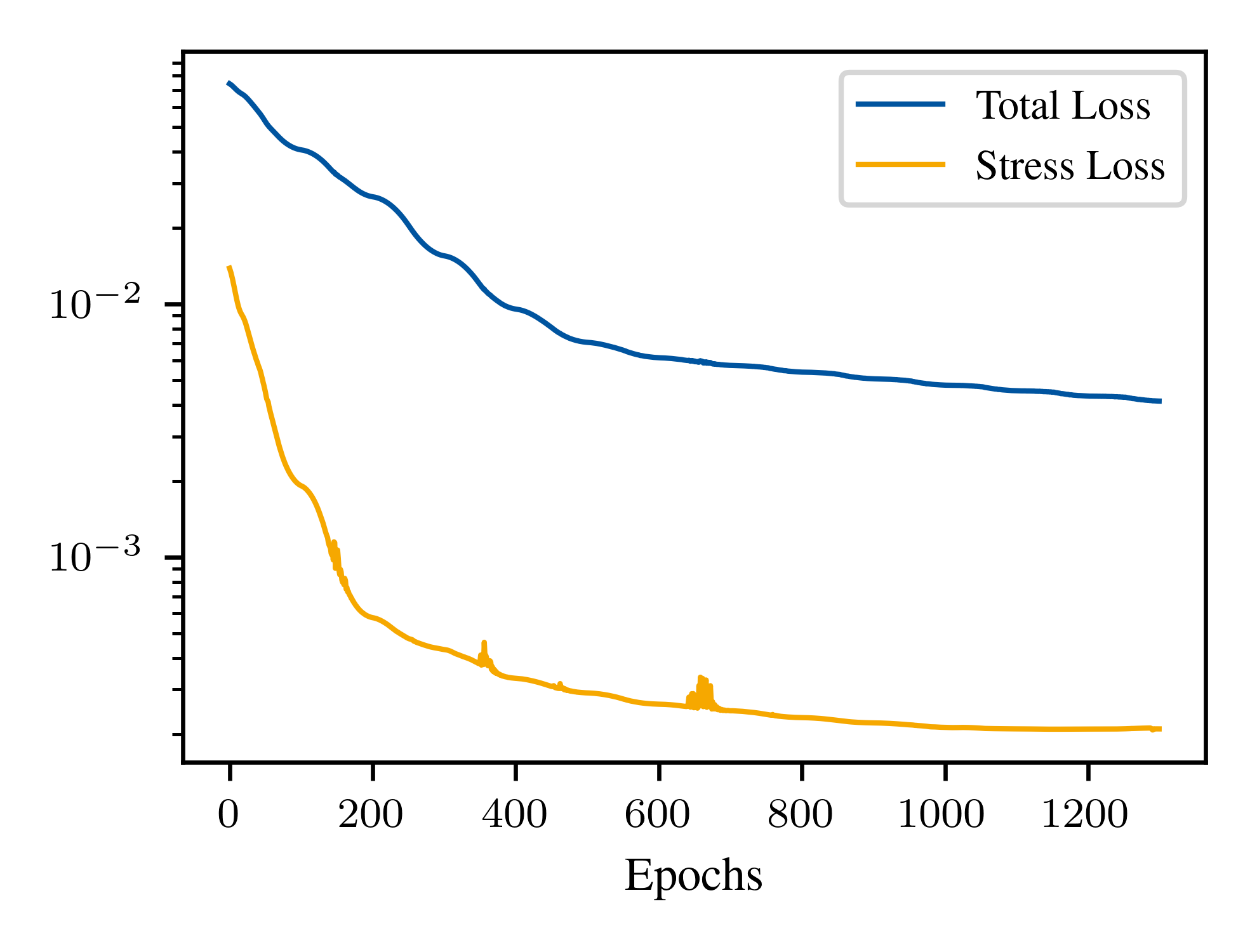}
        \put(-225,80){\rotatebox{90}{Loss}}
        \caption{VHB 4905}
    \end{subfigure}
    \caption{Losses during training of the iCKAN for the experimental data of (a) VHB 4910 and (b) VHB4905 polymer. The losses are plotted on a logarithmic scale. For training, 1300 epochs were used.}
    \label{fig:duo_iCKAN_VHB4910_4905_explicit_loss}
\end{figure}

\color{black}

%% file: standalone_iCKAN_implicit.tex
\begin{tikzpicture}[
    font=\rmfamily\small,
    >=LaTeX,
    cell/.style={
        rectangle, 
        rounded corners=5mm, 
        draw,
        very thick,
        },
    operator/.style={
        circle,
        draw,
        inner sep=-0.5pt,
        minimum height =.2cm,
        },
    function/.style={
        ellipse,
        draw,
        inner sep=1pt
        },
    ct/.style={
        circle,
        draw,
        line width = .75pt,
        minimum width=1cm,
        inner sep=1pt,
        },
    ct2/.style={
        circle,
        draw,
        line width = 1.5pt,
        minimum width=1.2cm,
        inner sep=1pt,
        },
    gt/.style={
        rectangle,
        draw,
        minimum width=4mm,
        minimum height=3mm,
        inner sep=1pt
        },
    mylabel/.style={
        font=\small\rmfamily
        },
    ArrowC1/.style={
        rounded corners=.2cm,
        thick,
        },
    ArrowC2/.style={
        rounded corners=.5cm,
        thick,
        },
    ArrowC3/.style={
        rounded corners=.3cm,
        thick,
        },
    ]

    \node [cell, minimum height = 4cm, minimum width=9cm,fill=RWTHBlau!5] at (0,0){} ;

    \node [ct, fill=RWTHBlau!25] (ibox1) at (-3,-1.20) {$\bar{\vec{C}}_e^{\text{trial}}$};
    \node [inner sep=0pt, outer sep=0pt, draw=none] (ibox2) at (-1.5,-1.20) {\includegraphics[width=1.2cm]{figures_KAN.png}};
    \node [ct, fill=RWTHBlau!25] (ibox3) at (0.25,-1.20) {$\bar{\vec{\Sigma}}^{\text{trial}}$};
    \node [inner sep=0pt, outer sep=0pt, draw=none] (ibox4) at (1.75,-1.20) {\includegraphics[width=1.2cm]{figures_KAN2.png}};
    \node [ct, fill=RWTHBlau!25] (ibox5) at (3.5,-1.20) {$\bar{\vec{D}}_i^{\text{trial}}$};

    \node [ct2, fill=RWTHGruen!25] (LTC) at (-0.5,0.75) {LTC};
    \node [ct, fill=RWTHBlau!25] (add1) at (1.25,0.75) {$\bar{\vec{C}}_e$};
    \node [inner sep=0pt, outer sep=0pt, draw=none] (psi2) at (2.75,0.75) {\includegraphics[width=1.2cm]{figures_KAN.png}};

    \node[ct, label={[mylabel]State}, fill=RWTHOrange!25, align=center] (h) at (-5.5,-1.20) {$\vec{C}_{t-1}$,\\$\vec{U}_{i,t-1}$};
    \node[ct, label={[mylabel]left:Input}, fill=RWTHOrange!25] (x) at (-4.25,-3) {$\vec{F}_{t},\Delta t$};

    \node[ct, label={[mylabel]State}, fill=RWTHBordeaux!25,align=center] (h1) at (5.5,1.75) {$\vec{C}_{t}$,\\$\vec{U}_{i,t}$};
    \node[ct, label={[mylabel]left:Output}, fill=RWTHBordeaux!25] (x2) at (2.75,3) {$\vec{P}_{t}$};

    \node at (-0.8,-0.53) {\normalsize $\psi^{\text{trial}}$};
    \node at (2.5,-0.53) {\normalsize $\omega^{\text{trial}}$};
    \node at (3.25,1.4) {\normalsize $\psi_t$};

    \draw [->, ArrowC1] (h) -- (ibox1);
    \draw [-, ArrowC1] (ibox1) -- (ibox2);
    \draw [->, ArrowC1] (ibox2) -- (ibox3);
    \draw [-, ArrowC1] (ibox3) -- (ibox4);
    \draw [->, ArrowC2] (x) |- (ibox1);
    \draw [->, ArrowC2] (x) |- (LTC);
    \draw [->, ArrowC1] (LTC) -- (add1);
    \draw [-, ArrowC1] (add1) -- (psi2);
    \draw [->, ArrowC1] (ibox4) -- (ibox5);
    \draw [->, ArrowC1] (LTC) |- (h1);
    \draw [->, ArrowC1] (psi2) -- (x2);
    \draw [-, ArrowC1] (ibox5) |- (1,-0.125);
    \draw [->, ArrowC1] (1,-0.125) -| (LTC);

\end{tikzpicture}

%% file: Appendix_B.tex
\section{Declarations}
\label{app:B}

\subsection{Acknowledgements}

This work was supported by the Emmy Noether Grant 533187597 by the Deutsche Forschungsgemeinschaft to Kevin Linka. Christian Cyron gratefully acknowledges support of the European Research Council (ERC) under the European Union’s Horizon Europe research and innovation programme (grant agreement No 101167207 / MechVivo).

\subsection{Conflict of interest}

The authors of this work certify that they have no affiliations with or involvement in any organization or entity with any financial interest (such as honoraria; participation in speakers' bureaus; membership, employment, consultancies, stock ownership, or other equity interest; and expert testimony or patent-licensing arrangements), or non-financial interest (such as personal or professional relationships, affiliations, knowledge or beliefs) in the subject matter or materials discussed in this manuscript.

\subsection{Data availability}

The developed source code of the iCKAN architecture and accompanying example data sets are available at \\\mbox{https://github.com/AME-COMED/iCKAN}.

\subsection{Contributions by the authors}
\textbf{Chenyi Ji:} Writing – original draft, Writing – review \& editing, Visualization, Validation, Software, Methodology, Investigation, Formal analysis, Data curation, Conceptualization.
\textbf{Kian P. Abdolazizi:} Writing – review \& editing, Supervision, Software, Methodology, Data curation, Conceptualization.
\textbf{Hagen Holthusen:} Writing – review \& editing, Supervision, Methodology, Conceptualization.
\textbf{Chrstian J. Cyron:} Writing – review \& editing, Methodology, Conceptualization.
\textbf{Kevin Linka:} Writing – review \& editing, Supervision, Methodology, Funding acquisition, Conceptualization.

\subsection{Declaration of generative AI and AI-assisted technologies in the manuscript preparation process}
During the preparation of this work the authors used OpenAI's ChatGPT in order to refine the language. After using this tool/service, the authors reviewed and edited the content as needed and take full responsibility for the content of the published article.

%% file: bibliography.bib
@article{abdolazizi2024viscoelastic,
  title={Viscoelastic constitutive artificial neural networks (vCANNs)--A framework for data-driven anisotropic nonlinear finite viscoelasticity},
  author={Abdolazizi, Kian P and Linka, Kevin and Cyron, Christian J},
  journal={Journal of computational physics},
  volume={499},
  pages={112704},
  year={2024},
  publisher={Elsevier}
}

@article{abdolazizi2025constitutive,
archivePrefix = {arXiv},
arxivId = {2502.05682},
author = {Abdolazizi, Kian P. and Aydin, Roland C. and Cyron, Christian J. and Linka, Kevin},
doi = {10.1016/j.jmps.2025.106212},
eprint = {2502.05682},
issn = {00225096},
journal = {Journal of the Mechanics and Physics of Solids},
month = {oct},
pages = {106212},
title = {{Constitutive Kolmogorov–Arnold Networks (CKANs): Combining accuracy and interpretability in data-driven material modeling}},
url = {http://arxiv.org/abs/2502.05682 https://linkinghub.elsevier.com/retrieve/pii/S0022509625001887},
volume = {203},
year = {2025}
}

@article{abdolazizi2026thermodynamically,
  title={Thermodynamically consistent viscoelastic constitutive artificial neural networks: Automating the pipeline from experimental data to finite element simulations},
  author={Abdolazizi, Kian P and Aydin, Roland C and Cyron, Christian J and Linka, Kevin},
  journal={Computer Methods in Applied Mechanics and Engineering},
  volume={460},
  pages={119080},
  year={2026},
  publisher={Elsevier}
}

@article{abdusalamov2023automatic,
  title={Automatic generation of interpretable hyperelastic material models by symbolic regression},
  author={Abdusalamov, Rasul and Hillg{\"a}rtner, Markus and Itskov, Mikhail},
  journal={International Journal for Numerical Methods in Engineering},
  volume={124},
  number={9},
  pages={2093--2104},
  year={2023},
  publisher={Wiley Online Library}
}

@article{abueidda2025deepokan,
  title={Deepokan: Deep operator network based on kolmogorov arnold networks for mechanics problems},
  author={Abueidda, Diab W and Pantidis, Panos and Mobasher, Mostafa E},
  journal={Computer Methods in Applied Mechanics and Engineering},
  volume={436},
  pages={117699},
  year={2025},
  publisher={Elsevier}
}

@inproceedings{amos2017input,
  title={Input convex neural networks},
  author={Amos, Brandon and Xu, Lei and Kolter, J Zico},
  booktitle={International conference on machine learning},
  pages={146--155},
  year={2017},
  organization={PMLR}
}

@inproceedings{as2023mechanics,
  title={A mechanics-informed neural network framework for data-driven nonlinear viscoelasticity},
  author={As' ad, Faisal and Farhat, Charbel},
  booktitle={AIAA SCITECH 2023 Forum},
  pages={0949},
  year={2023}
}

@article{ball1976convexity,
  title={Convexity conditions and existence theorems in nonlinear elasticity},
  author={Ball, John M},
  journal={Archive for rational mechanics and Analysis},
  volume={63},
  number={4},
  pages={337--403},
  year={1976},
  publisher={Springer}
}

@article{benoit1924note,
  title={Note sur une m{\'e}thode de r{\'e}solution des {\'e}quations normales provenant de l’application de la m{\'e}thode des moindres carr{\'e}s {\`a} un syst{\`e}me d’{\'e}quations lin{\'e}aires en nombre inf{\'e}rieur {\`a} celui des inconnues (Proc{\'e}d{\'e} du Commandant Cholesky)},
  author={Benoit, Commandant},
  journal={Bulletin g{\'e}od{\'e}sique},
  volume={2},
  number={1},
  pages={67--77},
  year={1924}
}

@article{boes2024accounting,
title = {Accounting for plasticity: An extension of inelastic constitutive artificial neural networks},
journal = {European Journal of Mechanics - A/Solids},
volume = {117},
pages = {105998},
year = {2026},
issn = {0997-7538},
author = {Birte Boes and Jaan-Willem Simon and Hagen Holthusen},
}

@article{bomarito2021development,
  title={Development of interpretable, data-driven plasticity models with symbolic regression},
  author={Bomarito, GF and Townsend, TS and Stewart, KM and Esham, KV and Emery, JM and Hochhalter, JD},
  journal={Computers \& Structures},
  volume={252},
  pages={106557},
  year={2021},
  publisher={Elsevier}
}

@article{dammass2025neural,
  title={Neural networks meet phase-field: A hybrid fracture model},
  author={Damma{\ss}, Franz and Kalina, Karl A and K{\"a}stner, Markus},
  journal={Computer Methods in Applied Mechanics and Engineering},
  volume={440},
  pages={117937},
  year={2025},
  publisher={Elsevier}
}

@article{deschatre2025input,
  title={Input Convex Kolmogorov Arnold Networks},
  author={Deschatre, Thomas and Warin, Xavier},
  journal={arXiv preprint arXiv:2505.21208},
  year={2025}
}

@article{fernandez2021anisotropic,
  title={Anisotropic hyperelastic constitutive models for finite deformations combining material theory and data-driven approaches with application to cubic lattice metamaterials},
  author={Fern{\'a}ndez, Mauricio and Jamshidian, Mostafa and B{\"o}hlke, Thomas and Kersting, Kristian and Weeger, Oliver},
  journal={Computational Mechanics},
  volume={67},
  number={2},
  pages={653--677},
  year={2021},
  publisher={Springer}
}

@article{fernandez2022material,
  title={Material modeling for parametric, anisotropic finite strain hyperelasticity based on machine learning with application in optimization of metamaterials},
  author={Fern{\'a}ndez, Mauricio and Fritzen, Felix and Weeger, Oliver},
  journal={International Journal for Numerical Methods in Engineering},
  volume={123},
  number={2},
  pages={577--609},
  year={2022},
  publisher={Wiley Online Library}
}

@phdthesis{flaschel2023automated,
  title={Automated discovery of material models in continuum solid mechanics},
  author={Flaschel, Moritz},
  year={2023},
  school={ETH Zurich}
}

@article{flaschel2021unsupervised,
  title={Unsupervised discovery of interpretable hyperelastic constitutive laws},
  author={Flaschel, Moritz and Kumar, Siddhant and De Lorenzis, Laura},
  journal={Computer Methods in Applied Mechanics and Engineering},
  volume={381},
  pages={113852},
  year={2021},
  publisher={Elsevier}
}

@article{flaschel2023automated2,
  title={Automated discovery of generalized standard material models with EUCLID},
  author={Flaschel, Moritz and Kumar, Siddhant and De Lorenzis, Laura},
  journal={Computer Methods in Applied Mechanics and Engineering},
  volume={405},
  pages={115867},
  year={2023},
  publisher={Elsevier}
}

@article{flaschel2025convex,
  title={Convex neural networks learn generalized standard material models},
  author={Flaschel, Moritz and Steinmann, Paul and De Lorenzis, Laura and Kuhl, Ellen},
  journal={Journal of the Mechanics and Physics of Solids},
  volume={200},
  pages={106103},
  year={2025},
  publisher={Elsevier}
}

@article{flory1961thermodynamic,
  title={Thermodynamic relations for high elastic materials},
  author={Flory, PJ128117},
  journal={Transactions of the Faraday Society},
  volume={57},
  pages={829--838},
  year={1961},
  publisher={Royal Society of Chemistry}
}

@article{fuhg2024review,
  title={A review on data-driven constitutive laws for solids},
  author={Fuhg, Jan Niklas and Padmanabha, Govinda Anantha and Bouklas, Nikolaos and Bahmani, Bahador and Sun, WaiChing and Vlassis, Nikolaos N and Flaschel, Moritz and Carrara, Pietro and De Lorenzis, Laura},
  journal={arXiv preprint arXiv:2405.03658},
  year={2024}
}

@article{germain1983continuum,
  title={Continuum thermodynamics},
  author={Germain, Paul and Nguyen, Quoc Son and Suquet, Pierre},
  journal={Journal of applied mechanics},
  volume={50},
  pages={1010--1020},
  year={1983}
}

@article{guo2025history,
  title={History-Aware Neural Operator: Robust Data-Driven Constitutive Modeling of Path-Dependent Materials},
  author={Guo, Binyao and Lin, Zihan and He, QiZhi},
  journal={arXiv preprint arXiv:2506.10352},
  year={2025}
}

@article{halphen1975materiaux,
  title={Sur les mat{\'e}riaux standard g{\'e}n{\'e}ralis{\'e}s},
  author={Halphen, Bernard and Nguyen, Quoc Son},
  journal={Journal de m{\'e}canique},
  volume={14},
  number={1},
  pages={39--63},
  year={1975}
}

@article{hartmann2003polyconvexity,
  title={Polyconvexity of generalized polynomial-type hyperelastic strain energy functions for near-incompressibility},
  author={Hartmann, Stefan and Neff, Patrizio},
  journal={International journal of solids and structures},
  volume={40},
  number={11},
  pages={2767--2791},
  year={2003},
  publisher={Elsevier}
}

@inproceedings{hasani2021liquid,
  title={Liquid time-constant networks},
  author={Hasani, Ramin and Lechner, Mathias and Amini, Alexander and Rus, Daniela and Grosu, Radu},
  booktitle={Proceedings of the AAAI Conference on Artificial Intelligence},
  volume={35},
  pages={7657--7666},
  year={2021}
}

@article{hospedales2021meta,
  title={Meta-learning in neural networks: A survey},
  author={Hospedales, Timothy and Antoniou, Antreas and Micaelli, Paul and Storkey, Amos},
  journal={IEEE transactions on pattern analysis and machine intelligence},
  volume={44},
  number={9},
  pages={5149--5169},
  year={2021},
  publisher={IEEE}
}

@article{hossain2012experimental,
  title={Experimental study and numerical modelling of VHB 4910 polymer},
  author={Hossain, Mokarram and Vu, Duc Khoi and Steinmann, Paul},
  journal={Computational Materials Science},
  volume={59},
  pages={65--74},
  year={2012},
  publisher={Elsevier}
}

@article{holthusen2023inelastic,
title = {Inelastic material formulations based on a co-rotated intermediate configuration—Application to bioengineered tissues},
journal = {Journal of the Mechanics and Physics of Solids},
volume = {172},
pages = {105174},
year = {2023},
issn = {0022-5096},
author = {Hagen Holthusen and Christiane Rothkranz and Lukas Lamm and Tim Brepols and Stefanie Reese},
}

@article{holthusen2024theory,
  title={Theory and implementation of inelastic constitutive artificial neural networks},
  author={Holthusen, Hagen and Lamm, Lukas and Brepols, Tim and Reese, Stefanie and Kuhl, Ellen},
  journal={Computer Methods in Applied Mechanics and Engineering},
  volume={428},
  pages={117063},
  year={2024},
  publisher={Elsevier}
}

@article{holthusen2024PAMM,
author = {Holthusen, Hagen and Lamm, Lukas and Brepols, Tim and Reese, Stefanie and Kuhl, Ellen},
title = {Polyconvex inelastic constitutive artificial neural networks},
journal = {PAMM},
volume = {24},
number = {3},
pages = {e202400032},
year = {2024}
}

@article{holthusen2025automated,
  title={Automated model discovery for tensional homeostasis: Constitutive machine learning in growth and remodeling},
  author={Holthusen, Hagen and Brepols, Tim and Linka, Kevin and Kuhl, Ellen},
  journal={Computers in biology and medicine},
  volume={186},
  pages={109691},
  year={2025},
  publisher={Elsevier}
}

@article{holthusen2025generalized,
title = {A generalized dual potential for inelastic Constitutive Artificial Neural Networks: A JAX implementation at finite strains},
journal = {Journal of the Mechanics and Physics of Solids},
volume = {206},
pages = {106337},
year = {2026},
issn = {0022-5096},
author = {Hagen Holthusen and Kevin Linka and Ellen Kuhl and Tim Brepols},
}

@article{holthusen2026complement,
  title={A complement to neural networks for anisotropic inelasticity at finite strains},
  author={Holthusen, Hagen and Kuhl, Ellen},
  journal={Computer Methods in Applied Mechanics and Engineering},
  volume={450},
  pages={118612},
  year={2026},
  publisher={Elsevier}
}

@book{holzapfel2000nonlinear,
  title     = {Nonlinear Solid Mechanics: A Continuum Approach for Engineering Science},
  author    = {Holzapfel, Gerhard A.},
  year      = {2000},
  publisher = {John Wiley \& Sons},
  address   = {Chichester},
  isbn      = {0-471-82319-8}
}

@article{hornik1989multilayer,
  title={Multilayer feedforward networks are universal approximators},
  author={Hornik, Kurt and Stinchcombe, Maxwell and White, Halbert},
  journal={Neural networks},
  volume={2},
  number={5},
  pages={359--366},
  year={1989},
  publisher={Elsevier}
}

@article{huang2022variational,
  title={Variational Onsager Neural Networks (VONNs): A thermodynamics-based variational learning strategy for non-equilibrium PDEs},
  author={Huang, Shenglin and He, Zequn and Reina, Celia},
  journal={Journal of the Mechanics and Physics of Solids},
  volume={163},
  pages={104856},
  year={2022},
  publisher={Elsevier}
}

@article{hudobivnik2016closed,
  title={Closed-form representation of matrix functions in the formulation of nonlinear material models},
  author={Hudobivnik, Bla{\v{z}} and Korelc, Jo{\v{z}}e},
  journal={Finite Elements in Analysis and Design},
  volume={111},
  pages={19--32},
  year={2016},
  publisher={Elsevier}
}

@article{jadoon2025automated,
  title={Automated model discovery of finite strain elastoplasticity from uniaxial experiments},
  author={Jadoon, Asghar Arshad and Meyer, Knut Andreas and Fuhg, Jan Niklas},
  journal={Computer Methods in Applied Mechanics and Engineering},
  volume={435},
  pages={117653},
  year={2025},
  publisher={Elsevier}
}

@article{ji2024comprehensive,
  title={A comprehensive survey on kolmogorov arnold networks (kan)},
  author={Ji, Tianrui and Hou, Yuntian and Zhang, Di},
  journal={arXiv preprint arXiv:2407.11075},
  year={2024}
}

@article{jones2022neural,
  title={A neural ordinary differential equation framework for modeling inelastic stress response via internal state variables},
  author={Jones, Reese E and Frankel, Ari L and Johnson, KL},
  journal={Journal of Machine Learning for Modeling and Computing},
  volume={3},
  number={3},
  year={2022},
  publisher={Begel House Inc.}
}

@article{kabliman2021application,
  title={Application of symbolic regression for constitutive modeling of plastic deformation},
  author={Kabliman, Evgeniya and Kolody, Ana Helena and Kronsteiner, Johannes and Kommenda, Michael and Kronberger, Gabriel},
  journal={Applications in Engineering Science},
  volume={6},
  pages={100052},
  year={2021},
  publisher={Elsevier}
}

@article{kalina2025neural,
  title={Neural networks meet anisotropic hyperelasticity: A framework based on generalized structure tensors and isotropic tensor functions},
  author={Kalina, Karl A and Brummund, J{\"o}rg and Sun, WaiChing and K{\"a}stner, Markus},
  journal={Computer Methods in Applied Mechanics and Engineering},
  volume={437},
  pages={117725},
  year={2025},
  publisher={Elsevier}
}

@article{kalina2026physics,
  title={A physics-augmented neural network framework for finite strain incompressible viscoelasticity},
  author={Kalina, Karl A and Brummund, J{\"o}rg and K{\"a}stner, Markus},
  journal={Computer Methods in Applied Mechanics and Engineering},
  volume={455},
  pages={118892},
  year={2026},
  publisher={Elsevier}
}

@article{kirchdoerfer2016data,
  title={Data-driven computational mechanics},
  author={Kirchdoerfer, Trenton and Ortiz, Michael},
  journal={Computer Methods in Applied Mechanics and Engineering},
  volume={304},
  pages={81--101},
  year={2016},
  publisher={Elsevier}
}

@article{kiyani2025optimizing,
  title={Optimizing the optimizer for physics-informed neural networks and Kolmogorov-Arnold networks},
  author={Kiyani, Elham and Shukla, Khemraj and Urb{\'a}n, Jorge F and Darbon, J{\'e}r{\^o}me and Karniadakis, George Em},
  journal={Computer Methods in Applied Mechanics and Engineering},
  volume={446},
  pages={118308},
  year={2025},
  publisher={Elsevier}
}

@book{kolmogorov1961representation,
  title={On the representation of continuous functions of several variables by superpositions of continuous functions of a smaller number of variables},
  author={Kolmogorov, Andre{\u\i} Nikolaevich},
  year={1961},
  publisher={American Mathematical Society}
}

@article{li2025long,
  title={A long short-term memory-based constitutive modeling framework for capturing strain path dependence in plastic deformation},
  author={Li, Jin-Zhao and Guan, Zhi-Ping and Chen, Jiong-Rui and Jin, Hui-Chao},
  journal={Mechanics of Materials},
  volume={205},
  pages={105325},
  year={2025},
  publisher={Elsevier}
}

@article{liao2020thermo,
  title={On thermo-viscoelastic experimental characterization and numerical modelling of VHB polymer},
  author={Liao, Zisheng and Hossain, Mokarram and Yao, Xiaohu and Mehnert, Markus and Steinmann, Paul},
  journal={International Journal of Non-Linear Mechanics},
  volume={118},
  pages={103263},
  year={2020},
  publisher={Elsevier}
}

@article{linden2023neural,
  title={Neural networks meet hyperelasticity: A guide to enforcing physics},
  author={Linden, Lennart and Klein, Dominik K and Kalina, Karl A and Brummund, J{\"o}rg and Weeger, Oliver and K{\"a}stner, Markus},
  journal={Journal of the Mechanics and Physics of Solids},
  volume={179},
  pages={105363},
  year={2023},
  publisher={Elsevier}
}

@article{linka2021constitutive,
  title={Constitutive artificial neural networks: A fast and general approach to predictive data-driven constitutive modeling by deep learning},
  author={Linka, Kevin and Hillg{\"a}rtner, Markus and Abdolazizi, Kian P and Aydin, Roland C and Itskov, Mikhail and Cyron, Christian J},
  journal={Journal of Computational Physics},
  volume={429},
  pages={110010},
  year={2021},
  publisher={Elsevier}
}

@article{linka2023new,
  title={A new family of constitutive artificial neural networks towards automated model discovery},
  author={Linka, Kevin and Kuhl, Ellen},
  journal={Computer Methods in Applied Mechanics and Engineering},
  volume={403},
  pages={115731},
  year={2023},
  publisher={Elsevier}
}

@article{liu2024kan,
  title={Kan: Kolmogorov-arnold networks},
  author={Liu, Ziming and Wang, Yixuan and Vaidya, Sachin and Ruehle, Fabian and Halverson, James and Solja{\v{c}}i{\'c}, Marin and Hou, Thomas Y and Tegmark, Max},
  journal={arXiv preprint arXiv:2404.19756},
  year={2024}
}

@article{masi2021thermodynamics,
  title={Thermodynamics-based artificial neural networks for constitutive modeling},
  author={Masi, Filippo and Stefanou, Ioannis and Vannucci, Paolo and Maffi-Berthier, Victor},
  journal={Journal of the Mechanics and Physics of Solids},
  volume={147},
  pages={104277},
  year={2021},
  publisher={Elsevier}
}

@article{maurizi2022predicting,
  title={Predicting stress, strain and deformation fields in materials and structures with graph neural networks},
  author={Maurizi, Marco and Gao, Chao and Berto, Filippo},
  journal={Scientific reports},
  volume={12},
  number={1},
  pages={21834},
  year={2022},
  publisher={Nature Publishing Group UK London}
}

@article{narouie2026unsupervised,
  title={Unsupervised Constitutive Model Discovery from Sparse and Noisy Data},
  author={Narouie, Vahab Knauf and Urrea-Quintero, Jorge-Humberto and Cirak, Fehmi and Wessels, Henning},
  journal={Computer Methods in Applied Mechanics and Engineering},
  volume={452},
  pages={118722},
  year={2026},
  publisher={Elsevier}
}

@article{polo2024monokan,
  title={MonoKAN: Certified Monotonic Kolmogorov-Arnold Network},
  author={Polo-Molina, Alejandro and Alfaya, David and Portela, Jose},
  journal={arXiv preprint arXiv:2409.11078},
  year={2024}
}

@article{raissi2019physics,
  title={Physics-informed neural networks: A deep learning framework for solving forward and inverse problems involving nonlinear partial differential equations},
  author={Raissi, Maziar and Perdikaris, Paris and Karniadakis, George E},
  journal={Journal of Computational physics},
  volume={378},
  pages={686--707},
  year={2019},
  publisher={Elsevier}
}

@article{reddi2019convergence,
  title={On the convergence of adam and beyond},
  author={Reddi, Sashank J and Kale, Satyen and Kumar, Sanjiv},
  journal={arXiv preprint arXiv:1904.09237},
  year={2019}
}

@article{rosenkranz2024viscoelasticty,
  title={Viscoelasticty with physics-augmented neural networks: model formulation and training methods without prescribed internal variables},
  author={Rosenkranz, Max and Kalina, Karl A and Brummund, J{\"o}rg and Sun, WaiChing and K{\"a}stner, Markus},
  journal={Computational Mechanics},
  volume={74},
  number={6},
  pages={1279--1301},
  year={2024},
  publisher={Springer}
}

@article{samadi2024smooth,
  title={Smooth Kolmogorov Arnold networks enabling structural knowledge representation},
  author={Samadi, Moein E and M{\"u}ller, Younes and Schuppert, Andreas},
  journal={arXiv preprint arXiv:2405.11318},
  year={2024}
}

@inproceedings{smith2017cyclical,
  title={Cyclical learning rates for training neural networks},
  author={Smith, Leslie N},
  booktitle={2017 IEEE winter conference on applications of computer vision (WACV)},
  pages={464--472},
  year={2017},
  organization={IEEE}
}

@article{tac2022data,
  title={Data-driven tissue mechanics with polyconvex neural ordinary differential equations},
  author={Tac, Vahidullah and Costabal, Francisco Sahli and Tepole, Adrian B},
  journal={Computer Methods in Applied Mechanics and Engineering},
  volume={398},
  pages={115248},
  year={2022},
  publisher={Elsevier}
}

@article{tacc2023data,
  title={Data-driven anisotropic finite viscoelasticity using neural ordinary differential equations},
  author={Ta{\c{c}}, Vahidullah and Rausch, Manuel K and Costabal, Francisco Sahli and Tepole, Adrian Buganza},
  journal={Computer methods in applied mechanics and engineering},
  volume={411},
  pages={116046},
  year={2023},
  publisher={Elsevier}
}

@article{tacke2025constitutive,
  title={Constitutive scientific generative agent (CSGA): Leveraging large language models for automated constitutive model discovery},
  author={Tacke, Marius and Busch, Matthias and Bali, Kartik and Abdolazizi, Kian and Linka, Kevin and Cyron, Christian and Aydin, Roland},
  journal={Machine Learning for Computational Science and Engineering},
  volume={1},
  number={1},
  pages={23},
  year={2025},
  publisher={Springer}
}

@article{thakolkaran2025can,
  title={Can kan cans? input-convex kolmogorov-arnold networks (kans) as hyperelastic constitutive artificial neural networks (cans)},
  author={Thakolkaran, Prakash and Guo, Yaqi and Saini, Shivam and Peirlinck, Mathias and Alheit, Benjamin and Kumar, Siddhant},
  journal={Computer Methods in Applied Mechanics and Engineering},
  volume={443},
  pages={118089},
  year={2025},
  publisher={Elsevier}
}

@article{versino2017data,
  title={Data driven modeling of plastic deformation},
  author={Versino, Daniele and Tonda, Alberto and Bronkhorst, Curt A},
  journal={Computer Methods in Applied Mechanics and Engineering},
  volume={318},
  pages={981--1004},
  year={2017},
  publisher={Elsevier}
}

@article{wang2025kolmogorov,
  title={Kolmogorov--Arnold-Informed neural network: A physics-informed deep learning framework for solving forward and inverse problems based on Kolmogorov--Arnold Networks},
  author={Wang, Yizheng and Sun, Jia and Bai, Jinshuai and Anitescu, Cosmin and Eshaghi, Mohammad Sadegh and Zhuang, Xiaoying and Rabczuk, Timon and Liu, Yinghua},
  journal={Computer Methods in Applied Mechanics and Engineering},
  volume={433},
  pages={117518},
  year={2025},
  publisher={Elsevier}
}

@article{wiesheier2026data,
  title={Data-adaptive spline-based viscoelasticity for soft solids},
  author={Wiesheier, Simon and Moreno-Mateos, Miguel Angel and Steinmann, Paul},
  journal={Computer Methods in Applied Mechanics and Engineering},
  volume={451},
  pages={118705},
  year={2026},
  publisher={Elsevier}
}

@article{you2022physics,
  title={A physics-guided neural operator learning approach to model biological tissues from digital image correlation measurements},
  author={You, Huaiqian and Zhang, Quinn and Ross, Colton J and Lee, Chung-Hao and Hsu, Ming-Chen and Yu, Yue},
  journal={Journal of Biomechanical Engineering},
  volume={144},
  number={12},
  pages={121012},
  year={2022},
  publisher={American Society of Mechanical Engineers}
}
